\documentclass[10pt]{article}

\usepackage{amsmath,amssymb}
\usepackage{mathtools}
\usepackage{geometry} 

\usepackage{authblk}

\usepackage[amsmath,amsthm,thmmarks]{ntheorem} \usepackage[nottoc]{tocbibind}
\usepackage{appendix}
\usepackage{subcaption}

\usepackage{stackrel} \usepackage{leftidx} \usepackage{scalerel} 
\usepackage{amsfonts}
\usepackage{euscript}
\usepackage{pxfonts}
\usepackage{bbold}

\usepackage[shortlabels]{enumitem}
\usepackage{hyperref}
\hypersetup{
    colorlinks=true,
    linkcolor=red!50!black,
    citecolor=green!50!black,
    urlcolor=blue!80!black
}
\usepackage{cleveref}
\usepackage[normalem]{ulem}

\usepackage[dvipsnames]{xcolor}
\usepackage{tcolorbox}
\usepackage{colortbl}

\usepackage{graphicx}
\usepackage{adjustbox}
\usepackage{float} 
\usepackage[all,2cell,arrow,matrix]{xy} \UseAllTwocells \SilentMatrices
\usepackage{tikz-cd}
\usepackage{tikz}
\usetikzlibrary{arrows}
\usetikzlibrary{decorations.pathreplacing,decorations.markings,decorations.pathmorphing}
\usepackage{commutative-diagrams} 

\usepackage{multirow}
\usepackage{array}
\usepackage{makecell} 
 \makeatletter

\newcommand\pf                {\begin{proof}}
\newcommand\epf               {\end{proof}}
\newcommand\be                {\begin{equation}}
\newcommand\ee                {\end{equation}}
\newcommand\bea               {\begin{eqnarray}}
\newcommand\eea               {\end{eqnarray}}
\newcommand\bnu               {\begin{enumerate}}
\newcommand\enu               {\end{enumerate}}
\newcommand\bit               {\begin{itemize}}
\newcommand\eit               {\end{itemize}}

\newcommand\id                {\mathrm{id}}

\newcommand{\widesim}[3][1.5] {\mathrel{\overset{#2}{\underset{#3}{\scalebox{#1}[1]{$\sim$}}}}}

\def\Circlearrowright{\ensuremath{\rotatebox[origin=c]{90}{$\circlearrowright$}}}

\newcommand{\ncomm}{\mathrel{\raisebox{-0.6ex}{\scalebox{1}[2.2]{$\boldsymbol{\setminus}$}}\kern-0.77em\raisebox{-0.7ex}{\scalebox{1.5}{$\Circlearrowright$}}}}

\DeclareMathOperator{\Hom}           {\mathrm{Hom}}

\DeclareMathOperator{\Id}            {\mathrm{Id}}

\DeclareMathOperator{\Fun}           {\EuScript{F}\mathrm{un}}

\DeclareMathOperator{\Alg}           {\EuScript{A}\mathrm{lg}}

\DeclareMathOperator{\Mod}           {\mathrm{Mod}}

\DeclareMathOperator{\LMod}          {\mathrm{LMod}}

\DeclareMathOperator{\Aut}           {\mathcal{A}\mathrm{ut}}
\DeclareMathOperator{\Ind}           {\mathrm{Ind}}

\DeclareMathOperator{\ot}			       {\otimes}

\DeclareMathOperator*{\td}{\times} \DeclareMathOperator*{\otd}{\otimes} \DeclareMathOperator*{\btd}{\boxtimes} \DeclareMathOperator*{\odotd}{\odot}

\newcommand\Cat			        {\EuScript{C}\mathrm{at}}

\newcommand\Ab			        {\EuScript{A}\mathrm{b}}
\newcommand\vect			      {\mathrm{Vec}}

\newcommand\Mrt             {\EuScript{M}\mathrm{rt}}

\newcommand{\bscale}	{0.7}
\newcommand{\bc}[2][]	{{\@ec{#1}{#2}}} \newcommand{\@ec}[2]	{\mathchoice
    {\displaystyle \raise.9ex\hbox{$\scaleobj{\bscale}{#1}$} {#2}}{\textstyle \raise.9ex\hbox{$\scaleobj{\bscale}{#1}$} {#2}}{\scriptstyle \raise.55ex\hbox{$\scriptstyle \scaleobj{\bscale}{#1}$} {#2}}{\scriptscriptstyle \raise.38ex\hbox{$\scriptscriptstyle \scaleobj{\bscale}{#1}$} {#2}}}

\newcommand\rso{\bgroup\markoverwith {\textcolor{red}{\rule[0.5ex]{2pt}{0.4pt}}}\ULon}

\newcommand\CA              {\EuScript{A}}
\newcommand\CB              {\EuScript{B}}
\newcommand\CC              {\EuScript{C}}
\newcommand\CD              {\EuScript{D}}
\newcommand\CE              {\EuScript{E}}

\newcommand\CM              {\EuScript{M}}
\newcommand\CN              {\EuScript{N}}

\newcommand\CP              {\EuScript{P}}

\newcommand\CS              {\EuScript{S}}

\DeclareMathAlphabet{\mathcal}{OMS}{cmsy}{m}{n}

\newcommand\fZ              {\mathfrak{Z}}

\newcommand\bfone           {\mathbf{1}}

\newcommand\sB			    {\mathsf{B}}

\newcommand\sW			    {\mathsf{W}}
\newcommand\sX			    {\mathsf{X}}

\newcommand\bC			    {\mathbb{C}}

\newcommand\bk			    {\mathbb{k}}
\newcommand\bone            {\mathbb{1}}

\tikzset{->-/.style={decoration={markings,mark=at position #1 with {\arrow{stealth}}},postaction={decorate}},->-/.default=0.55}

\newtheorem{thm}{Theorem}
\newtheorem{prp}[thm]{Proposition}
\newtheorem{lem}[thm]{Lemma}
\newtheorem{crl}[thm]{Corollary}

\newtheorem{cnj}{Conjecture}

\theoremstyle{definition}
\newtheorem{dfn}{Definition}

\newtheorem{cexpl}[dfn]{Counter-Example}

\newtheorem{notation}{Notation}

\theoremstyle{remark}
\newtheorem{rmk}{Remark}

\numberwithin{thm}{section}
\numberwithin{dfn}{section}
\numberwithin{rmk}{section}

 \DeclareMathSizes{6}{6}{5}{4}

\newcommand{\hindAux}[4]{\vcenter{\hbox{$#1 #2$}}#3\vcenter{\hbox{$#1 #4$}}}

\newcommand{\hind}[3]{\mathchoice
    {\hindAux{\scriptstyle}{#1}{#2}{#3}}{\hindAux{\scriptstyle}{#1}{#2}{#3}}{\hindAux{\scriptscriptstyle}{#1}{#2}{#3}}{\hindAux{\scriptscriptstyle}{#1}{#2}{#3}}}

\newcommand{\vind}[3]{\overset{#1}{\underset{#3}{#2}}
}

\DeclareRobustCommand{\vfus}[3]{\mathpalette\vfusAux{{#1}{#2}{#3}}}
\newcommand{\vfusAux}[2]{\vfusDraw{#1}#2}
\newcommand{\vfusDraw}[4]{\mathord{\vcenter{\offinterlineskip
    \halign{\hfil$\m@th #1##$\hfil\cr
      #2\cr
      \noalign{\kern0.25ex}#3\cr
      \noalign{\kern0.25ex}#4\cr
    }}}}

\newsavebox{\catboxReferenceBox}
\newsavebox{\catboxContentBox}
\newsavebox{\catboxDisplayFrame}
\newsavebox{\catboxTextFrame}
\newsavebox{\catboxScriptFrame}
\newsavebox{\catboxScriptscriptFrame}
\newlength{\catboxInnerWidth}
\newlength{\catboxInnerHeight}
\newlength{\catboxFitWidth}
\newlength{\catboxFitHeight}
\newlength{\catboxClearance}
\newlength{\catboxOuterWidth}
\newlength{\catboxOuterHeight}
\newlength{\catboxFrameHeight}
\newlength{\catboxContentHeight}
\newlength{\catboxContentRaise}
\newcommand{\catboxBuildFrame}[3]{\begingroup
    \sbox{\catboxReferenceBox}{$\m@th #1 \CC$}\setlength{\catboxInnerWidth}{\wd\catboxReferenceBox}\setlength{\catboxInnerHeight}{\dimexpr\ht\catboxReferenceBox+\dp\catboxReferenceBox\relax}\setlength{\catboxOuterWidth}{\catboxInnerWidth}\addtolength{\catboxOuterWidth}{#2}\addtolength{\catboxOuterWidth}{#2}\setlength{\catboxOuterHeight}{\catboxInnerHeight}\addtolength{\catboxOuterHeight}{#2}\addtolength{\catboxOuterHeight}{#2}\ifdim\catboxOuterWidth<\catboxOuterHeight
      \setlength{\catboxOuterWidth}{\catboxOuterHeight}\else
      \setlength{\catboxOuterHeight}{\catboxOuterWidth}\fi
    \global\sbox{#3}{\begin{tikzpicture}[inner sep=0pt,outer sep=0pt]
        \node[draw,rectangle,line width=0.4pt,inner sep=0pt,outer sep=0pt,
          minimum width=\catboxOuterWidth,
          minimum height=\catboxOuterHeight] {};
      \end{tikzpicture}}\endgroup
}
\newcommand{\catboxInitializeFrames}{\catboxBuildFrame{\displaystyle}{1.6pt}{\catboxDisplayFrame}\catboxBuildFrame{\textstyle}{1.6pt}{\catboxTextFrame}\catboxBuildFrame{\scriptstyle}{1.1pt}{\catboxScriptFrame}\catboxBuildFrame{\scriptscriptstyle}{0.8pt}{\catboxScriptscriptFrame}}
\DeclareRobustCommand{\btc}[1]{\mathchoice
    {\catboxCompose{\displaystyle}{1.6pt}{#1}{\catboxDisplayFrame}}{\catboxCompose{\textstyle}{1.6pt}{#1}{\catboxTextFrame}}{\catboxCompose{\scriptstyle}{1.1pt}{#1}{\catboxScriptFrame}}{\catboxCompose{\scriptscriptstyle}{0.8pt}{#1}{\catboxScriptscriptFrame}}}
\newcommand{\catboxCompose}[4]{\begingroup
    \sbox{\catboxReferenceBox}{$\m@th #1 \CC$}\setlength{\catboxInnerWidth}{\wd\catboxReferenceBox}\setlength{\catboxInnerHeight}{\dimexpr\ht\catboxReferenceBox+\dp\catboxReferenceBox\relax}\setlength{\catboxOuterWidth}{\catboxInnerWidth}\addtolength{\catboxOuterWidth}{#2}\addtolength{\catboxOuterWidth}{#2}\setlength{\catboxOuterHeight}{\catboxInnerHeight}\addtolength{\catboxOuterHeight}{#2}\addtolength{\catboxOuterHeight}{#2}\ifdim\catboxOuterWidth<\catboxOuterHeight
      \setlength{\catboxOuterWidth}{\catboxOuterHeight}\else
      \setlength{\catboxOuterHeight}{\catboxOuterWidth}\fi
    \setlength{\catboxClearance}{#2}\divide\catboxClearance by 2\relax
    \setlength{\catboxFitWidth}{\catboxOuterWidth}\addtolength{\catboxFitWidth}{-\catboxClearance}\setlength{\catboxFitHeight}{\catboxOuterHeight}\addtolength{\catboxFitHeight}{-\catboxClearance}\sbox{\catboxContentBox}{\adjustbox{max width=\catboxFitWidth,
        max totalheight=\catboxFitHeight}{$\m@th #1 #3$}}\setlength{\catboxFrameHeight}{\dimexpr\ht#4+\dp#4\relax}\setlength{\catboxContentHeight}{\dimexpr\ht\catboxContentBox+\dp\catboxContentBox\relax}\setlength{\catboxContentRaise}{\catboxFrameHeight}\addtolength{\catboxContentRaise}{-\catboxContentHeight}\divide\catboxContentRaise by 2\relax
    \addtolength{\catboxContentRaise}{\dp\catboxContentBox}\addtolength{\catboxContentRaise}{-\dp#4}\mathord{\vcenter{\hbox{\rlap{\usebox{#4}}\makebox[\wd#4][c]{\raisebox{\catboxContentRaise}{\usebox{\catboxContentBox}}}}}}\endgroup
}

\newsavebox{\catotimesReferenceBox}
\newsavebox{\catotimesContentBox}
\newsavebox{\catotimesDisplayFrame}
\newsavebox{\catotimesTextFrame}
\newsavebox{\catotimesScriptFrame}
\newsavebox{\catotimesScriptscriptFrame}
\newlength{\catotimesOuterSize}
\newlength{\catotimesFitSize}
\newlength{\catotimesFrameHeight}
\newlength{\catotimesContentHeight}
\newlength{\catotimesContentRaise}
\newcommand{\catotimesBuildFrame}[2]{\begingroup
    \sbox{\catotimesReferenceBox}{$\m@th #1\otimes$}\setlength{\catotimesOuterSize}{\wd\catotimesReferenceBox}\ifdim\catotimesOuterSize<\dimexpr
        \ht\catotimesReferenceBox+\dp\catotimesReferenceBox\relax
      \setlength{\catotimesOuterSize}{\dimexpr
        \ht\catotimesReferenceBox+\dp\catotimesReferenceBox\relax}\fi
    \global\sbox{#2}{\begin{tikzpicture}[inner sep=0pt,outer sep=0pt]
        \node[draw,circle,line width=0.4pt,inner sep=0pt,outer sep=0pt,
          minimum size=\catotimesOuterSize] {};
      \end{tikzpicture}}\endgroup
}
\newcommand{\catotimesInitializeFrames}{\catotimesBuildFrame{\displaystyle}{\catotimesDisplayFrame}\catotimesBuildFrame{\textstyle}{\catotimesTextFrame}\catotimesBuildFrame{\scriptstyle}{\catotimesScriptFrame}\catotimesBuildFrame{\scriptscriptstyle}{\catotimesScriptscriptFrame}}
\DeclareRobustCommand{\otc}[1]{\mathchoice
    {\catotimesCompose{\displaystyle}{\scriptstyle}{1.2pt}{#1}{\catotimesDisplayFrame}}{\catotimesCompose{\textstyle}{\scriptstyle}{1.2pt}{#1}{\catotimesTextFrame}}{\catotimesCompose{\scriptstyle}{\scriptscriptstyle}{0.8pt}{#1}{\catotimesScriptFrame}}{\catotimesCompose{\scriptscriptstyle}{\scriptscriptstyle}{0.6pt}{#1}{\catotimesScriptscriptFrame}}}
\newcommand{\catotimesCompose}[5]{\begingroup
    \sbox{\catotimesReferenceBox}{$\m@th #1\otimes$}\setlength{\catotimesOuterSize}{\wd\catotimesReferenceBox}\ifdim\catotimesOuterSize<\dimexpr
        \ht\catotimesReferenceBox+\dp\catotimesReferenceBox\relax
      \setlength{\catotimesOuterSize}{\dimexpr
        \ht\catotimesReferenceBox+\dp\catotimesReferenceBox\relax}\fi
    \setlength{\catotimesFitSize}{\catotimesOuterSize}\addtolength{\catotimesFitSize}{-#3}\addtolength{\catotimesFitSize}{-#3}\sbox{\catotimesContentBox}{\adjustbox{max width=\catotimesFitSize,
        max totalheight=\catotimesFitSize}{$\m@th #2 #4$}}\setlength{\catotimesFrameHeight}{\dimexpr\ht#5+\dp#5\relax}\setlength{\catotimesContentHeight}{\dimexpr\ht\catotimesContentBox+\dp\catotimesContentBox\relax}\setlength{\catotimesContentRaise}{\catotimesFrameHeight}\addtolength{\catotimesContentRaise}{-\catotimesContentHeight}\divide\catotimesContentRaise by 2\relax
    \addtolength{\catotimesContentRaise}{\dp\catotimesContentBox}\addtolength{\catotimesContentRaise}{-\dp#5}\mathbin{\vcenter{\hbox{\rlap{\usebox{#5}}\makebox[\wd#5][c]{\raisebox{\catotimesContentRaise}{\usebox{\catotimesContentBox}}}}}}\endgroup
}

\AtBeginDocument{\catboxInitializeFrames
  \catotimesInitializeFrames
}

\newcommand{\sideindices}[3]{\mathpalette\sideindices@{{#1}{#2}{#3}}}

\newcommand{\sideindices@}[2]{\sideindices@@#1#2}

\newcommand{\sideindices@@}[4]{\vcenter{\hbox{\scalebox{0.75}{$\m@th#1#2$}}}#3\vcenter{\hbox{\scalebox{0.75}{$\m@th#1#4$}}}}

 \makeatother

\definecolor{hmDeepGreen}{HTML}{9EB5A7}
\definecolor{hmLightGreen}{HTML}{DCE8E1}
\definecolor{hmDeepCyan}{HTML}{82C8BD}
\definecolor{hmMidCyan}{HTML}{C3E7E0}
\definecolor{hmLightCyan}{HTML}{D8F1EC}
\definecolor{hmLineGreen}{HTML}{1F4E79}
\definecolor{hmLineLightBlue}{HTML}{8FAAC2}
\definecolor{hmLineMidBlue}{HTML}{577C9E}
\tikzset{
    topological bulk/.style={fill=hmDeepCyan},
    condensed phase/.style={fill=hmLightCyan},
    topological wall/.style={draw=hmLineGreen, very thick}
}

\begin{document}

\title{2-Morita Theory of $E_2$-Algebras and Module Categories}
\author[a]{Rongge Xu \thanks{Email: \href{mailto:xurongge@westlake.edu.cn}{\tt xurongge@tsinghua.edu.cn}}}
\author[b]{Holiverse Yang \thanks{Email: \href{mailto:holiversey@gmail.com}{\tt holiversey@gmail.com}}}
\affil[a]{Yau Mathematical Science Center, Tsinghua University, Beijing,  China}
\affil[b]{Department of Physics, The Chinese University of Hong Kong,\authorcr Shatin, New Territories, Hong Kong, China}

\date{\vspace{-5ex}}

\maketitle

\begin{abstract}

    In our previous work \cite{XY25}, we introduced 2-Morita equivalence for $E_2$-algebras using topological pictures. Here we develop a more systematic understanding of Morita equivalence for topological orders in different dimensions through the notion of $n$-Morita categories $\Mrt_{E_n}(\CC)$, in which different definitions of $n$-Morita equivalence can be unified as equivalences of objects in $\Mrt_{E_n}(\CC)$. This framework also provides a comparison between Haugseng's and Gwilliam--Scheimbauer's constructions of higher Morita categories.

    For $n=1$ and $n=2$, we prove that the functor $\mathrm{Mod}_n:\Mrt_{E_n}(\CC)\to \Mrt_{E_{n-1}}(\LMod^{\mathrm{rep}}(\CC))$ 
    is an equivalence, relating the algebraic description to its module category realization. In this formulation, the useful notion of a bi-bimodule emerges naturally, which unifies local modules and confined modules underneath. Its explicit orientation, together with the corresponding fusion rules, helps clarify the relations among defects arising in the condensation theory of topological orders.

\end{abstract}
 
\tableofcontents

\section{Introduction} \label{sec:introduction}

Classical Morita theory \cite{Morita58} has become a fundamental tool in many areas of mathematics, particularly in representation theory \cite{Ost03}. In recent years, Morita-theoretic ideas have also played an important role in quantum many-body systems and quantum field theories, including open--closed conformal field theory \cite{KR08,KR09} and topological quantum field theory \cite{FSV13,CRS18,CRS18defect,Turzillo20}. In these settings, Morita equivalence appears in several closely related forms: as an equivalence of module categories, as the existence of an invertible bimodule, and, for fusion categories, through equivalence of their Drinfeld centers.

These are all manifestations of 1-Morita theory, which governs $E_1$-algebras. Classically, two algebras $A_1$ and $A_2$ are Morita equivalent if their categories of modules are equivalent. Equivalently, under the usual hypotheses, there exists an invertible $A_1$--$A_2$ bimodule. This admits a natural categorical reformulation: one considers the 2-category whose objects are algebras, whose 1-morphisms are bimodules, and whose 2-morphisms are bimodule homomorphisms. Morita equivalence is then simply equivalence of objects in this 2-category. This formulation places the different classical descriptions of Morita equivalence in a common categorical framework and is particularly well suited for higher-dimensional generalizations.

For $E_n$-algebras, however, the corresponding higher Morita structures are considerably more intricate. Different constructions of higher Morita categories exist, and the relation between algebraic descriptions, module-category realizations, and their interpretations in higher-dimensional topological phases is less immediate. As categorical methods become increasingly important in the study of topological orders and quantum field theories, it is therefore natural to seek a systematic framework in which these different notions of higher Morita equivalence can be compared and unified.

\subsection{Higher Morita Theory and Module Categories}\label{subsec:higher-morita-module-categories}
There are several models and concrete realizations of higher Morita categories \cite{Hau17,GS18,DSPS20,BJS21}. We briefly recall the algebraic model of Haugseng and the factorization-algebra model of Gwilliam--Scheimbauer, before introducing a third description in terms of bi-bimodules that will be used throughout this paper.

\begin{dfn}[Haugseng's model, \cite{Hau17}]\label{dfn:higher_morita_haugseng}
Let $\CC$ be a suitable $E_n$-monoidal $\infty$-category in which the required relative tensor products exist and are preserved by tensoring.
Haugseng's higher Morita category $\Mrt^{Haug}_{E_n}(\CC)$ is an $(\infty,n+1)$-category whose
\begin{itemize}
\item objects are $E_n$-algebras in $\CC$;
\item for two $E_n$-algebras $\mathsf{X}_1$ and $\mathsf{X}_2$, 1-morphisms are objects of
\begin{align*}
\Alg_{E_{n-1}}\bigl(\hind{\mathsf{X}_1}{\CC}{\mathsf{X}_2}\bigr);
\end{align*}
\item for two 1-morphisms $\mathsf{W}_1$ and $\mathsf{W}_2$ from $\mathsf{X}_1$ to $\mathsf{X}_2$, 2-morphisms are objects of
\begin{align*}
\Alg_{E_{n-2}}\bigl(
\hind{\mathsf{W}_1}
{\hind{\mathsf{X}_1}{\CC}{\mathsf{X}_2}}
{\mathsf{W}_2}
\bigr);
\end{align*}
\item for $2<k<n$, the construction is iterated by taking $E_{n-k}$-algebras in the corresponding iterated bimodule category;
\item $n$-morphisms are the final iterated bimodules, without an additional pointing; and
\item $(n+1)$-morphisms are maps of the final iterated bimodules.
\end{itemize}
Recursively, for two $E_n$-algebras $\mathsf{X}_1$ and $\mathsf{X}_2$, its mapping $(\infty,n)$-category is
\begin{align*}
\hom_{\Mrt^{Haug}_{E_n}(\CC)}(\mathsf{X}_1,\mathsf{X}_2)
\simeq \Mrt^{Haug}_{E_{n-1}}\bigl(\hind{\mathsf{X}_1}{\CC}{\mathsf{X}_2}\bigr).
\end{align*}
\end{dfn}

Haugseng's construction gives an algebraic and operadic organization of iterated bimodules, with composition encoded by relative tensor products. It is therefore particularly well suited for formulating higher Morita categories abstractly and for studying their categorical structure.

\begin{dfn}[Gwilliam--Scheimbauer's model, \cite{GS18}]\label{dfn:higher_morita_gs}
Let $\CS$ be a sufficiently well-behaved symmetric monoidal $(\infty,N)$-category whose tensor product preserves sifted colimits separately in each variable.
The factorization higher Morita category $\Mrt^{GS}_{E_n}(\CS)$ is an $(\infty,n+N)$-category whose objects are $E_n$-algebras in $\CS$ and whose $k$-morphisms, for $1\leq k\leq n$, are $E_{n-k}$-algebras in the appropriate iterated bimodule categories.
In the iterated-bimodule notation used in this paper, a 1-morphism from $\mathsf{X}_1$ to $\mathsf{X}_2$ is schematically an object of
\begin{align*}
\hind{\mathsf{X}_1}{\Alg_{E_{n-1}}(\CS)}{\mathsf{X}_2},
\end{align*}
and, for $n\geq 2$, a 2-morphism between $\mathsf{W}_1$ and $\mathsf{W}_2$ is schematically an object of
\begin{align*}
\hind{\mathsf{W}_1}{\hind{\mathsf{X}_1}{\Alg_{E_{n-2}}(\CS)}{\mathsf{X}_2}}{\mathsf{W}_2}.
\end{align*}
This pattern is iterated through the $n$-morphisms; at the top level, an $E_0$-algebra is a pointed iterated bimodule. The remaining $N$ levels of morphisms are inherited from $\CS$.
\end{dfn}

The Gwilliam--Scheimbauer construction provides a geometric realization of the same hierarchy through locally constant and constructible factorization algebras. In this picture, higher morphisms can be interpreted as defects of increasing codimension, and their compositions are encoded geometrically by factorization products. This makes the model particularly natural for locality and defect-theoretic questions \cite{Sch14,GS18,JS17}.

The two constructions describe closely related hierarchies of $E_n$-algebras and iterated bimodules, although they differ in conventions at the highest morphism level: in the factorization model the top-dimensional bimodules are naturally pointed, while Haugseng's model is unpointed there. Gwilliam and Scheimbauer give a dictionary between the two approaches \cite[Sections~1.1 and~1.4]{GS18}, but we do not assume a general equivalence of higher categories between them.

For the study of topological defects, it is useful to make the multiple bimodule structures carried by higher morphisms more explicit. This motivates the following model, formulated in terms of the bi-bimodules introduced above.

\begin{dfn}[Bi-bimodule model]\label{dfn:higher_morita_bibimodule}
Let $\CC$ be an $E_{m+2}$-monoidal $n$-category.
The $(n+1)$-category $\Mrt_{E_n}(\CC)$ consists of
\begin{itemize}
\item objects are $E_n$-algebras in $\CC$, i.e. objects in $\Alg_{E_n}(\CC)$;
\item for two $E_n$-algebras $\mathsf{X}_1$ and $\mathsf{X}_2$, 1-morphisms are objects in $\hind{\sX_1}{\Alg_{E_{n-1}}(\CC)}{\sX_2}$;
\item 2-morphisms are objects in
\begin{align*}
    \hind{\sW_1}{\vind{\sX_1}{\Alg_{E_{n-2}}(\CC)}{\sX_2}}{\sW_2},
\end{align*}
which is the category of $\sX_1$--$\sX_2$-bi-$\sW_1$--$\sW_2$-bimodules in $E_{n-2}$-algebras;
\item $k$-morphisms ($k\leq n$) are objects in the corresponding iterated bi-bi-$\cdots$-bimodule categories of $E_{n-k}$-algebras.
\end{itemize}
\end{dfn}

The advantage of this description is not to provide another abstract construction of the higher Morita category, but to expose explicitly the source, target, and transverse bimodule structures of each higher morphism. These structures have a direct interpretation in topological orders: the different module actions record the phases and defects adjacent to a given defect, while relative tensor products describe their fusion.

In this paper, we compare the three descriptions level by level. For each $k$, we identify the corresponding $k$-morphisms in the Haugseng, Gwilliam--Scheimbauer, and bi-bimodule models. We do not, however, establish a comparison of all higher compositions and coherence data, and therefore do not claim an equivalence of the three higher categories in full generality.

This levelwise comparison is sufficient for our applications. Haugseng's model provides the abstract algebraic framework for higher Morita theory, while the factorization-algebra model gives a geometric interpretation in terms of stratified defects. The bi-bimodule model makes the orientations and fusion rules of these defects explicit, and is therefore especially convenient for describing and computing condensation defects of different dimensions in topological orders.

\vspace{3em}

For a suitable $E_n$-monoidal category $\CC$, the category $\LMod(\CC)$ of left $\CC$-module categories carries an $E_{n-1}$-monoidal structure \cite{DR18}. This leads to the expected module-realization functor
\begin{align*}
    \mathrm{Mod}_n:
    \Mrt_{E_n}(\CC)
    \longrightarrow
    \Mrt_{E_{n-1}}(\LMod(\CC)).
\end{align*}
Its representable essential image $\Mrt_{E_{n-1}}(\LMod^{rep}(\CC))\subset \Mrt_{E_{n-1}}(\LMod(\CC))$ is described by the following data:
\begin{itemize}
    \item objects are the representable left module $n$-categories $\hind{}{\CC}{\mathsf{X}_i}$ associated with $E_n$-algebras $\mathsf{X}_i$ in $\CC$;
    \item 1-morphisms are the $E_{n-2}$-monoidal bimodules $\hind{}{\CC}{\mathsf{W}}$ induced by $E_{n-1}$-algebra bimodules $\mathsf{W}$ between the corresponding $E_{n-1}$-monoidal categories $\hind{}{\CC}{\mathsf{X}_i}$;
    \item for $2\leq k\leq n-1$, $k$-morphisms are obtained inductively as $E_{n-k-1}$-monoidal bimodules between the $E_{n-k}$-monoidal bimodules at the preceding level;
    \item $n$-morphisms are compatible bimodule functors; and
    \item $(n+1)$-morphisms are bimodule natural transformations compatible with all the preceding module structures.
\end{itemize}

For $n=1$, this recovers the classical 2-functor that sends an algebra $A$ to its module category $\hind{}{\CC}{A}$ and a bimodule to the corresponding functor given by relative tensor product. In this paper, we construct the module-realization functor for $n=1,2$ and establish the following result.
\begin{thm}
    For $n\in\{1,2\}$, the functor
    \begin{align*}
        \mathrm{Mod}_n:
        \Mrt_{E_n}(\CC)
        \longrightarrow
        \Mrt_{E_{n-1}}(\LMod(\CC))
    \end{align*}
    is $n$-fully faithful.
\end{thm}

We expect the same $n$-fully faithful conclusion to hold in all dimensions.

The module-realization functor also translates algebraic higher Morita data into a direct geometric interpretation: here we show the case of $n=1,2$:
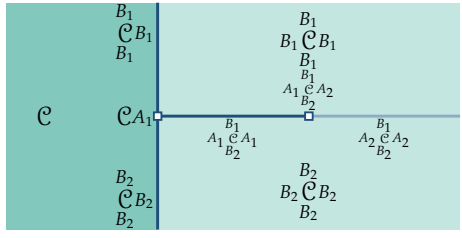
\begin{figure}[H]
    \centering
    \begin{tikzpicture}
        \path[fill=hmDeepCyan] (-2,-1.5) rectangle(0,1.5);
        \path[fill=hmMidCyan] (0,-1.5) rectangle(4,1.5);
        \draw[draw=hmLineGreen,line width=1.1pt] (0,-1.5) -- (0,1.5);
        \draw[draw=hmLineGreen,line width=1.1pt] (0,0) -- (2,0);
        \draw[draw=hmLineLightBlue,line width=1.1pt] (2,0) -- (4,0);
        \filldraw[fill=white,draw=hmLineGreen,line width=0.7pt]
            (-0.05,-0.05) rectangle(0.05,0.05);
        \filldraw[fill=white,draw=hmLineGreen,line width=0.7pt]
            (1.95,-0.05) rectangle(2.05,0.05);

        \node at(-1.5,0){$\CC$};
        \node at(-0.3,1.1){$\hind{}{\vind{B_1}{\CC}{B_1}}{B_1}$};
        \node at(-0.3,-1.1){$\hind{}{\vind{B_2}{\CC}{B_2}}{B_2}$};
        \node at(-0.3,0){$\hind{}{\CC}{A_1}$};
        \node at(2,1){$\hind{B_1}{\vind{B_1}{\CC}{B_1}}{B_1}$};
        \node at(2,-1){$\hind{B_2}{\vind{B_2}{\CC}{B_2}}{B_2}$};
        \node at(1,-0.3){\tiny$\hind{A_1}{\vind{B_1}{\CC}{B_2}}{A_1}$};
        \node at(2,0.35){\tiny$\hind{A_1}{\vind{B_1}{\CC}{B_2}}{A_2}$};
        \node at(3,-0.3){\tiny$\hind{A_2}{\vind{B_1}{\CC}{B_2}}{A_2}$};
    \end{tikzpicture}
\end{figure}

In the finite setting relevant to topological orders, representability gives a stronger conclusion: for algebras in finite multi-tensor $n$-categories and after restricting the target to finite module categories, we prove the following result for $n=1,2$.
\begin{thm}
    For $n\in\{1,2\}$, the functor
    \begin{align*}
        \mathrm{Mod}^{fin}_n:
        \Mrt_{E_n}(\CC)
        \longrightarrow
        \Mrt_{E_{n-1}}(\LMod^{fin}(\CC))
    \end{align*}
    is an equivalence, where $\LMod^{fin}(\CC):=\LMod_{\CC}(\Cat^{fin}_{\bk})$ is the category of finite $\bk$-linear left $\CC$-module categories.
\end{thm}
We expect an appropriately finitary analogue of this equivalence for general $n$.

In a $(2+1)$D topological order described by a modular fusion category $\CC$, a condensable $E_2$-algebra $B$ determines both the condensed bulk phase $\hind{}{\CC}{B}^{loc}$ and the intervening gapped domain wall $\hind{}{\CC}{B}$. In the bi-bimodule description, the local-module category is recovered as
\begin{align*}
    \hind{B}{\vind{B}{\CC}{B}}{B}
    \simeq
    \hind{}{\CC}{B}^{loc}.
\end{align*}
Condensation may also take place within a one-dimensional gapped domain wall \cite{XY25}.
The module realization functor $\mathrm{Mod}_2$ places point defects, one-dimensional domain walls, and their higher junctions in a single framework: the bi-bimodule structures record the adjacent condensed $E_2$ and $E_1$ phases, while relative tensor products encode the fusion of condensation defects \ref{sec:functor_to_module_cat}.

\subsection{Notation and Conventions}

In this paper, for $n=1,2$, we mainly consider two universes.
The basic one is $\Cat^{rex}$, which consists of
\begin{itemize}
    \item objects given by finitely cocomplete categories;
    \item 1-morphisms given by right exact functors, i.e. functors that preserve finite colimits;
    \item 2-morphisms given by natural transformations;
    \item a symmetric monoidal structure given by the Kelly--López tensor product \cite{Kel82,LF13}.
\end{itemize}
When discussing topological orders, we consider the category $\Cat^{fin}_{\bk}$ over an algebraically closed field $\bk$:
\begin{itemize}
    \item objects are finite $\bk$-linear categories;
    \item 1-morphisms are $\bk$-linear functors;
    \item its monoidal structure is given by the Deligne--Kelly tensor product \cite{EGNO15,LF13};
    \item the tensor unit is $\vect$, the category of finite-dimensional $\bk$-vector spaces.
\end{itemize}
We denote the tensor product of categories and the tensor unit in these two universes by $\boxtimes$ and $1$, respectively.
We also consider some special objects in $\Cat^{fin}_{\bk}$, such as finite multi-tensor categories, i.e. $E_1$-algebras in $\Cat^{fin}_{\bk}$ with rigidity.

\begin{dfn}[\cite{Lur17}]
    An {\bf $E_k$-algebra} is an algebra over an $E_k$-operad, in which the $n$-ary operations on an object are parameterized by rectilinear disjoint embeddings of $n$ $k$-dimensional cubes into another $k$-dimensional cube.
\end{dfn}

Intuitively, an $E_m$-monoidal $n$-category $\CC$ can be understood as an $n$-category equipped with $m$ compatible tensor products.

\begin{dfn}
    An {\bf $E_m$-monoidal $n$-category} is an $E_m$-algebra in $n\Cat$.
\end{dfn}

For example, an $E_1$-algebra in the 2-category $\Cat$ is equivalent to an $E_1$-monoidal 1-category, since the higher morphisms encoding higher coherences are trivial in $\Cat$.

We also use the following iterative working definition of $E_k$-algebras, which is equivalent to the original one in the sense that applying the functor $\Alg_{E_1}$ at each step provides a new compatible tensor product structure.
\begin{dfn}[Working Definition, \cite{Lur17}]
    Let $\CC$ be an $E_m$-monoidal $1$-category.
    An {\bf $E_1$-algebra} is a unital associative algebra in $\CC$.
    We denote the category of $E_1$-algebras in $\CC$ by $\Alg_{E_1}(\CC)$, whose objects are $E_1$-algebras in $\CC$ and whose morphisms are $E_1$-algebra homomorphisms.
    An {\bf $E_k$-algebra} is an $E_1$-algebra in the category $\Alg_{E_{k-1}}(\CC)$ of $E_{k-1}$-algebras in $\CC$, where
    $\Alg_{E_k}(\CC)=\Alg_{E_1}(\Alg_{E_{k-1}}(\CC))$.
\end{dfn}

\begin{notation}
    We use $A_i$ to denote $E_1$-algebras, $B_i$ to denote $E_2$-algebras, etc., and $M_i$ to denote $E_0$-algebras.
\end{notation}

\begin{notation}
    For an $E_1$-monoidal $n$-category $\CC$, we use $\hind{A_1}{\CC}{A_2}$ to denote the category of $A_1$-$A_2$-bimodules in $\CC$ whose actions are horizontal.
    For an $E_2$-monoidal $n$-category, there are two possible multiplication directions. We use $\vind{A_3}{\CC}{A_4}$ to denote the category of $A_3$-$A_4$-bimodules in $\CC$ whose actions are vertical.
\end{notation}

\begin{notation}
    We use $X_1\otc{Y}X_2$ to denote the relative tensor product $X_1\otd\limits_{Y} X_2$, representing the horizontal relative fusion of $X_1$ and $X_2$.
    Similarly, we use $\vfus{X_1}{\otc{Y}}{X_2}$ to denote the relative tensor product $X_1\otd\limits_{Y} X_2$, representing the vertical relative fusion of $X_1$ and $X_2$.
\end{notation}

\begin{notation}
    We use $\CC_1\btc{\CD}\CC_2$ to denote the relative tensor product $\CC_1\btd\limits_{\CD} \CC_2$, representing the horizontal relative fusion of $\CC_1$ and $\CC_2$.
    We use $\vfus{\CC_1}{\btc{\CD}}{\CC_2}$ to denote the relative tensor product $\CC_1\btd\limits_{\CD} \CC_2$, representing the vertical relative fusion of $\CC_1$ and $\CC_2$.
\end{notation}

\begin{notation}
    We use $\hind{A_1}{\overset{\rotatebox[origin=c]{-90}{$\scriptstyle($}}{\underset{\rotatebox[origin=c]{90}{$\scriptstyle($}}{\vind{B_1}{\CC}{B_2}}}}{A_2}$ \footnote{We omit the brackets in bi-bimodules when there is no ambiguity.} to denote the category of $B_1$-$B_2$-bi-$A_1$-$A_2$-bimodules in $\CC$ (see Definition \ref{dfn:cat_of_bibimod}).
    We also use $\hind{(B_1}{\vind{A_1}{\CC}{A_2}}{B_2)}$ when $E_1$-algebras are placed between the left and right $E_2$-algebras.
    We use $\hind{A_1}{\vind{A_2}{\CC}{A_3}}{A_4}$ to denote the category of $A_1$-$A_2$-$A_3$-$A_4$-quadrimodules in $\CC$ (see Definition \ref{dfn:quadrimodule}).

\begin{figure}[ht]
\centering
\begin{subfigure}[c]{0.45\textwidth}
    \centering
    \begin{tikzpicture}[line cap=round,line join=round]
\path[use as bounding box] (-2,-2) rectangle (2,2);

        \filldraw[fill=gray!40,draw=white]
            (-2,1.33) rectangle (2,-1.33);

        \draw[thick] (-2,0) -- (2,0);
        \fill[black] (0,0) circle (0.07);

        \node at (0,0.8) {$B_1$};
        \node at (0,-0.8) {$B_2$};

        \node at (-1,0.25) {$A_1$};
        \node at (1,0.25) {$A_2$};

        \node at (0,-0.2) {$M$};
    \end{tikzpicture}
\end{subfigure}
\begin{subfigure}[c]{0.45\textwidth}
    \centering
    \begin{tikzpicture}[line cap=round,line join=round]
\path[use as bounding box] (-2,-2) rectangle (2,2);

        \filldraw[fill=gray!40,draw=white]
            (-2,1.33) rectangle (2,-1.33);

        \draw[thick] (-2,0) -- (2,0);
        \draw[thick] (0,-1.33) -- (0,1.33);
        \fill[black] (0,0) circle (0.07);

        \node at (-1,0.8) {$B_1$};
        \node at (-1,-0.8) {$B_4$};
        \node at (1,0.8) {$B_2$};
        \node at (1,-0.8) {$B_3$};

        \node at (-1.4,0.25) {$A_1$};
        \node at (1.4,0.25) {$A_2$};
        \node at (0.3,1) {$A_3$};
        \node at (0.3,-1) {$A_4$};

        \node at (-0.25,-0.2) {$M'$};
    \end{tikzpicture}
\end{subfigure}
\caption{$M \in \hind{A_1}{\overset{\rotatebox[origin=c]{-90}{$\scriptstyle($}}{\underset{\rotatebox[origin=c]{90}{$\scriptstyle($}}{\vind{B_1}{\CC}{B_2}}}}{A_2}$ for the left figure, and $M' \in \hind{A_1}{\vind{A_2}{\CC}{A_3}}{A_4} \simeq \hind{\vfus{A_3}{\otc{B_1}}{A_1}}{\overset{\rotatebox[origin=c]{-90}{$\scriptstyle($}}{\underset{\rotatebox[origin=c]{90}{$\scriptstyle($}}{\vind{B_2}{\CC}{B_4}}}}{\vfus{A_2}{\otc{B_3}}{A_4}} \simeq \hind{(B_4}{\vind{A_1\otc{B_1}A_3}{\CC}{A_4\otc{B_3}A_2}}{B_2)}$ for the right figure.}
\end{figure}
\end{notation}

\begin{notation}
    We use $\hind{\CC}{\vind{\CA_1}{\Fun}{\CA_2}}{}(\CM_1,\CM_2)$ to denote the category of $\CA_1$-$\CA_2$-bimodule functors in $\LMod(\CC)$.
\end{notation}

\begin{dfn}[\cite{EGNO15}]
    A {\bf multitensor category} (over $\bC$) is a $\bC$-linear abelian rigid monoidal category such that $\otimes:\CC\times \CC\to \CC$ is $\bC$-linear.
    A {\bf fusion category} is a finite semisimple tensor category whose tensor unit is simple.
\end{dfn}

\begin{dfn}
    An algebra $(A,m)$ in a monoidal category $\CC$ is called {\bf separable} if there exists an $A$-$A$-bimodule homomorphism $s:A\to A\ot A$ such that $m\circ s=\id_A$.
\end{dfn}

\subsection*{Layout}\label{subsec:layout}
The paper is organized as follows.
In Section~\ref{sec:introduction}, we introduce the background of higher Morita theory, and establish the notations used throughout the paper.
Section~2 is devoted to 1-Morita equivalences. We review different definitions of 1-Morita equivalence and use their equivalence to study their relation to condensations in topological orders.
In Section~3, we develop a model of the 2-Morita category using bi-bimodules and construct the $E_2$ module realization functor $\mathrm{Mod}_2$. We also compare 2-Morita equivalence of $E_2$-algebras with $E_2$-monoidal equivalence of the corresponding local module categories, and study higher-codimensional defect condensations. 
Section~4 briefly extends these constructions to the $n$-Morita category. We also compare Haugseng's model with the Gwilliam--Scheimbauer model at each morphism level.

\vspace{3em}
 
\section{The 1-Morita Category and Module Realization}\label{sec:one-morita}
This section develops 1-Morita theory and its module realization. We introduce algebras, bimodules, and Morita equivalence in \cref{subsec:one-morita-algebras-bimodules}, construct the module realization functor in \cref{subsec:one-morita-module-realization}, specialize to finite tensor categories in \cref{subsec:one-morita-finite}, and relate local 1-Morita equivalence to condensation in \cref{subsec:local-one-morita}.

Classically, two ($E_1$-)algebras are said to be (1-)Morita equivalent if their categories of modules are equivalent as module categories \cite{Morita58}:

\begin{dfn}[\cite{Morita58}]\label{dfn:1-Morita}
Let $\CC$ be an $E_1$-monoidal $n$-category.
Two $E_1$-algebras $A_1,A_2\in\CC$ are {\bf 1-Morita equivalent} if
\[
    \hind{}{\CC}{A_1}\simeq \hind{}{\CC}{A_2}
\]
as left $\CC$-module $n$-categories.
\end{dfn}

On the other hand, one may organize $E_1$-algebras and bimodules into a (1-)Morita category \cite{Benabou1967}, in which Morita equivalence is expressed as equivalence between the corresponding algebra objects.

These two viewpoints are related by the generalized Eilenberg--Watts theorem, which identifies bimodules with module functors. In this way, Morita equivalence can be described equivalently at both the algebraic and categorical levels. In the physical setting, this allows us to treat condensable algebras and the corresponding topological orders within the same framework.

\subsection{Algebras, Bimodules, and Morita Equivalence}\label{subsec:one-morita-algebras-bimodules}

\begin{figure}[H]
    \centering
    \begin{tikzpicture}
        \draw[thick] (-2,0) -- (4,0);
        \fill[black] (0,0) circle (0.06);
        \fill[black] (2,0) circle (0.06);

        \node at(-1,0.5){$A_1$};
        \node at(1,0.5){$A_2$};
        \node at(3,0.5){$A_3$};
        \node at(0,0.4) {$M_2$};
        \node at(2,0.4) {$M_2$};
        \end{tikzpicture}
\end{figure}

Now we give the definition of $E_1$ Morita category, which is originally proposed in \cite{Benabou1967}.
\begin{dfn}[\cite{Benabou1967}]
    Let $\CC$ be a finitely cocomplete $E_{m+1}$-monoidal 1-category such that $\otimes:\CC\times\CC\to\CC$ preserves all finite colimits in each variable.
    The 2-category $\Mrt_{E_1}(\CC)$ consists of the following data:
    \begin{itemize}
        \item objects $A_i$ are objects in $\Alg_{E_1}(\CC)$;
        \item 1-morphisms between two $E_1$-algebras $A_1$ and $A_2$ are objects in $\hind{A_1}{\CC}{A_2}$;
        \item 2-morphisms between two $A_1$-$A_2$-bimodules $M_1$, $M_2$ are $E_0$ $A_1$-$A_2$-bimodule homomorphisms.
        \item the identity 1-morphism is $A$ itself for an object $A\in\Alg_{E_1}(\CC)$;
        \item identity 2-morphism on $M$ is identity morphism $\id_M$ in $\CC$;
        \item the composition of 1-morphisms are the relative tensor product: for $M_1\in \hind{A_1}{\CC}{A_2}$ and $M_2\in \hind{A_2}{\CC}{A_3}$, their composition is $M_1\otc{A_2} M_2\in \hind{A_1}{\CC}{A_3}$
        \item the horizontal composition of 2-morphisms are obvious and the vertical composition of bimodule homomorphisms are composition of morphisms in $\CC$.
    \end{itemize}
\end{dfn}

In 2-category $\Mrt_{E_1}(\CC)$, an invertible bimodule is indeed an invertible morphism. Hence, 
\begin{dfn}\label{dfn: 1-Mrt_equivalnce}
    Two $E_1$-algebras are {\bf 1-Morita equivalent} if they are equivalent objects in $\Mrt_{E_1}(\CC)$.
\end{dfn}

\begin{thm}[\cite{Hau17}]
    Let $\CC$ be a finitely cocomplete $E_{m+1}$-monoidal 1-category.
    The 2-category $\Mrt_{E_1}(\CC)$ is $E_m$-monoidal.
\end{thm}

\begin{dfn}
    Let $\CC$ be a finitely cocomplete $E_{m+1}$-monoidal 1-category.
    The 2-category $\Mrt^{sep}_{E_1}(\CC)$ is the full sub 2-category of $\Mrt_{E_1}(\CC)$ whose objects are separable $E_1$-algebras, 1-morphisms are finitely generated bimodules.
\end{dfn}

$\Mrt_{E_1}^{sep}(\CC)$ provides a model of condensation completion $\Sigma(\CC)$ of $\CC$ \cite{GJF19}.



\subsection{The Module Realization Functor}\label{subsec:one-morita-module-realization}

\begin{figure}[H]
    \centering
\begin{tikzpicture}
          \path[use as bounding box] (-2,-0.7) rectangle (4,0.8);
          \draw[draw=hmLineGreen,line width=1.1pt] (-2,0) -- (0,0);
          \draw[draw=hmLineMidBlue,line width=1.1pt] (0,0) -- (2,0);
          \draw[draw=hmLineLightBlue,line width=1.1pt] (2,0) -- (4,0);
          \filldraw[fill=white] (-0.05,-0.05) rectangle(0.05,0.05);
          \filldraw[fill=white] (1.95,-0.05) rectangle(2.05,0.05);

          \node at(-1,0.5){$\CC$};
          \node at(1,0.5){$\hind{A_1}{\CC}{A_1}$};
          \node at(3,0.5){$\hind{A_2}{\CC}{A_2}$};
          \node at(0,-0.4) {$\hind{}{\CC}{A_1}$};
          \node at(2,-0.4) {$\hind{A_1}{\CC}{A_2}$};
      \end{tikzpicture}
    \begin{tikzpicture}
        \draw[-latex] (-0.5,1) -- (0.6,1);
        \draw[draw=none] (-0.5,0) -- (0.5,0.5);
    \end{tikzpicture}
\begin{tikzpicture}
          \path[use as bounding box] (-2,-0.7) rectangle (4,0.8);
          \draw[draw=hmLineGreen,line width=1.1pt] (-2,0) -- (0,0);
          \draw[draw=hmLineLightBlue,line width=1.1pt] (0,0) -- (4,0);
          \filldraw[fill=white] (-0.05,-0.05) rectangle(0.05,0.05);

          \node at(-1,0.5){$\CC$};
          \node at(3,0.5){$\hind{A_2}{\CC}{A_2}$};
          \node at(0,-0.4) {$\hind{}{\CC}{A_2}$};
      \end{tikzpicture}
\end{figure}

\begin{dfn}[\cite{DR18}]
    Let $\CC$ be a finitely cocomplete $E_{m+1}$-monoidal 1-category.
    The 2-category $\LMod(\CC)$ consists of 
    \begin{itemize}
        \item objects are left $\CC$-module categories $\CM_i$;
        \item 1-morphisms between two left $\CC$-module categories are left $\CC$-module functors;
        \item 2-morphisms between two left $\CC$-module functors are left $\CC$-module natural transformations.
    \end{itemize}
\end{dfn}

\begin{thm}[\cite{Hau17}]
    Let $\CC$ be a finitely cocomplete $E_{m+1}$-monoidal 1-category.
    The 2-category $\LMod(\CC)$ is $E_m$-monoidal.
\end{thm}

\begin{thm}[$E_1$ Eilenberg--Watts Theorem, \cite{EGNO15,FSS20}]
\label{thm:generalized_Eilenberg-Watts}

    Let $\CC$ be a finitely cocomplete $E_1$-monoidal $1$-category.
    For two $E_1$-algebras $A_1,A_2\in\Alg_{E_1}(\CC)$, there is an equivalence of categories
    \begin{align*}
        \hind{A_1}{\CC}{A_2}
        &\longrightarrow
        \Fun_{\CC}(\hind{}{\CC}{A_1},\hind{}{\CC}{A_2}),\\
        M
        &\longmapsto
        -\otc{A_1} M.
    \end{align*}

    In particular, the equivalence
    \begin{align*}
        \hind{A}{\CC}{A}
        \longrightarrow
        \Fun_{\CC}(\hind{}{\CC}{A},\hind{}{\CC}{A})
    \end{align*}
    is monoidal.

\end{thm}

\begin{proof}

    We define two functors as follows:

    \[
    \begin{tikzcd}
        \hind{A_1}{\CC}{A_2}
        &
        \Fun_{\CC}(\hind{}{\CC}{A_1},\hind{}{\CC}{A_2})
        \\[-2em]
        {M}
        &
        {-\otc{A_1} M}
        \\
        {N}
        &
        {-\otc{A_1} N}
        \arrow[from=1-1, to=1-2]
        \arrow[maps to, from=2-1, to=2-2]
        \arrow[""{name=0, anchor=center, inner sep=0}, "f"', from=2-1, to=3-1]
        \arrow[""{name=0p, anchor=center, inner sep=0}, phantom, from=2-1, to=3-1, start anchor=center, end anchor=center]
        \arrow[""{name=1, anchor=center, inner sep=0}, "{-\otc{A_1} f}", from=2-2, to=3-2]
        \arrow[""{name=1p, anchor=center, inner sep=0}, phantom, from=2-2, to=3-2, start anchor=center, end anchor=center]
        \arrow[maps to, from=3-1, to=3-2]
    \end{tikzcd}
    \]

    and

    \[
    \begin{tikzcd}
        \Fun_{\CC}(\hind{}{\CC}{A_1},\hind{}{\CC}{A_2})
        &
        \hind{A_1}{\CC}{A_2}
        \\[-2em]
        {F_1}
        &
        {F_1(A_1)}
        \\
        {F_2}
        &
        {F_2(A_1)}
        \arrow[from=1-1, to=1-2]
        \arrow[maps to, from=2-1, to=2-2]
        \arrow[""{name=0, anchor=center, inner sep=0}, "\delta"', from=2-1, to=3-1]
        \arrow[""{name=0p, anchor=center, inner sep=0}, phantom, from=2-1, to=3-1, start anchor=center, end anchor=center]
        \arrow[""{name=1, anchor=center, inner sep=0}, "{\delta_{A_1}}", from=2-2, to=3-2]
        \arrow[""{name=1p, anchor=center, inner sep=0}, phantom, from=2-2, to=3-2, start anchor=center, end anchor=center]
        \arrow[maps to, from=3-1, to=3-2]
    \end{tikzcd}
    \]

    We first show that these functors are well-defined.
    We begin by showing that $-\otc{A_1}M$ is a right exact left $\CC$-module functor. Namely, for any $c\in\CC$ and $X\in\hind{}{\CC}{A_1}$, there is a canonical isomorphism
    \[
        c\ot(X\otc{A_1}M)
        \cong
        (c\ot X)\otc{A_1}M.
    \]

    Next, we show that $F_1(A_1)\in\hind{A_1}{\CC}{A_2}$.
    Since $F_1(A_1)\in\hind{}{\CC}{A_2}$, it remains to equip $F_1(A_1)$ with a compatible left $A_1$-module structure. This action is induced by the left $\CC$-module structure of $F_1$:
    \[
        A_1\ot F_1(A_1)
        \cong
        F_1(A_1\ot A_1)
        \xrightarrow{F_1(m)}
        F_1(A_1).
    \]
\end{proof}


\begin{thm}[Module realization functor]\label{thm:E1_module_realization}
    There is a fully faithful 2-functor from $\Mrt_{E_1}(\CC)$ to $\LMod(\CC)$.
\end{thm}
\begin{proof}
    Let $\mathrm{Mod}_1:\Mrt_{E_1}(\CC)\to\LMod(\CC)$ be the 2-functor 
    \begin{align*}
        \xymatrix@C=1pc@R=1.5pc{ 
            A_1 
            \ar@/^1.25pc/[rr]_{\quad}^{M}="M" 
            \ar@/_1.25pc/[rr]_{N}="N" 
            && A_2 
            \ar@{}"M";"N"|(0.135){\,}="7" 
            \ar@{}"M";"N"|(0.875){\,}="8" 
            \ar@{=>}|-{f}"7" ;"8"
        }\mapsto
        \xymatrix@C=1pc@R=1.5pc{ 
            \hind{}{\CC}{A_1} 
            \ar@/^1.25pc/[rr]_{\quad}^{- \otc{A_1} M}="M" 
            \ar@/_1.25pc/[rr]_{- \otc{A_1} N}="N" 
            && \hind{}{\CC}{A_2} 
            \ar@{}"M";"N"|(0.135){\,}="7" 
            \ar@{}"M";"N"|(0.875){\,}="8" 
            \ar@{=>}|-{- \otc{A_1} f}"7" ;"8"
        }
    \end{align*}
    By Theorem \ref{thm:generalized_Eilenberg-Watts}, $\mathrm{Mod}_1$ is fully faithful.
\end{proof}


Thus, by the above theorem, we obtain an equivalent definition of 1-Morita equivalence.

\begin{crl}
    Let $\CC$ be an $E_1$-monoidal 1-category.
    Two $E_1$-algebras $A_1,A_2$ are 1-Morita equivalent, i.e. they are equivalent objects in $\Mrt_{E_1}(\CC)$, if and only if they are classically 1-Morita equivalent (see Definition \ref{dfn:1-Morita}).
\end{crl}

The notion of 1-Morita equivalence in different contexts can be unified by Definition \ref{dfn: 1-Mrt_equivalnce}; the different cases arise from different choices of the $n$-category $\CC$.
For instance, when $\CC=\Cat$, the 2-category of 1-categories, $E_1$-algebras in $\Cat$ are monoidal categories, and Definition \ref{dfn:1-Morita} characterizes 1-Morita equivalence between monoidal categories.

In particular, let $\CC=\Ab$, the category of Abelian groups.
It is clear that $\Ab$ is an $E_{\infty}$-monoidal 1-category, i.e. a symmetric monoidal 1-category, and that an $E_1$-algebra in $\Ab$ is a ring.
For two rings $R,S$, we say that they are Morita equivalent if their module categories are equivalent, i.e.
\[
    \hind{}{\Ab}{R}\simeq \hind{}{\Ab}{S}
\]
as additive categories.
Equivalently, ordinary Morita equivalence can be described in terms of equivalences in the bicategory $\Mrt_{E_1}(\Ab)$.

\begin{crl}\label{crl:Morita_rings}
    Two rings $R$ and $S$ are 1-Morita equivalent if and only if they are equivalent objects in the bicategory $\Mrt_{E_1}(\Ab)$, i.e. there exist an $R$-$S$-bimodule $M$ and an $S$-$R$-bimodule $N$ such that
    \[
        M\otc{S}N\cong R
        \qquad\text{and}\qquad
        N\otc{R}M\cong S.
    \]
\end{crl}

This is a classical result in ordinary Morita theory, which studies categories of modules over rings.
We also recover the ordinary Eilenberg--Watts theorem.

\begin{crl}[Eilenberg--Watts Theorem]
    For $R,S\in\Alg_{E_1}(\Ab)$, the functor
    \begin{align*}
        \hind{R}{\Ab}{S}
        &\longrightarrow
        \Fun^{rex}_{\Ab}(\hind{}{\Ab}{R},\hind{}{\Ab}{S}),\\
        M
        &\longmapsto
        -\otc{R}M
    \end{align*}
is an equivalence of categories.
\end{crl}

\subsection{1-Morita Theory in Finite Tensor Categories}
\label{subsec:one-morita-finite}

We now restrict to the finite setting. The main result of this subsection is that, for a finite pre-multi-tensor category $\CC$, the module realization functor is not only fully faithful but an equivalence of 2-categories.

\begin{dfn}
    Let $\CM$ be a finite $\CC$-module category.
    An object $P\in\CM$ is called
    \begin{itemize}
        \item $\CC$-{\bf projective} if $[P,-]:\CM\to\CC$ is right exact;
        \item a $\CC$-{\bf generator} if $[P,-]:\CM\to\CC$ is faithful.
    \end{itemize}
\end{dfn}

\begin{lem}[\cite{DS26}, Lemma 2.14]
\label{lem:existence_of_progenerator}
    Any finite left $\CC$-module category admits a $\CC$-projective generator.
    The analogous statement for finite right $\CC$-module categories follows by passing to the monoidal opposite.
\end{lem}

\begin{thm}[\cite{DSPS19}, Theorem 2.24; \cite{DS26}, Theorem 2.13]
\label{thm:reconstruct_of_module_cat}
    Let $\CC$ be a finite pre-multi-tensor category, and let $\CM$ be a finite $\CC$-module category.
    If $P$ is a $\CC$-projective generator in $\CM$, then the functor
    \[
        [P,-]:\CM\to \CC_{[P,P]}
    \]
    is an equivalence of left $\CC$-module categories.
\end{thm}

The reconstruction theorem above implies that every finite $\CC$-module category is realized by modules over an $E_1$-algebra in $\CC$. Together with the fully faithful module realization functor established above, this gives the following equivalence.

\begin{thm}[Finite module realization]
\label{thm:finite_module_realization}
    Let $\CC$ be a finite pre-multi-tensor category.
    There is an equivalence of 2-categories
    \[
        \mathrm{Mod}_1:\Mrt_{E_1}(\CC)\longrightarrow \LMod(\CC)^{fin}_{\bk}.
    \]
\end{thm}

\begin{proof}
    Full faithfulness follows from Theorem \ref{thm:E1_module_realization}.
    It remains to prove essential surjectivity.

    Let $\CM$ be a finite left $\CC$-module category.
    By Lemma \ref{lem:existence_of_progenerator}, $\CM$ admits a $\CC$-projective generator $P$.
    Then Theorem \ref{thm:reconstruct_of_module_cat} gives an equivalence
    \[
        \CM\simeq \CC_{[P,P]}
    \]
    of left $\CC$-module categories.
    Since $[P,P]$ is an $E_1$-algebra in $\CC$, $\CM$ lies in the essential image of $\mathrm{Mod}_1$.
\end{proof}

We next record several consequences of the finite module realization and the Eilenberg--Watts equivalence. These results describe relative tensor products of finite module categories and will culminate in a fusion formula for bimodule categories.

\begin{thm}[\cite{KZ18}]
\label{thm:AC_CA_inverse}
    Let $\CC\in\Alg_{E_1}(\Cat^{rex})$, and let $A$ and $A'$ be two $E_1$-algebras in $\CC$.
    Then there are equivalences
    \[
        \hind{A}{\CC}{}\btc{\CC}\hind{}{\CC}{A'}
        \simeq \hind{A}{\CC}{A'}
        \simeq \Fun_{\CC}^{rex}(\hind{}{\CC}{A},\hind{}{\CC}{A'}).
    \]
\end{thm}

Let $\CB$ be a finite pre-multi-tensor category. The following two descriptions of $\hind{}{\CN}{A}$ will be useful below.

\begin{prp}
\label{prp:variant_Eilenberg_Watts}
    For any finite right $\CB$-module category $\CN$ and any $E_1$-algebra $A\in\Alg_{E_1}(\CB)$, there is an equivalence
    \[
        \Fun_{1\mid\CB}(\hind{A}{\CB}{},\CN)\simeq \hind{}{\CN}{A}.
    \]
\end{prp}

\begin{prp}
\label{prp:fusion_of_walls}
    Let $\CN$ be a finite right $\CB$-module category and let $A$ be an $E_1$-algebra in $\CB$.
    Then there is an equivalence
    \[
        \CN\btc{\CB}\hind{}{\CB}{A}\simeq \hind{}{\CN}{A}.
    \]
\end{prp}

\begin{proof}
    Any finite right $\CB$-module category can be written as $\hind{A'}{\CB}{}$ for some $E_1$-algebra $A'\in\Alg_{E_1}(\CB)$.
    The result then follows from Theorem \ref{thm:AC_CA_inverse}.
\end{proof}

Combining the preceding two propositions gives a tensor--Hom description of finite module categories.

\begin{crl}
\label{crl:N_btimes_BA_over_B}
    Let $\CN$ be a finite right $\CB$-module category and let $A$ be an $E_1$-algebra in $\CB$.
    Then
    \[
        \CN\btc{\CB}\hind{}{\CB}{A}
        \simeq
        \Fun_{1\mid\CB}(\hind{A}{\CB}{},\CN).
    \]
\end{crl}

\begin{proof}
    This follows from Proposition \ref{prp:variant_Eilenberg_Watts} and Proposition \ref{prp:fusion_of_walls}.
\end{proof}

We now apply these reconstruction results to the canonical module category $\hind{}{\CC}{A}$ over $\hind{A}{\CC}{A}$. This will allow us to identify $\hind{}{\CC}{A}$ and $\hind{A}{\CC}{}$ as relative inverses under an indecomposability assumption.

\begin{lem}[\cite{DS26}, Corollary 2.26]
    Let $A$ be an $E_1$-algebra in a finite multi-tensor category $\CC$.
    Then $\hind{A}{\CC}{A}$ is a pre-multi-tensor category.
\end{lem}

\begin{lem}[\cite{DS26}, Example 4.2]
\label{lem:CA_is_finite_right_ACA_module_cat}
    Let $A$ be an $E_1$-algebra in a finite multi-tensor category $\CC$.
    Then $\hind{}{\CC}{A}$ is a finite right $\hind{A}{\CC}{A}$-module category.
\end{lem}

\begin{crl}
\label{crl:reconstruct_of_CA_in_ACA}
    There is an equivalence of right $\hind{A}{\CC}{A}$-module categories
    \[
        \hind{}{\CC}{A}\simeq \hind{[P,P]}{(\hind{A}{\CC}{A})}{}
    \]
    for some $\hind{A}{\CC}{A}$-projective generator $P\in\hind{}{\CC}{A}$.
\end{crl}

\begin{proof}
    By Lemma \ref{lem:CA_is_finite_right_ACA_module_cat} and Lemma \ref{lem:existence_of_progenerator}, there exists a $\hind{A}{\CC}{A}$-projective generator $P$ in $\hind{}{\CC}{A}$.
    The result then follows from Theorem \ref{thm:reconstruct_of_module_cat}.
\end{proof}

We will also use the following double-centralizer-type equivalence.

\begin{thm}[\cite{DS26}, Theorem 2.31]
\label{thm:canonical_functor_is_equiv}
    Let $\CC$ be an indecomposable finite multi-tensor category and let $A$ be a nonzero $E_1$-algebra in $\CC$.
    Then the canonical functor
    \[
        \CC\longrightarrow
        \Fun_{1\mid \hind{A}{\CC}{A}}(\hind{}{\CC}{A},\hind{}{\CC}{A})
    \]
    is an equivalence.
\end{thm}

Combining this equivalence with the reconstruction of $\hind{}{\CC}{A}$ above, we obtain a corresponding description of the opposite bimodule category.

\begin{prp}
\label{prp:AC_is_equivalent_to_ACAPP}
    Let $\CC$ be an indecomposable finite multi-tensor category and let $A$ be a nonzero $E_1$-algebra in $\CC$.
    Suppose
    \[
        \hind{}{\CC}{A}\simeq \hind{[P,P]}{(\hind{A}{\CC}{A})}{}
    \]
    for some $P\in\hind{}{\CC}{A}$.
    Then
    \[
        \hind{A}{\CC}{}\simeq \hind{}{\hind{A}{\CC}{A}}{[P,P]}
    \]
    as $\hind{A}{\CC}{A}$-$\CC$-bimodule categories.
\end{prp}

\begin{proof}
    We have
    \begin{align*}
        \hind{A}{\CC}{}
        &\simeq \hind{A}{\CC}{}\btc{\CC}\CC
        \simeq \hind{A}{\CC}{}\btc{\CC}
        \Fun_{1\mid \hind{A}{\CC}{A}}
        (\hind{}{\CC}{A},\hind{}{\CC}{A})\\
        &\simeq \hind{A}{\CC}{}\btc{\CC}
        \left(
        \hind{}{\CC}{A}\btc{\hind{A}{\CC}{A}}
        \hind{}{\hind{A}{\CC}{A}}{[P,P]}
        \right)\\
        &\simeq
        \left(\hind{A}{\CC}{}\btc{\CC}\hind{}{\CC}{A}\right)
        \btc{\hind{A}{\CC}{A}}
        \hind{}{\hind{A}{\CC}{A}}{[P,P]}\\
        &\simeq
        \hind{A}{\CC}{A}\btc{\hind{A}{\CC}{A}}
        \hind{}{\hind{A}{\CC}{A}}{[P,P]}
        \simeq
        \hind{}{\hind{A}{\CC}{A}}{[P,P]}.
    \end{align*}
    Here the second equivalence follows from
    Theorem \ref{thm:canonical_functor_is_equiv},
    the third from Corollary \ref{crl:N_btimes_BA_over_B},
    and the fifth from Theorem \ref{thm:AC_CA_inverse}.
\end{proof}

As a consequence, the canonical bimodule categories $\hind{}{\CC}{A}$ and $\hind{A}{\CC}{}$ are relative inverses.

\begin{prp}
\label{prp:relative_inverse}
    Let $\CC$ be an indecomposable finite multi-tensor category and let $A$ be a nonzero $E_1$-algebra in $\CC$.
    Then there is an equivalence
    \[
        \hind{}{\CC}{A}\btc{\hind{A}{\CC}{A}}\hind{A}{\CC}{}\simeq \CC.
    \]
\end{prp}

\begin{proof}
    Let $P$ be a $\hind{A}{\CC}{A}$-projective generator of $\hind{}{\CC}{A}$ as in Corollary \ref{crl:reconstruct_of_CA_in_ACA}.
    We have
    \begin{align*}
        \hind{}{\CC}{A}\btc{\hind{A}{\CC}{A}}\hind{A}{\CC}{}
        &\simeq
        \hind{}{\CC}{A}\btc{\hind{A}{\CC}{A}}
        \hind{}{\hind{A}{\CC}{A}}{[P,P]}\\
        &\simeq
        \Fun_{1\mid \hind{A}{\CC}{A}}
        \left(
        \hind{[P,P]}{\hind{A}{\CC}{A}}{},
        \hind{}{\CC}{A}
        \right)\\
        &\simeq
        \Fun_{1\mid \hind{A}{\CC}{A}}
        (\hind{}{\CC}{A},\hind{}{\CC}{A})
        \simeq \CC.
    \end{align*}
    The first equivalence follows from Proposition \ref{prp:AC_is_equivalent_to_ACAPP}, the second from Corollary \ref{crl:N_btimes_BA_over_B}, the third from Corollary \ref{crl:reconstruct_of_CA_in_ACA}, and the final equivalence from Theorem \ref{thm:canonical_functor_is_equiv}.
\end{proof}

\begin{rmk}
    The indecomposability assumption is necessary.
    For example, let
    \[
        \CC:=\vect_{\bk}\boxplus\vect_{\bk},
        \qquad A:=(\bk,0).
    \]
    Then
    \[
        \hind{}{\CC}{A}\simeq \vect_{\bk}\simeq
        \hind{A}{\CC}{A}\simeq \hind{A}{\CC}{},
    \]
    and hence
    \[
        \hind{}{\CC}{A}\btd_{\hind{A}{\CC}{A}}\hind{A}{\CC}{}
        \simeq \vect_{\bk}\nsimeq \CC.
    \]
\end{rmk}

The relative inverse property immediately gives the following fusion formula for bimodule categories.

\begin{thm}
\label{thm:fusion_of_1d_phase}
    Let $\CC$ be an indecomposable finite multi-tensor category, and let $A_1,A_2,A_3$ be nonzero $E_1$-algebras in $\CC$.
    Then there is an equivalence
    \[
        \hind{A_1}{\CC}{A_2}
        \btc{\hind{A_2}{\CC}{A_2}}
        \hind{A_2}{\CC}{A_3}
        \simeq
        \hind{A_1}{\CC}{A_3}
    \]
    of categories.
\end{thm}

\begin{proof}
    We have
    \begin{align*}
        \hind{A_1}{\CC}{A_2}
        \btc{\hind{A_2}{\CC}{A_2}}
        \hind{A_2}{\CC}{A_3}
        &\simeq
        \left(\hind{A_1}{\CC}{}\btc{\CC}\hind{}{\CC}{A_2}\right)
        \btc{\hind{A_2}{\CC}{A_2}}
        \left(\hind{A_2}{\CC}{}\btc{\CC}\hind{}{\CC}{A_3}\right)\\
        &\simeq
        \hind{A_1}{\CC}{}\btc{\CC}\CC\btc{\CC}\hind{}{\CC}{A_3}
        \simeq
        \hind{A_1}{\CC}{}\btc{\CC}\hind{}{\CC}{A_3}
        \simeq
        \hind{A_1}{\CC}{A_3}.
    \end{align*}
    Here the first and last equivalences follow from Theorem \ref{thm:AC_CA_inverse}, and the second follows from Proposition \ref{prp:relative_inverse}.
\end{proof}

\subsection{Local 1-Morita Equivalence and Condensation}
\label{subsec:local-one-morita}

Theorem \ref{thm:generalized_Eilenberg-Watts} also leads to another characterization of 1-Morita equivalence.
First, we note that
\[
    \Omega_{A}(\Mrt_{E_1}(\CC))=\hind{A}{\CC}{A},
\]
where $\Omega_x(\CC):=\hom_{\CC}(x,x)$ for any $x\in\CC$.

\begin{prp}\label{prp:1-Morita_implies_locally_1-Morita}
    Let $A_1$ and $A_2$ be two $E_1$-algebras in $\CC$.
    If $A_1\widesim[3]{1-Morita}{} A_2$, then
    $\hind{A_1}{\CC}{A_1}$ and $\hind{A_2}{\CC}{A_2}$ are equivalent as $E_1$-monoidal categories.
\end{prp}

\begin{proof}
    By Definition \ref{dfn:1-Morita}, we have
    $\hind{}{\CC}{A_1}\simeq \hind{}{\CC}{A_2}$.
    By Theorem \ref{thm:generalized_Eilenberg-Watts},
    \[
        \hind{A_1}{\CC}{A_1}
        \simeq \Fun_{\CC}(\hind{}{\CC}{A_1},\hind{}{\CC}{A_1})
        \simeq \Fun_{\CC}(\hind{}{\CC}{A_2},\hind{}{\CC}{A_2})
        \simeq \hind{A_2}{\CC}{A_2}.
    \]
\end{proof}

A $E_1$-local $E_0$-module over an $E_1$-algebra $A$ is an $A$-$A$-bimodule, i.e. an object of $\hind{A}{\CC}{A}$.

\begin{dfn}
    Two $E_1$-algebras $A_1$ and $A_2$ are {\bf locally 1-Morita equivalent} if their $E_0$-module categories
    $\hind{A_1}{\CC}{A_1}$ and $\hind{A_2}{\CC}{A_2}$ are equivalent as $E_1$-monoidal categories.
\end{dfn}

\begin{thm}\label{thm:C_ACA_1-Morita_equivalent}
    Let $\CC$ be an indecomposable finite multi-tensor category, and let $A$ be a nonzero $E_1$-algebra in $\CC$.
    Then $\CC$ is 1-Morita equivalent to $\hind{A}{\CC}{A}$; equivalently,
    $\hind{}{\CC}{A}$ is an invertible
    $\CC$-$\hind{A}{\CC}{A}$-bimodule category.
\end{thm}

\begin{proof}
    This follows from Theorem \ref{thm:AC_CA_inverse} and Proposition \ref{prp:relative_inverse}.
\end{proof}

\begin{dfn}
    An $E_1$-monoidal category $\CC$ is called {\bf non-degenerate} if
    $\fZ_1(\CC)\simeq 1$.
\end{dfn}

\begin{thm}
    Let $\CC$ be a non-degenerate indecomposable finite multi-tensor category over an algebraically closed field $\bk$.
    Then two $E_1$-algebras $A_1$ and $A_2$ in $\CC$ are 1-Morita equivalent if and only if they are locally 1-Morita equivalent.
\end{thm}

\begin{proof}
    If one of the two algebras $A_1$ and $A_2$ is zero, then the statement holds trivially.
    Now suppose that both $A_1$ and $A_2$ are nonzero.

    $\Rightarrow$: This follows from Proposition \ref{prp:1-Morita_implies_locally_1-Morita}.

    $\Leftarrow$: By Theorem \ref{thm:C_ACA_1-Morita_equivalent},
    $\hind{}{\CC}{A_1}$ and $\hind{}{\CC}{A_2}$ are invertible bimodule categories.
    Since
    $\hind{A_1}{\CC}{A_1}\simeq \hind{A_2}{\CC}{A_2}$,
    we can transfer the right $\hind{A_2}{\CC}{A_2}$-action on
    $\hind{}{\CC}{A_2}$ to $\hind{}{\CC}{A_1}$.
    Thus $\hind{}{\CC}{A_1}$ becomes an invertible
    $\CC$-$\hind{A_2}{\CC}{A_2}$-bimodule category.
    Hence
    \[
        \hind{}{\CC}{A_1}^{op}:=\hind{A_1}{\CC}{}
    \]
    is a $\hind{A_2}{\CC}{A_2}$-$\CC$-bimodule category.

    Define the $\CC$-$\CC$-bimodule category
    \[
        \CP:=
        \hind{}{\CC}{A_2}
        \btd_{\hind{A_2}{\CC}{A_2}}
        \hind{}{\CC}{A_1}^{op}.
    \]
    Since $\hind{}{\CC}{A_2}$ and $\hind{}{\CC}{A_1}^{op}$ are both invertible bimodule categories, so is $\CP$.

    By \cite{ENO10}, the Brauer--Picard group satisfies
    \[
        \mathrm{BrPic}(\CC)\cong \Aut_{E_2}(\fZ_1(\CC)),
    \]
    i.e. invertible $\CC$-$\CC$-bimodule categories are classified by
    $E_2$-autoequivalences of the center $\fZ_1(\CC)$.
    Since $\CC$ is non-degenerate, $\fZ_1(\CC)\simeq 1$, and hence
    \[
        \Aut_{E_2}(\fZ_1(\CC))=\{*\}.
    \]
    Therefore the only invertible $\CC$-$\CC$-bimodule category is $\CC$ itself, so
    $\CP\simeq\CC$ as $\CC$-$\CC$-bimodule categories.

    Consequently,
    \begin{align*}
        \hind{}{\CC}{A_1}
        &\simeq \CP\btc{\CC}\hind{}{\CC}{A_1}
        \simeq
        \left(
        \hind{}{\CC}{A_2}
        \btc{\hind{A_2}{\CC}{A_2}}
        \hind{}{\CC}{A_1}^{op}
        \right)
        \btc{\CC}\hind{}{\CC}{A_1}\\
        &\simeq
        \hind{}{\CC}{A_2}
        \btc{\hind{A_1}{\CC}{A_1}}
        \left(
        \hind{A_1}{\CC}{}
        \btc{\CC}
        \hind{}{\CC}{A_1}
        \right)
        \simeq
        \hind{}{\CC}{A_2}
        \btc{\hind{A_1}{\CC}{A_1}}
        \hind{A_1}{\CC}{A_1}
        \simeq
        \hind{}{\CC}{A_2}
    \end{align*}
    as left $\CC$-module categories.
\end{proof}

\vspace{3em}
 
\section{The 2-Morita Category and Module Realization}\label{section:2-Morita}
This section develops a bi-bimodule model for 2-Morita theory. We introduce the model in \cref{subsec:two-morita-bibimodule-model}, relate it to the iterated-bimodule description in \cref{subsec:two-morita-iterated-bimodules}, construct the module realization functor in \cref{sec:functor_to_module_cat}, apply the framework to the fusion of condensation defects in \cref{subsec:condensation-defects-fusion}, and finally discuss local 2-Morita equivalence in \cref{subsec:local-two-morita}.

\subsection{A Bi-Bimodule Model of the 2-Morita Category}
\label{subsec:two-morita-bibimodule-model}

Let $\CC$ be an $E_2$-monoidal $n$-category.
when $n=1$, we use $\beta$ to denote the braiding in $\CC$.
Let $B_1$ and $B_2$ be two $E_2$-algebras in $\CC$.
We begin by describing the 1- and 2-morphism data that will enter our bi-bimodule model of the 2-Morita category.

Consider the category $\hind{B_1}{\CC}{B_2}$ of $B_1$-$B_2$-bimodules in $\CC$.
The object $B_1\ot B_2\in\hind{B_1}{\CC}{B_2}$ naturally carries an algebra structure in $\CC$.
Its multiplication is given by
\begin{align*}
    B_1\ot B_2\ot B_1\ot B_2
    \xrightarrow{\id_{B_1}\ot\beta_{B_2,B_1}\ot\id_{B_2}}
    B_1\ot B_1\ot B_2\ot B_2
    \xrightarrow{m_1\ot m_2}
    B_1\ot B_2.
\end{align*}

\begin{figure}[H]
    \centering
    \subcaptionbox{}{
        \begin{minipage}[t]{0.35\linewidth}
            \centering
            \begin{tikzpicture}[scale=1.3,line cap=round,line join=round]
                \filldraw[fill=gray!40,draw=white] (0,1) rectangle (3,-1);
                \draw[thick] (0,0)--(3,0);
                \node at (1.5,0.5){$B_1$};
                \node at (1.5,-0.5){$B_2$};
                \node at (0.25,0.22){$A$};
                \fill[black] (1,0.75) circle (0.054);
                \node [scale=0.6,left] at (1,0.75){$b'$};
                \fill[black] (2,0.75) circle (0.054);
                \node [scale=0.6,left] at (2,0.75){$b$};
                \fill[black] (1,0) circle (0.054);
                \node [scale=0.6,below] at (1,0){$a'$};
                \fill[black] (2,0) circle (0.054);
                \node [scale=0.6,below] at (2,0){$a$};
            \end{tikzpicture}
        \end{minipage}
    }
    \begin{tikzpicture}[line cap=round,line join=round]
        \node at(0,1.5){$u\ot u$};
        \draw[thick,-latex] (-0.5,1.3) -- (0.6,1.3);
        \draw[draw=none] (-0.5,0) -- (0.5,0.5);
    \end{tikzpicture}
    \subcaptionbox{}{
        \begin{minipage}[t]{0.35\linewidth}
            \centering
            \begin{tikzpicture}[scale=1.3,line cap=round,line join=round]
                \filldraw[fill=gray!40,draw=white] (0,1) rectangle (3,-1);
                \draw[thick] (0,0)--(3,0);
                \node at (1.5,0.5){$B_1$};
                \node at (1.5,-0.5){$B_2$};
                \node at (0.25,0.22){$A$};
                \fill[black] (1,0) circle (0.054);
                \node [scale=0.6,below] at (1,0){$b' \rhd a'$};
                \fill[black] (2,0) circle (0.054);
                \node [scale=0.6,below] at (2,0){$b \rhd a$};
            \end{tikzpicture}
        \end{minipage}
    }\\
    \begin{minipage}[t]{0.35\linewidth}
        \centering
        \begin{tikzpicture}
            \path[use as bounding box] (-1.4,-0.5) rectangle (1.4,0.5);
            \node[anchor=east] at(-0.24,0){$m_{B_1\ot A}$};
            \draw[thick,-latex,line cap=round] (-0.16,0.5) -- (-0.16,-0.5);
        \end{tikzpicture}
    \end{minipage}
    \begin{tikzpicture}
        \path[use as bounding box] (-0.5,-0.5) rectangle (0.6,0.5);
    \end{tikzpicture}
    \begin{minipage}[t]{0.35\linewidth}
        \centering
        \begin{tikzpicture}
            \path[use as bounding box] (-1.4,-0.5) rectangle (1.4,0.5);
            \draw[thick,-latex,line cap=round] (0.03,0.5) -- (0.03,-0.5);
            \node[anchor=west] at(0.11,0){$m_{A}$};
        \end{tikzpicture}
    \end{minipage}\\
    \subcaptionbox{}{
        \begin{minipage}[t]{0.35\linewidth}
            \centering
            \begin{tikzpicture}[scale=1.3,line cap=round,line join=round]
                \filldraw[fill=gray!40,draw=white] (0,1) rectangle (3,-1);
                \draw[thick] (0,0)--(3,0);
                \node at (1.5,0.5){$B_1$};
                \node at (1.5,-0.5){$B_2$};
                \node at (0.25,0.22){$A$};
                \fill[black] (1.5,0.75) circle (0.054);
                \node [scale=0.6,left] at (1.5,0.75){$bb'$};
                \fill[black] (1.5,0) circle (0.054);
                \node [scale=0.6,below] at (1.5,0){$aa'$};
            \end{tikzpicture}
        \end{minipage}
    }
    \begin{tikzpicture}[line cap=round,line join=round]
        \node at(0,1.5){$u$};
        \draw[thick,-latex] (-0.5,1.3) -- (0.6,1.3);
        \draw[draw=none] (-0.5,0) -- (0.5,0.5);
    \end{tikzpicture}
    \subcaptionbox{}{
        \begin{minipage}[t]{0.35\linewidth}
            \centering
            \begin{tikzpicture}[scale=1.3,line cap=round,line join=round]
                \filldraw[fill=gray!40,draw=white] (0,1) rectangle (3,-1);
                \draw[thick] (0,0)--(3,0);
                \node at (1.5,0.5){$B_1$};
                \node at (1.5,-0.5){$B_2$};
                \node at (0.25,0.22){$A$};
                \fill[black] (1.5,0) circle (0.054);
                \node [scale=0.6,above] at (1.5,0){$bb' \rhd aa'$};
                \node [scale=0.6,below] at (1.5,0){$(b \rhd a)(b' \rhd a')$};
            \end{tikzpicture}
        \end{minipage}
    }
    \caption{Elements in the algebras and their fusion behavior.}
\end{figure}
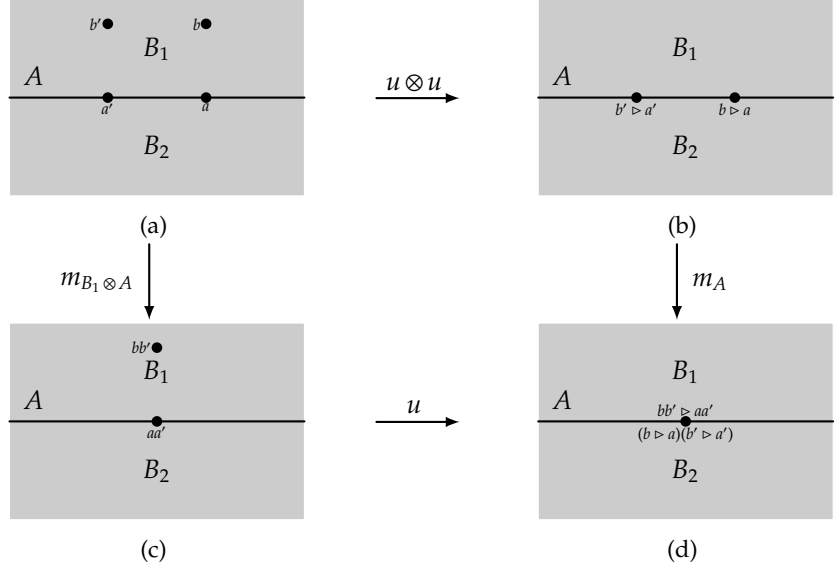

\begin{dfn}
    Let $B_1$ and $B_2$ be two $E_2$-algebras in $\CC$.
    An {\bf $E_1$ $B_1$-$B_2$-bimodule} is an $E_1$-algebra
    $(A,m_A,h_A)$ equipped with a left $B_1$-action
    $u:B_1\ot A\to A$ and a right $B_2$-action
    $d:A\ot B_2\to A$ satisfying
    \begin{align*}
        \xymatrix{
            B_1\ot A \ot B_1\ot A\ar[r]^{u\ot u}\ar[d]_{\id\ot\beta_{A,B_1}\ot\id}
            &A\ot A \ar[dd]^{m_A}
            &A\ot B_2\ot A\ot B_2\ar[l]_{d\ot d} \ar[d]^{\id\ot\beta_{B_2,A}\ot\id} \\
            B_1\ot B_1 \ot A\ot A\ar[d]_{m_{B_1}\ot m_A}
            &&
            A\ot A\ot B_2\ot B_2\ar[d]^{m_A\ot m_{B_2}} \\
            B_1\ot A\ar[r]_{u}
            &A
            &A\ot B_2\ar[l]^{d}
        }
    \end{align*}
\end{dfn}

The compatibility condition implies that $u$ and $d$ are $E_1$-algebra homomorphisms.
Since an $E_2$-algebra in $\CC$ can equivalently be viewed as an $E_1$-algebra in $\Alg_{E_1}(\CC)$, an $E_1$ $B_1$-$B_2$-bimodule is equivalently an object of
\[
    \hind{B_1}{\Alg_{E_1}(\CC)}{B_2}.
\]
This gives the 1-morphism data in the bi-bimodule model.
In particular, objects in $\hind{B}{\Alg_{E_1}(\CC)}{B}$ are called {\bf $E_1$-local $E_1$ $B$-modules}.

We next describe the corresponding 2-morphisms.
\begin{dfn}\label{dfn:bibimod}
    Let $(A_1,m_1,h_1,u_1,d_1)$ and $(A_2,m_2,h_2,u_2,d_2)$ be two
    $E_1$ $B_1$-$B_2$-bimodules.
    A {\bf $B_1$-$B_2$-bi-$A_1$-$A_2$-bimodule $E_0$-algebra}
    is an object $(M,l,r)\in\hind{A_1}{\CC}{A_2}$ such that the following diagram commutes:
    \begin{align}\label{eq:coherence_of_bibimod}
        \xymatrix@C=5.4em@R=2.8em{
            B_1\ot A_1\ot M \ar[dd]_{u_1\ot\id_M}
            &B_1\ot M \ar[l]_{\id_{B_1}\ot h_1\ot\id_M}
            \ar[r]^{\id_{B_1}\ot h_2\ot\id_M}
            &B_1\ot A_2\ot M \ar[d]^{u_2\ot\id_M}\\
            &&A_2\ot M\ar[d]^{\beta^{-1}_{M,A_2}}\\
            A_1\ot M \ar[r]_{l}
            &M
            &M\ot A_2\ar[l]^{r}\\
            M\ot A_1\ar[u]^{\beta^{-1}_{A_1,M}}
            &&\\
            M\ot A_1\ot B_1\ar[u]^{\id_M\ot d_1}
            &M\ot B_2 \ar[l]_{\id_M\ot h_1\ot\id_{B_2}}
            \ar[r]^{\id_M\ot h_2\ot\id_{B_2}}
            &M\ot A_2\ot B_2\ar[uu]_{\id_M\ot d_2}
        }
    \end{align}
\end{dfn}

Thus $M$ can be viewed as an $A_1$-$A_2$-bimodule whose induced
$B_1$- and $B_2$-actions are required to agree with the
$B_1$-$B_2$-actions carried by $A_1$ and $A_2$.
Equivalently, the different ways in which $B_1$ and $B_2$ act on $M$
through $A_1$ and $A_2$ must coincide.

To make these induced actions explicit, we first record the following observation.

\begin{prp}
    Let $A=(A,m_A,h_A,u,d)\in
    \hind{B_1}{\Alg_{E_1}(\CC)}{B_2}$.
    Then
    \[
        i_1:=u\circ(\id_{B_1}\ot h_A):B_1\to A,
        \qquad
        i_2:=d\circ(h_A\ot\id_{B_2}):B_2\to A
    \]
    are $E_1$-algebra homomorphisms.
\end{prp}

\begin{proof}
    Since $\Alg_{E_1}(\CC)$ is $E_1$-monoidal, the tensor product of two morphisms in
    $\Alg_{E_1}(\CC)$ is again a morphism in $\Alg_{E_1}(\CC)$.
    Hence $\id_{B_1}\ot h_A$ and $h_A\ot\id_{B_2}$ are $E_1$-algebra homomorphisms.
    The action maps $u$ and $d$ are also $E_1$-algebra homomorphisms, and therefore so are their compositions.
\end{proof}

\begin{crl}
    The algebra homomorphisms $i_1$ and $i_2$ induce restriction functors
    \begin{align*}
        i_1^{*,l}: \hind{A}{\CC}{}&\to\hind{B_1}{\CC}{},&
        (M,l)&\mapsto (M,l\circ(i_1\ot\id_M)),
    \end{align*}
    together with
    \[
        i_1^{*,r}:\hind{}{\CC}{A}\to\hind{}{\CC}{B_1},
        \qquad
        i_2^{*,l}:\hind{A}{\CC}{}\to\hind{B_2}{\CC}{},
        \qquad
        i_2^{*,r}:\hind{}{\CC}{A}\to\hind{}{\CC}{B_2}.
    \]
\end{crl}

\begin{notation}
    We write
    \begin{align*}
        u^l_M&:=l\circ(i_1\ot\id_M),&
        u^{r,-}_M&:=r\circ(\id_M\ot i_1)\circ\beta^{-1}_{M,B_1},\\
        d^r_M&:=r\circ(\id_M\ot i_2),&
        d^{l,-}_M&:=l\circ(i_2\ot\id_M)\circ\beta^{-1}_{B_2,M}.
    \end{align*}
\end{notation}

\begin{rmk}
    For instance, $(M,u^{r,-}_M)$ can be obtained from $(M,r)$ by the following sequence of restriction and induction functors \cite[Definition~2.21]{FFRS06}:
    \begin{align*}
        \xymatrix{
            \hind{}{\CC}{A}
                \ar[r]^{\alpha^-}
                \ar[d]_{i_1^*}
            &\hind{A}{\CC}{}
                \ar[d]^{i_1^*}
            \\
            \hind{}{\CC}{B_1}
                \ar[r]_{\alpha^-}
            &\hind{B_1}{\CC}{}
        }
    \end{align*}
\end{rmk}

It is routine to check that
\begin{prp}\label{prp:u^l_d^r_are_compatible}
    The pairs $(u^l_M,d^r_M)$ and $(u^{r,-}_M,d^{l,-}_M)$ define compatible
    $B_1$-$B_2$-module actions.
\end{prp}

\begin{prp}\label{prp:u^r-_d^r_are_compatible}
    The pairs $(u^{r,-}_M,d^r_M)$ and $(u^l_M,d^{l,-}_M)$ define compatible
    $B_1$-$B_2$-module actions.
\end{prp}

\begin{proof}
    We prove the statement for the first pair $(u^{r,-}_M,d^r_M)$.
    The multiplicativity of $d:A\ot B_2\to A$, after inserting the unit morphisms, gives
    \begin{align*}
        m_A\circ(i_2\ot\id_A)
        =
        m_A\circ(\id_A\ot i_2)\circ\beta_{B_2,A}.
    \end{align*}
    Precomposing this equality with $\id_{B_2}\ot i_1$ gives
    \begin{align*}
        m_A\circ(i_2\ot i_1)
        =
        m_A\circ(i_1\ot i_2)\circ\beta_{B_2,B_1},
    \end{align*}
    or equivalently
    \begin{align*}
        m_A\circ(i_1\ot i_2)
        =
        m_A\circ(i_2\ot i_1)\circ\beta^{-1}_{B_2,B_1}.
    \end{align*}
    We can therefore compute, as morphisms $B_1\ot M\ot B_2\to M$,
    \begin{align*}
        d^r_M\circ(u_M^{r,-}\ot\id_{B_2})
        &=
        r\circ(r\ot\id_A)
        \circ(\id_M\ot i_1\ot i_2)
        \circ(\beta^{-1}_{M,B_1}\ot\id_{B_2})
        &&\text{by definition}\\
        &=
        r\circ(\id_M\ot m_A)
        \circ(\id_M\ot i_1\ot i_2)
        \circ(\beta^{-1}_{M,B_1}\ot\id_{B_2})
        &&\text{by $r$-associativity}\\
        &=
        r\circ(\id_M\ot m_A)
        \circ(\id_M\ot i_2\ot i_1)
        \circ(\id_M\ot\beta^{-1}_{B_1,B_2})
        \circ(\beta^{-1}_{M,B_1}\ot\id_{B_2})
        &&\text{by negative commutation}\\
        &=
        r\circ(\id_M\ot m_A)
        \circ(\id_M\ot i_2\ot i_1)
        \circ\beta^{-1}_{M\ot B_2,B_1}
        &&\text{by the hexagon identity}\\
        &=
        r\circ(\id_M\ot i_1)
        \circ(d^r_M\ot\id_{B_1})
        \circ\beta^{-1}_{M\ot B_2,B_1}
        &&\text{by $r$-associativity}\\
        &=
        r\circ(\id_M\ot i_1)
        \circ\beta^{-1}_{M,B_1}
        \circ(\id_{B_1}\ot d^r_M)
        &&\text{by naturality}\\
        &=
        u_M^{r,-}\circ(\id_{B_1}\ot d^r_M).
    \end{align*}
    Thus $(M,u_M^{r,-},d^r_M)$ is a $B_1$-$B_2$-bimodule.
\end{proof}

\begin{rmk}
    For comparison, consider
    \[
        u^{r,+}_M:=r\circ(\id_M\ot i_1)\circ\beta_{B_1,M}.
    \]
    The analogous positive-crossing computation shows that
    $u^{r,+}_M$ and $d^r_M$ are compatible if and only if
    \begin{align*}
        m_A\circ(i_1\ot i_2)
        =
        m_A\circ(i_1\ot i_2)
        \circ\beta_{B_2,B_1}\circ\beta_{B_1,B_2}.
    \end{align*}
    This is precisely the mutual-locality condition for the two images.
    It is automatic in a symmetric monoidal category, but not in a general braided category.
    Consequently, over a general braided $\CC$, the $B_1$-action induced from a right
    $A$-module in the defining bi-bimodule coherence diagram must use the negative crossing
    $\beta^{-1}_{M,A}$.
    If one instead uses the positive crossing $\beta_{A,M}$, the mutual-locality condition above must be imposed as an additional hypothesis.
\end{rmk}

We now package these objects and morphisms into a category.

\begin{dfn}
    Let $M_1$ and $M_2$ be two $E_0$
    $B_1$-$B_2$-bi-$A_1$-$A_2$-bimodules.
    An {\bf $E_0$ $B_1$-$B_2$-bi-$A_1$-$A_2$-bimodule homomorphism}
    from $M_1$ to $M_2$ is an $A_1$-$A_2$-bimodule homomorphism.
\end{dfn}

\begin{dfn}\label{dfn:cat_of_bibimod}
    The category
    \[
        \hind{A_1}{\vind{B_1}{\Alg_{E_0}(\CC)}{B_2}}{A_2},
    \]
    or simply $\hind{A_1}{\vind{B_1}{\CC}{B_2}}{A_2}$, is defined as follows:
    \begin{itemize}
        \item its objects are $E_0$ $B_1$-$B_2$-bi-$A_1$-$A_2$-bimodules;
        \item its morphisms are $B_1$-$B_2$-bi-$A_1$-$A_2$-bimodule homomorphisms.
    \end{itemize}
\end{dfn}

The category $\hind{A_1}{\vind{B_1}{\CC}{B_2}}{A_2}$ is therefore a full subcategory of
$\hind{A_1}{\CC}{A_2}$.
The underlying $A_1$-$A_2$-bimodule is obtained canonically by restricting along the unit maps
$h_{B_1}:1\to B_1$ and $h_{B_2}:1\to B_2$.
More generally, 
\begin{prp}
    for any $E_2$-algebra homomorphisms $f_1:B_1\to B_3$ and $f_2:B_2\to B_4$, there is a fully faithful functor
    \begin{align*}
        (f_1,f_2)^*:
        \hind{A_1}{\vind{B_3}{\CC}{B_4}}{A_2}
        \longrightarrow
        \hind{A_1^*}{\vind{B_1}{\CC}{B_2}}{A_2^*}
    \end{align*}
    where $A_i^*$ is the image of $A_i$ under the pull-back functor $(f_1,f_2)^*:\hind{B_3}{\Alg_{E_1}(\CC)}{B_4}\to \hind{B_1}{\Alg_{E_1}(\CC)}{B_2}$.
\end{prp}
Equivalently, $\hind{A_1}{\vind{B_3}{\CC}{B_4}}{A_2}$ is a full subcategory of $\hind{A_1^*}{\vind{B_1}{\CC}{B_2}}{A_2^*}$.

A useful special case recovers the usual notion of a local module.
Suppose $B_1=B_2=B$ and choose $A_1=A_2=B$.
Then the coherence diagram \ref{eq:coherence_of_bibimod} reduces to
\begin{align*}
    \xymatrix{
         &B\ot M \ar[dl]_{\id} \ar[dr]^{\beta^{-1}_{M,B}} & \\
        B\ot M \ar[r]_l &M &M\ot B \ar[l]_{r}\\
         &M\ot B \ar[ul]^{\beta^{-1}_{B,M}} \ar[ur]_{\id} &
    }
\end{align*}
and an $E_0$ $B$-$B$-bi-$B$-$B$-bimodule is precisely a local $B$-module
in the usual sense \cite[Definition~3.15 and Proposition~3.17]{FFRS06}.

\begin{crl}
    There is an equivalence of categories
    \[
        \hind{B}{\vind{B}{\CC}{B}}{B}
        \simeq
        \hind{}{\CC}{B}^{loc}.
    \]
\end{crl}

\begin{dfn}
    An $E_0$ $B$-$B$-bi-$B$-$B$-bimodule is called an {\bf $E_2$-local $E_0$ $B$-module}.
\end{dfn}

\vspace{1.5em}

We are now ready to assemble these data into the 2-Morita 3-category.

\begin{dfn}\label{dfn: 2-Morita_cat}
    Let $\CC$ be an $E_2$-monoidal $1$-category.
    The 3 category $\Mrt_{E_2}(\CC)$ consists of the following data:
    \begin{itemize}
        \item the objects are $E_2$-algebras in $\CC$, i.e. objects of
        $\Alg_{E_2}(\CC)$;

        \item a $1$-morphism from $B_1$ to $B_2$ is an $E_1$
        $B_1$-$B_2$-bimodule, i.e. an object of
        $\hind{B_1}{\Alg_{E_1}(\CC)}{B_2}$;

        \item a $2$-morphism between two $E_1$
        $B_1$-$B_2$-bimodules $A_1$ and $A_2$ is an $E_0$
        $B_1$-$B_2$-bi-$A_1$-$A_2$-bimodule, i.e. an object of
        $\hind{A_1}{\vind{B_1}{\CC}{B_2}}{A_2}$;

        \item a $3$-morphism between two such $E_0$ bi-bimodules is an
        $A_1$-$A_2$-bimodule homomorphism;

        \item the identity $1$-morphism on an $E_2$-algebra $B$ is $B$
        itself;

        \item the identity $2$-morphism on an $E_1$
        $B_1$-$B_2$-bimodule $A$ is $A$ itself;

        \item the identity $3$-morphisms are the ordinary identity
        morphisms;

        \item the composition of $1$-morphisms is given by the relative
        tensor product:
        \begin{align*}
            \hind{B_1}{\Alg_{E_1}(\CC)}{B_2}
            \times
            \hind{B_2}{\Alg_{E_1}(\CC)}{B_3}
            &\longrightarrow
            \hind{B_1}{\Alg_{E_1}(\CC)}{B_3},\\
            (A,A')
            &\longmapsto
            \vfus{A}{\otc{B_2}}{A'};
        \end{align*}

        \item the composition of $2$-morphisms along $B_2$ is likewise
        given by the relative tensor product. For an $E_0$
        $B_1$-$B_2$-bi-$A_1$-$A_2$-bimodule $M$ and an $E_0$
        $B_2$-$B_3$-bi-$A_3$-$A_4$-bimodule $M'$,
        \begin{align*}
            \vfus{\hind{A_1}{\vind{B_1}{\CC}{B_2}}{A_2}}
            {\times}
            {\hind{A_3}{\vind{B_2}{\CC}{B_3}}{A_4}}
            &\longrightarrow
            \hind{\vfus{A_1}{\otc{B_2}}{A_3}}
            {\vind{B_1}{\CC}{B_3}}
            {\vfus{A_2}{\otc{B_2}}{A_4}},\\
            (M,M')
            &\longmapsto
            \vfus{M}{\otc{B_2}}{M'};
        \end{align*}

        \item the composition of $2$-morphisms along $A_2$ is also given
        by the relative tensor product. For an $E_0$
        $B_1$-$B_2$-bi-$A_1$-$A_2$-bimodule $M_1$ and an $E_0$
        $B_1$-$B_2$-bi-$A_2$-$A_3$-bimodule $M_2$,
        \begin{align*}
            \hind{A_1}{\vind{B_1}{\CC}{B_2}}{A_2}
            \times
            \hind{A_2}{\vind{B_1}{\CC}{B_2}}{A_3}
            &\longrightarrow
            \hind{A_1}{\vind{B_1}{\CC}{B_2}}{A_3},\\
            (M_1,M_2)
            &\longmapsto
            M_1\otc{A_2}M_2;
        \end{align*}

        \item the corresponding compositions of $3$-morphisms are induced
        by the relative tensor products above, while the ordinary
        composition of $3$-morphisms is given by composition in $\CC$.
    \end{itemize}
\end{dfn}

The two compositions of 2-morphisms satisfy the expected interchange compatibility, illustrated in Figure \ref{fig:Morita_E2}.

\begin{figure}[H]
    \centering
    \newcommand{\diagramotc}[1]{%
        \mathbin{\vcenter{\hbox{%
            \tikz[baseline=(diagramotc.base)]{%
                \node[
                    draw,
                    circle,
                    line width=0.35pt,
                    inner sep=0.15pt,
                    minimum size=1.15em,
                    font=\tiny
                ] (diagramotc) {$#1$};%
            }%
        }}}%
    }
    \subcaptionbox{}{\resizebox{0.265\textwidth}{!}{%
        \begin{tikzpicture}[
            line cap=round,
            line join=round,
            every node/.style={font=\scriptsize,xshift=0.08cm}]
            \path[use as bounding box] (-2.30,-1.65) rectangle (2.30,1.65);
            \filldraw[fill=gray!25,draw=none] (-2,-1.5) rectangle (2,1.5);
            \draw[line width=0.6pt] (-2,0.5) -- (2,0.5);
            \draw[line width=0.6pt] (-2,-0.5) -- (2,-0.5);
            \fill[black] (-1,0.5) circle (0.055);
            \fill[black] (1,0.5) circle (0.055);
            \fill[black] (-1,-0.5) circle (0.055);
            \fill[black] (1,-0.5) circle (0.055);

            \node at (0,1) {$B_1$};
            \node at (0,0) {$B_2$};
            \node at (0,-1) {$B_3$};

            \node at (-1.5,0.75) {$A_1$};
            \node at (0,0.75) {$A_2$};
            \node at (1.5,0.75) {$A_3$};
            \node at (-1.5,-0.25) {$A_4$};
            \node at (0,-0.25) {$A_5$};
            \node at (1.5,-0.25) {$A_6$};

            \node at (-1,0.2) {$M_1$};
            \node at (1,0.2) {$M_2$};
            \node at (-1,-0.8) {$M_3$};
            \node at (1,-0.8) {$M_4$};
        \end{tikzpicture}}}
    \hspace{0.07\textwidth}
    \subcaptionbox{}{\resizebox{0.265\textwidth}{!}{%
        \begin{tikzpicture}[
            line cap=round,
            line join=round,
            every node/.style={font=\scriptsize,xshift=0.08cm}]
            \path[use as bounding box] (-2.30,-1.65) rectangle (2.30,1.65);
            \filldraw[fill=gray!25,draw=none] (-2,-1.5) rectangle (2,1.5);
            \draw[line width=0.6pt] (-2,0) -- (2,0);
            \fill[black] (-1,0) circle (0.055);
            \fill[black] (1,0) circle (0.055);

            \node at (0,0.76) {$B_1$};
            \node at (0,-0.90) {$B_3$};
            \node at (-2.25,0)
                {$\vfus{A_1}{\diagramotc{B_2}}{A_4}$};
            \node at (0,0.20)
                {$\vfus{A_2}{\diagramotc{B_2}}{A_5}$};
            \node at (2.08,0)
                {$\vfus{A_3}{\diagramotc{B_2}}{A_6}$};

            \node at (-1,-0.52)
                {$\vfus{M_1}{\diagramotc{B_2}}{M_3}$};
            \node at (1,-0.52)
                {$\vfus{M_2}{\diagramotc{B_2}}{M_4}$};
        \end{tikzpicture}}}\\[0.45em]
    \subcaptionbox{}{\resizebox{0.265\textwidth}{!}{%
        \begin{tikzpicture}[
            line cap=round,
            line join=round,
            every node/.style={font=\scriptsize,xshift=0.08cm}]
            \path[use as bounding box] (-2.30,-1.65) rectangle (2.30,1.65);
            \filldraw[fill=gray!25,draw=none] (-2,-1.5) rectangle (2,1.5);
            \draw[line width=0.6pt] (-2,0.5) -- (2,0.5);
            \draw[line width=0.6pt] (-2,-0.5) -- (2,-0.5);
            \fill[black] (0,0.5) circle (0.055);
            \fill[black] (0,-0.5) circle (0.055);

            \node at (0,1) {$B_1$};
            \node at (1.32,0) {$B_2$};
            \node at (0,-1.18) {$B_3$};

            \node at (-1.5,0.75) {$A_1$};
            \node at (1.5,0.75) {$A_3$};
            \node at (-1.5,-0.25) {$A_4$};
            \node at (1.5,-0.25) {$A_6$};

            \node[font=\small] at (0,0.2)
                {$M_1\diagramotc{A_2}M_2$};
            \node[font=\small] at (0,-0.8)
                {$M_3\diagramotc{A_5}M_4$};
        \end{tikzpicture}}}
    \hspace{0.07\textwidth}
    \subcaptionbox{}{\resizebox{0.265\textwidth}{!}{%
        \begin{tikzpicture}[
            line cap=round,
            line join=round,
            every node/.style={font=\scriptsize,xshift=0.08cm}]
            \path[use as bounding box] (-2.30,-1.65) rectangle (2.30,1.65);
            \filldraw[fill=gray!25,draw=none] (-2,-1.5) rectangle (2,1.5);
            \draw[line width=0.6pt] (-2,0) -- (2,0);
            \fill[black] (0,0) circle (0.055);

            \node at (0,0.70) {$B_1$};
            \node at (1.48,-0.78) {$B_3$};
            \node at (-2.25,0)
                {$\vfus{A_1}{\diagramotc{B_2}}{A_4}$};
            \node at (2.08,0)
                {$\vfus{A_3}{\diagramotc{B_2}}{A_6}$};
            \node[font=\small] at (0,-0.7)
                {$\vfus{M_1\diagramotc{A_2}M_2}{\diagramotc{B_2}}
                    {M_3\diagramotc{A_5}M_4}$};
        \end{tikzpicture}}}
    \caption[Compatibility between the two compositions of 2-morphisms]{%
        Compatibility between the two compositions of 2-morphisms requires
        $\displaystyle
        \vfus{M_1\otc{A_2}M_2}{\otc{B_2}}{M_3\otc{A_5}M_4}
        \cong
        \vfus{M_1}{\otc{B_2}}{M_3}
        \otd\limits_{\vfus{A_2}{\otc{B_2}}{A_5}}
        \vfus{M_2}{\otc{B_2}}{M_4}$.
    }
    \label{fig:Morita_E2}
\end{figure}

It remains to verify that the relative tensor products used above preserve the required algebra and bi-bimodule structures.

\begin{prp}
    The 1-composition functor
    \begin{align*}
        \hind{B_1}{\Alg_{E_1}(\CC)}{B_2}
        \times
        \hind{B_2}{\Alg_{E_1}(\CC)}{B_3}
        &\to
        \hind{B_1}{\Alg_{E_1}(\CC)}{B_3},\\
        ((A_i,m_{A_i},h_{A_i}),(A_j,m_{A_j},h_{A_j}))
        &\mapsto
        (\vfus{A_i}{\otc{B_2}}{A_j},,)
    \end{align*}
    is well-defined.
\end{prp}

\begin{proof}
\end{proof}

\begin{prp}
    The composition functor along $B_2$,
    \begin{align*}
        \vfus{\hind{A_1}{\vind{B_1}{\CC}{B_2}}{A_2}}
        {\times}
        {\hind{A_3}{\vind{B_2}{\CC}{B_3}}{A_4}}
        &\to
        \hind{\vfus{A_1}{\otc{B_2}}{A_3}}
        {\vind{B_1}{\CC}{B_3}}
        {\vfus{A_2}{\otc{B_2}}{A_4}},\\
        ((M,l,r),(M',l',r'))
        &\mapsto
        (\vfus{M}{\otc{B_2}}{M'},,)
    \end{align*}
    is well-defined.
\end{prp}

\begin{prp}
    The composition functor along $A_2$,
    \begin{align*}
        \hind{A_1}{\vind{B_1}{\CC}{B_2}}{A_2}
        \times
        \hind{A_2}{\vind{B_1}{\CC}{B_2}}{A_3}
        &\to
        \hind{A_1}{\vind{B_1}{\CC}{B_2}}{A_3},\\
        ((M_1,l_1,r_1),(M_2,l_2,r_2))
        &\mapsto
        (M_1\otc{A_2}M_2,
        l_{M_1\otc{A_2}M_2},
        r_{M_1\otc{A_2}M_2})
    \end{align*}
    is well-defined.
\end{prp}

\begin{proof}
    The relative tensor product
    $M_1\otc{A_2}M_2$ naturally carries an $A_1$-$A_3$-bimodule structure:
    the left action $l_{M_1\otc{A_2}M_2}$ is induced by $l_1$, while the right action
    $r_{M_1\otc{A_2}M_2}$ is induced by $r_2$.
\end{proof}

\begin{thm}
    The tricategory structure defined above is well-defined.
\end{thm}

\begin{proof}
    The associators for the relative tensor-product compositions are the canonical associativity isomorphisms of relative tensor products, and the interchange isomorphism is induced in the same way.
    The required coherence conditions follow from the universal properties of the corresponding coequalizers.
\end{proof}

\subsubsection{Four-fold Junctions and Quadrimodules}

Ordinary bi-bimodules describe junctions between two parallel
$E_1$-interfaces separating the same pair of $E_2$-algebras.
More general defect networks may contain a junction at which four
interfaces meet and four bulk regions are incident.  The corresponding
algebraic datum is captured by the following notion.

\begin{figure}[H]\label{fig:quadrimodule}
    \centering
    \subcaptionbox{}{%
        \begin{tikzpicture}[line cap=round,line join=round]
            \path[use as bounding box] (-2.4,-2) rectangle (2.4,2);
            \filldraw[fill=gray!40,draw=white]
                (-2,1.33) rectangle (2,-1.33);

            \coordinate (right-junction) at (0.6,0);
            \draw[thick]
                (-2,0.70) .. controls (-1.15,0.70) and (0.05,0.18) .. (right-junction);
            \draw[thick] (-2,0) -- (right-junction);
            \draw[thick]
                (-2,-0.70) .. controls (-1.15,-0.70) and (0.05,-0.18) .. (right-junction);
            \draw[thick] (right-junction) -- (2,0);
            \fill[black] (right-junction) circle (0.07);

            \node at (-0.20,0.95) {$B_2$};
            \node at (-1.00,0.35) {$B_1$};
            \node at (-1.00,-0.35) {$B_4$};
            \node at (-0.20,-0.95) {$B_3$};
            \node[left] at (-1.95,0)
                {$\vfus{\vfus{A_2}{\otc{B_1}}{A_1}}{\otc{B_4}}{A_4^{\mathrm{op}}}$};
            \node at (1.40,0.25) {$A_3$};
        \end{tikzpicture}%
    }
    \hfill
    \subcaptionbox{}{%
        \begin{tikzpicture}[line cap=round,line join=round]
            \path[use as bounding box] (-2.4,-2) rectangle (2.4,2);
            \filldraw[fill=gray!40,draw=white]
                (-2,1.33) rectangle (2,-1.33);

            \draw[thick] (-2,0) -- (2,0);
            \draw[thick] (0,-1.33) -- (0,1.33);
            \fill[black] (0,0) circle (0.07);

            \node at (-1,0.8) {$B_1$};
            \node at (1,0.8) {$B_2$};
            \node at (1,-0.8) {$B_3$};
            \node at (-1,-0.8) {$B_4$};

            \node at (-1.4,0.25) {$A_1$};
            \node at (1.4,0.25) {$A_3$};
            \node at (0.3,1) {$A_2$};
            \node at (0.3,-1) {$A_4$};
        \end{tikzpicture}%
    }
    \hfill
    \subcaptionbox{}{%
        \begin{tikzpicture}[line cap=round,line join=round]
            \path[use as bounding box] (-2.4,-2) rectangle (2.4,2);
            \filldraw[fill=gray!40,draw=white]
                (-2,1.33) rectangle (2,-1.33);

            \coordinate (left-junction) at (-0.6,0);
            \draw[thick] (-2,0) -- (left-junction);
            \draw[thick]
                (left-junction) .. controls (-0.05,0.18) and (1.15,0.70) .. (2,0.70);
            \draw[thick] (left-junction) -- (2,0);
            \draw[thick]
                (left-junction) .. controls (-0.05,-0.18) and (1.15,-0.70) .. (2,-0.70);
            \fill[black] (left-junction) circle (0.07);

            \node at (0.20,0.95) {$B_1$};
            \node at (1.00,0.35) {$B_2$};
            \node at (1.00,-0.35) {$B_3$};
            \node at (0.20,-0.95) {$B_4$};
            \node at (-1.40,0.25) {$A_1$};
            \node[right] at (1.95,0)
                {$\vfus{\vfus{A_2^{\mathrm{op}}}{\otc{B_2}}{A_3}}{\otc{B_3}}{A_4}$};
        \end{tikzpicture}%
    }
\end{figure}

\begin{dfn}[Quadrimodule]\label{dfn:quadrimodule}
    Let $\CC$ be an $E_2$-monoidal category, and let
    $B_1,B_2,B_3,B_4$ be $E_2$-algebras in $\CC$.
    Let
    \[
        A_1\in\hind{B_1}{\Alg_{E_1}(\CC)}{B_4},
        \qquad
        A_2\in\hind{B_1}{\Alg_{E_1}(\CC)}{B_2},
    \]
    \[
        A_3\in\hind{B_2}{\Alg_{E_1}(\CC)}{B_3},
        \qquad
        A_4\in\hind{B_4}{\Alg_{E_1}(\CC)}{B_3},
    \]
    cyclically arranged as in \cref{fig:quadrimodule} (b).

    For every incident pair $(B_i,A_j)$, denote by
    \[
        \iota_{B_i}^{A_j}:B_i\longrightarrow A_j
    \]
    the map induced by the corresponding unital $B_i$-action on $A_j$;
    explicitly, it is obtained by inserting the unit
    $\bfone\to A_j$ into that action.

    An
    $A_1$-$A_2$-$A_3$-$A_4$-\emph{quadrimodule}
    is a tuple
    \[
        (M,l,u,r,d),
    \]
    consisting of an object $M\in\CC$ together with actions
    \[
        l:A_1\ot M\longrightarrow M,
        \qquad
        u:A_2\ot M\longrightarrow M,
    \]
    \[
        r:M\ot A_3\longrightarrow M,
        \qquad
        d:M\ot A_4\longrightarrow M,
    \]
    satisfying the following conditions.

    \begin{enumerate}
        \item
        The four adjacent pairs of $A_i$-actions define compatible
        bimodule structures:
        \[
            \bigl(l,\,
            u\circ\beta^{-1}_{A_2,M}\bigr),
            \qquad
            (u,r),
        \]
        \[
            \bigl(
            d\circ\beta^{-1}_{M,A_4},\,r
            \bigr),
            \qquad
            (l,d).
        \]
        More precisely, these give respectively an
        $A_1$-$A_2$,
        $A_2$-$A_3$,
        $A_4$-$A_3$, and
        $A_1$-$A_4$ bimodule structure on $M$.

        \item
        At each corner $B_i$, the two $B_i$-actions on $M$
        induced through the two incident $A_j$-actions agree.
        Explicitly,
        \begin{align}
            l\circ
            \bigl(\iota_{B_1}^{A_1}\ot\id_M\bigr)
            &=
            u\circ
            \bigl(\iota_{B_1}^{A_2}\ot\id_M\bigr),
            \label{eq:quadrimodule-B1}
            \\[0.5em]
            u\circ
            \bigl(\iota_{B_2}^{A_2}\ot\id_M\bigr)
            &=
            r\circ\beta^{-1}_{M,A_3}
            \circ
            \bigl(\iota_{B_2}^{A_3}\ot\id_M\bigr),
            \label{eq:quadrimodule-B2}
            \\[0.5em]
            r\circ
            \bigl(\id_M\ot\iota_{B_3}^{A_3}\bigr)
            &=
            d\circ
            \bigl(\id_M\ot\iota_{B_3}^{A_4}\bigr),
            \label{eq:quadrimodule-B3}
            \\[0.5em]
            l\circ
            \bigl(\iota_{B_4}^{A_1}\ot\id_M\bigr)
            &=
            d\circ\beta^{-1}_{M,A_4}
            \circ
            \bigl(\iota_{B_4}^{A_4}\ot\id_M\bigr).
            \label{eq:quadrimodule-B4}
        \end{align}
    \end{enumerate}
\end{dfn}

\begin{rmk}
    A $B_1$-$B_3$-bi-$A_1$-$A_3$-bimodule is recovered as a special case of a quadrimodule by setting
$B_1=A_2=B_2$ and $B_4=A_4=B_3$.
\end{rmk}

\subsection{Bi-bimodules and Iterated Bimodules}\label{subsec:two-morita-iterated-bimodules}

This subsection compares the explicit bi-bimodule description introduced above
with the iterated-bimodule description introduced by Haugseng in Definition \ref{dfn:higher_morita_haugseng}.
The main idea is to regard the category
$\hind{B_1}{\CC}{B_2}$ itself as a monoidal category and then internalize
the next layers of Morita data inside it.

We proceed in three steps.
First, Theorem \ref{thm:B_1CB_2_to_B_1B_2CB_1B_2_is_fully_faithful}
embeds $B_1$-$B_2$-bimodules into bimodules over $B_1\ot B_2$.
Using this embedding, Theorem \ref{thm:bimodule-category-monoidal} equips
$\hind{B_1}{\CC}{B_2}$ with the monoidal structure needed for the
iterated construction.
Second, we show that an $E_1$ $B_1$-$B_2$-bimodule in our sense is
equivalently an $E_1$-algebra internal to this monoidal bimodule category:
\[
    \hind{B_1}{\Alg_{E_1}(\CC)}{B_2}
    \simeq
    \Alg_{E_1}\!\left(\hind{B_1}{\CC}{B_2}\right),
\]
as proved in Theorem \ref{thm:algebras_in_bimodule_categories}.
Finally, applying the same idea one categorical level higher, we prove that
our explicit $B_1$-$B_2$-bi-$A_1$-$A_2$-bimodules are precisely
$A_1$-$A_2$-bimodules internal to $\hind{B_1}{\CC}{B_2}$:
\[
    \hind{A_1}{\vind{B_1}{\CC}{B_2}}{A_2}
    \simeq
    \hind{A_1}{\left(\hind{B_1}{\CC}{B_2}\right)}{A_2},
\]
see Theorem \ref{thm:bibimod_equivalent_to_bimod_in_bimod}.

\begin{thm}\label{thm:B_1CB_2_to_B_1B_2CB_1B_2_is_fully_faithful}
    Let $B_1$ and $B_2$ be two $E_2$-algebras in a finitely cocomplete $E_{m}$-monoidal category $\CC$ ($m\geq 2$) and tensor functor is right exact at each variables.
    Then there is a fully faithful functor
    \begin{align*}
        \hind{B_1}{\CC}{B_2}&\to \hind{B_1\ot B_2}{\CC}{B_1\ot B_2}\\
        (M,u_M,d_M)&\mapsto (M,l^{12}_M,r^{12}_M)\\
        f&\mapsto f
    \end{align*} 
    where 
    \begin{align*}
        l^{12}_M:B_1\ot B_2\ot M
        \xrightarrow{\id_{B_1}\ot\beta_{B_2,M}}
        B_1\ot M\ot B_2
        \xrightarrow{u_M\ot\id_{B_2}}
        M\ot B_2
        \xrightarrow{d_M}
        M
    \end{align*}
    and 
    \begin{align*}
        r^{12}_M:M\ot B_1\ot B_2
        \xrightarrow{\beta_{M,B_1}\ot\id_{B_2}}
        B_1\ot M\ot B_2
        \xrightarrow{u_M\ot\id_{B_2}}
        M\ot B_2
        \xrightarrow{d_M}
        M .
    \end{align*}
\end{thm}
\begin{proof}
    On object level:
    Let $M\in\hind{B_1}{\CC}{B_2}$ with left and right actions
    $u_M:B_1\ot M\to M$ and $d_M:M\ot B_2\to M$.  
    The object $M$ has a left $B_1\ot B_2$-action
    \begin{align*}
        l^{12}_M:B_1\ot B_2\ot M
        \xrightarrow{\id_{B_1}\ot\beta_{B_2,M}}
        B_1\ot M\ot B_2
        \xrightarrow{u_M\ot\id_{B_2}}
        M\ot B_2
        \xrightarrow{d_M}
        M
    \end{align*}
    We check the associativity:
    \begin{align*}
        \resizebox{0.98\linewidth}{!}{$
        \xymatrix@C=2.2em@R=4.8em{
            B_1\ot B_2\ot B_1\ot B_2\ot M
                \ar[rr]^{\id_{B_1\ot B_2\ot B_1}\ot\beta_{B_2,M}}
                \ar[d]_{\id_{B_1}\ot\beta_{B_2,B_1}\ot\id_{B_2\ot M}}
            &{}
            &B_1\ot B_2\ot B_1\ot M\ot B_2
                \ar[r]^{\id_{B_1\ot B_2\ot B_1}\ot d_M}
                \ar[d]_{\id_{B_1}\ot\beta_{B_2,B_1\ot M\ot B_2}}
            &B_1\ot B_2\ot B_1\ot M
                \ar[r]^{\id_{B_1\ot B_2}\ot u_M}
                \ar[d]_{\id_{B_1}\ot\beta_{B_2,B_1\ot M}}
            &B_1\ot B_2\ot M
                \ar[d]^{\id_{B_1}\ot\beta_{B_2,M}}
            \\
            B_1\ot B_1\ot B_2\ot B_2\ot M
                \ar[r]^{\id_{B_1\ot B_1}\ot\beta_{B_2\ot B_2,M}}
                \ar[d]_{\id_{B_1\ot B_1}\ot m_2\ot\id_M}
            &B_1\ot B_1\ot M\ot B_2\ot B_2
                \ar[r]^{\id_{B_1\ot B_1\ot M}\ot\beta_{B_2,B_2}}
            &B_1\ot B_1\ot M\ot B_2\ot B_2
                \ar[r]^{\id_{B_1\ot B_1}\ot d_M\ot\id_{B_2}}
                \ar[d]_{\id_{B_1\ot B_1\ot M}\ot m_2}
            &B_1\ot B_1\ot M\ot B_2
                \ar[r]^{\id_{B_1}\ot u_M\ot\id_{B_2}}
                \ar[d]_{\id_{B_1\ot B_1}\ot d_M}
            &B_1\ot M\ot B_2
                \ar[d]^{\id_{B_1}\ot d_M}
            \\
            B_1\ot B_1\ot B_2\ot M
                \ar[rr]^{\id_{B_1\ot B_1}\ot\beta_{B_2,M}}
                \ar[d]_{m_1\ot\id_{B_2\ot M}}
            &{}
            &B_1\ot B_1\ot M\ot B_2
                \ar[r]^{\id_{B_1\ot B_1}\ot d_M}
                \ar[d]_{m_1\ot\id_{M\ot B_2}}
            &B_1\ot B_1\ot M
                \ar[r]^{\id_{B_1}\ot u_M}
                \ar[d]_{m_1\ot\id_M}
            &B_1\ot M
                \ar[d]^{u_M}
            \\
            B_1\ot B_2\ot M
                \ar[rr]_{\id_{B_1}\ot\beta_{B_2,M}}
            &{}
            &B_1\ot M\ot B_2
                \ar[r]_{\id_{B_1}\ot d_M}
            &B_1\ot M
                \ar[r]_{u_M}
            &M .
        }
        $}
    \end{align*}
    where the left upper corner commutes as we expand it as follows:
    \begin{align*}
        \xymatrix@C=5.0em@R=5.4em{
            {}
            &B_2\ot B_1\ot B_2\ot M
                \ar[r]^{\id_{B_2\ot B_1}\ot\beta_{B_2,M}}
                \ar[dl]_{\beta_{B_2,B_1}\ot\id_{B_2\ot M}}
            &B_2\ot B_1\ot M\ot B_2
                \ar[dl]_{\beta_{B_2,B_1}\ot\id_{M\ot B_2}}
                \ar[ddl]^{\beta_{B_2,B_1\ot M}\ot\id_{B_2}}
                \ar[dd]^{\beta_{B_2,B_1\ot M\ot B_2}}
            \\
            B_1\ot B_2\ot B_2\ot M
                \ar[r]^{\id_{B_1\ot B_2}\ot\beta_{B_2,M}}
                \ar[dr]_{\id_{B_1}\ot\beta_{B_2\ot B_2,M}}
            &B_1\ot B_2\ot M\ot B_2
                \ar[d]_{\id_{B_1}\ot\beta_{B_2,M}\ot\id_{B_2}}
            \\
            {}
            &B_1\ot M\ot B_2\ot B_2
                \ar[r]_{\id_{B_1\ot M}\ot\beta_{B_2,B_2}}
            &B_1\ot M\ot B_2\ot B_2 .
        }
    \end{align*}
    The unitality:
    \begin{align*}
        \xymatrix@C=4.0em@R=3.6em{
            1\ot M
                \ar[r]^{h_2\ot\id_M}
                \ar[d]_{\beta_{1,M}}
            &B_2\ot M
                \ar[r]^{h_1\ot\id_{B_2\ot M}}
                \ar[d]_{\beta_{B_2,M}}
            &B_1\ot B_2\ot M
                \ar[d]^{\id_{B_1}\ot\beta_{B_2,M}}
            \\
            M\ot 1
                \ar[r]_{\id_M\ot h_2}
                \ar[ddrr]_{\sim}
            &M\ot B_2
                \ar[r]_{h_1\ot\id_{M\ot B_2}}
                \ar[dr]^{\sim}
            &B_1\ot M\ot B_2
                \ar[d]^{u_M\ot\id_{B_2}}
            \\
            &&M\ot B_2
                \ar[d]^{d_M}
            \\
            &&M .
        }
    \end{align*}
    we can also check the validity of the right $B_1\ot B_2$-action similarly
    \begin{align*}
        r^{12}_M:M\ot B_1\ot B_2
        \xrightarrow{\beta_{M,B_1}\ot\id_{B_2}}
        B_1\ot M\ot B_2
        \xrightarrow{u_M\ot\id_{B_2}}
        M\ot B_2
        \xrightarrow{d_M}
        M .
    \end{align*}
    The $E_2$-algebra structures on $B_1$ and $B_2$ imply that these are associative and unital $B_1\ot B_2$-actions and that they commute.  Indeed, the only extra point beyond the ordinary bimodule axioms is that two copies of $B_i$ may have to pass each other; this is exactly the commutativity relation
    $m_i\circ\beta_{B_i,B_i}=m_i$.

    Now, we need to check that $l^{12}_M$ and $r^{12}_M$ are compatible bimodule actions.
    \begin{align*}
        \resizebox{0.98\linewidth}{!}{$
        
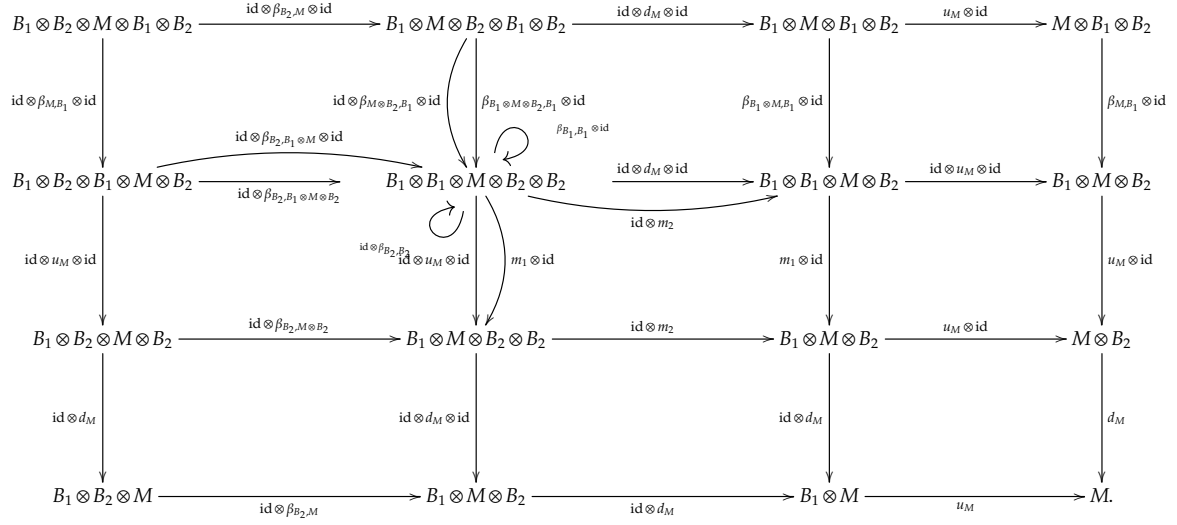
\xymatrix@C=7.0em@R=6.4em{
            B_1\ot B_2\ot M\ot B_1\ot B_2
                \ar[r]^{\id\ot\beta_{B_2,M}\ot\id}
                \ar[d]_{\id\ot\beta_{M,B_1}\ot\id}
            &B_1\ot M\ot B_2\ot B_1\ot B_2
                \ar[r]^{\id\ot d_M\ot\id}
                \ar@/_1.2pc/[d]_{\scriptstyle\id\ot\beta_{M\ot B_2,B_1}\ot\id}
                \ar[d]^{\scriptstyle\beta_{B_1\ot M\ot B_2,B_1}\ot\id}
            &B_1\ot M\ot B_1\ot B_2
                \ar[r]^{u_M\ot\id}
                \ar[d]_{\beta_{B_1\ot M,B_1}\ot\id}
            &M\ot B_1\ot B_2
                \ar[d]^{\beta_{M,B_1}\ot\id}
            \\
            B_1\ot B_2\ot B_1\ot M\ot B_2
                \ar@/^1.2pc/[r]^{\scriptstyle\id\ot\beta_{B_2,B_1\ot M}\ot\id}
                \ar[r]_{\scriptstyle\id\ot\beta_{B_2,B_1\ot M\ot B_2}}
                \ar[d]_{\id\ot u_M\ot\id}
            &\hphantom{\qquad}
                B_1\ot B_1\ot
                \begin{tikzcd}[
                    baseline=(morita-loop-node.base),
                    cells={nodes={inner sep=0pt,outer sep=0pt}}
                ]
                    |[alias=morita-loop-node]|M
                    \arrow[
                        overlay,shift left=5,loop,
                        in=190,out=260,distance=10mm,
                        "{\scriptscriptstyle\id\ot\beta_{B_2,B_2}}"'
                            {xshift=-14mm,yshift=-5mm}
                    ]
                    \arrow[
                        overlay,shift left=5,loop,
                        in=10,out=80,distance=10mm,
                        "{\scriptscriptstyle\beta_{B_1,B_1}\ot\id}"
                            {xshift=5mm,yshift=-2mm}
                    ]
                \end{tikzcd}
                \ot B_2\ot B_2
                \hphantom{\qquad}
                \ar[r]^{\scriptstyle\id\ot d_M\ot\id}
                \ar@/_1.2pc/[r]_{\scriptstyle\id\ot m_2}
                \ar[d]_{\id\ot u_M\ot\id}
                \ar@/^1.2pc/[d]^{m_1\ot\id}
            &B_1\ot B_1\ot M\ot B_2
                \ar[r]^{\id\ot u_M\ot\id}
                \ar[d]_{m_1\ot\id}
            &B_1\ot M\ot B_2
                \ar[d]^{u_M\ot\id}
            \\
            B_1\ot B_2\ot M\ot B_2
                \ar[r]^{\id\ot\beta_{B_2,M\ot B_2}}
                \ar[d]_{\id\ot d_M}
            &B_1\ot M\ot B_2\ot B_2
                \ar[r]^{\id\ot m_2}
                \ar[d]_{\id\ot d_M\ot\id}
            &B_1\ot M\ot B_2
                \ar[d]_{\id\ot d_M}
                \ar[r]^{u_M\ot\id}
            &M\ot B_2
                \ar[d]^{d_M}
            \\
            B_1\ot B_2\ot M
                \ar[r]_{\id\ot\beta_{B_2,M}}
            &B_1\ot M\ot B_2
                \ar[r]_{\id\ot d_M}
            &B_1\ot M
                \ar[r]_{u_M}
            &M .
        }
        $}
    \end{align*}
    
    On morphism level:
    Let $f:M\to N$ be a $B_1$-$B_2$-bimodule homomorphism.
    Then it is obvious that $f$ is also a left $B_1\ot B_2$-module homomorphism, i.e the following diagram (without squiggly part) commutes
    
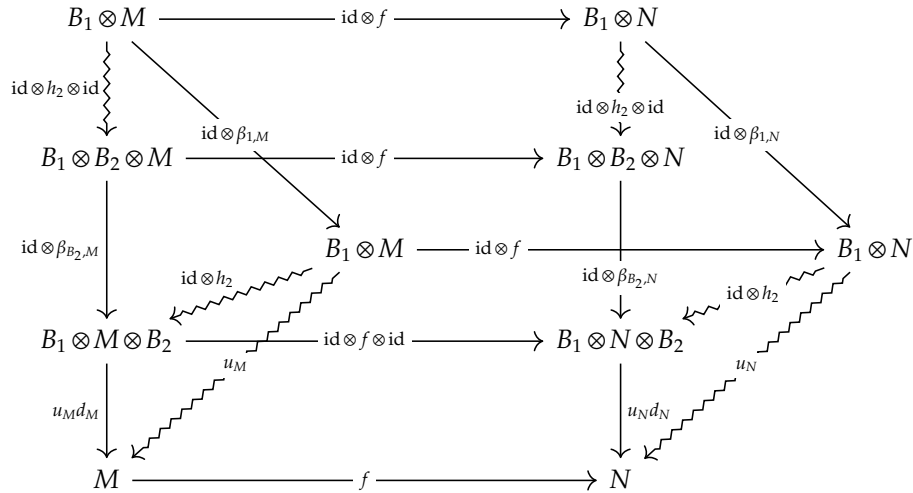
\[\begin{tikzcd}
        {B_1\ot M} &&&& {B_1\ot N} && \\
        \\
        {B_1\ot B_2\ot M} &&&& {B_1\ot B_2\ot N} \\
        && {B_1\ot M} &&&& B_1\ot N \\
        {B_1\ot M\ot B_2} &&&& {B_1\ot N\ot B_2} \\
        \\
        M &&&& N
        \arrow["{\id\ot f}"{description}, from=1-1, to=1-5]
        \arrow["{\id\ot h_2\ot \id}"', squiggly, from=1-1, to=3-1]
        \arrow["{\id\ot\beta_{1,M}}"{description}, from=1-1, to=4-3]
        \arrow["{\id\ot h_2\ot \id}"{description, pos=0.7}, squiggly, from=1-5, to=3-5]
        \arrow["{\id\ot\beta_{1,N}}"{description}, from=1-5, to=4-7]
        \arrow["{\id\ot f}"{description}, from=3-1, to=3-5]
        \arrow["{\id\ot\beta_{B_2,M}}"', from=3-1, to=5-1]
        \arrow["{\id\ot\beta_{B_2,N} }"{description, pos=0.7}, from=3-5, to=5-5]
        \arrow["{\id\ot f}"{description, pos=0.2}, from=4-3, to=4-7]
        \arrow["{\id\ot h_2}"', squiggly, from=4-3, to=5-1]
        \arrow["{u_M}"{description}, squiggly, from=4-3, to=7-1]
        \arrow["{\id\ot h_2}"{description}, squiggly, from=4-7, to=5-5]
        \arrow["{u_N}"{description}, squiggly, from=4-7, to=7-5]
        \arrow["{\id\ot f\ot \id}"{description}, from=5-1, to=5-5]
        \arrow["{u_Md_M}"', from=5-1, to=7-1]
        \arrow["{u_Nd_N}", from=5-5, to=7-5]
        \arrow["f"{description}, from=7-1, to=7-5]
    \end{tikzcd}\]
    Similarly, one can also show that $f$ is a right $B_1\ot B_2$-module homomorphism.

    Now suppose $f$ is a $B_1\ot B_2$-$B_1\ot B_2$-bimodule homomorphism, consider the full part of above diagram, it is clear that $f$ is a left $B_1$-module homomorphism. 
    Similarly, we can demonstrate that $f$ is also a right $B_2$-module homomorphism.
    So the functor is well-defined now.
    And it is manifestly fully faithful.
\end{proof}

\begin{crl}
    Let $B$ be an $E_2$-algebra in a finitely cocomplete $E_{m}$-monoidal category $\CC$ ($m\geq 2$) and tensor functor is right exact at each variables.
    Then there is a fully faithful functor
    \begin{align*}
        \alpha^+: \hind{}{\CC}{B}&\to \hind{B}{\CC}{B}\\
        (M,d_M)&\mapsto (M,l_M,d_M)\\
        f&\mapsto f
    \end{align*} 
    where 
    \begin{align*}
        l_M:B\ot M
        \xrightarrow{\beta_{B,M}}
        M\ot B_2
        \xrightarrow{d_M}
        M
    \end{align*}
\end{crl}

\begin{thm}\label{thm:bimodule-category-monoidal}
    Let $B_1$ and $B_2$ be two $E_2$-algebras in a finitely cocomplete $E_{2}$-monoidal category $\CC$ and tensor functor is right exact at each variables.
    Then $(\hind{B_1}{\CC}{B_2},\otd\limits_{B_1\ot B_2},B_1\ot B_2)$ is an $E_{1}$-monoidal category.
\end{thm}
\begin{proof}
    Now consider $(B_1\ot B_2,m_1\ot\id,\id\ot m_2)\in\hind{B_1}{\CC}{B_2}$, 
    it is clear that $l^{12}_{B_1\ot B_2},r^{12}_{B_1\ot B_2}$ are regular action of the $E_1$-algebra $B_1\ot B_2$, so $(B_1\ot B_2,l^{12}_{B_1\ot B_2},r^{12}_{B_1\ot B_2})$ is the tensor unit in $\hind{B_1}{\CC}{B_2}$. 

    For $M,N\in\hind{B_1}{\CC}{B_2}$, define $M\otd\limits_{B_1\ot B_2}N$ by the coequalizer
    \begin{align*}
        \xymatrix{
            M\ot B_1\ot B_2\ot N
            \ar@<0.5ex>[r]^{r^{12}_M\ot\id_N}
            \ar@<-0.5ex>[r]_{\id_M\ot l^{12}_N}
            &M\ot N
            \ar[r]^-p
            &M\otd\limits_{B_1\ot B_2}N .
        }
    \end{align*}
    We claim that $M\otd\limits_{B_1\ot B_2}N$ is again a $B_1$-$B_2$-bimodule.  The left $B_1$-action is induced from the left action on the first factor,
    \begin{align*}
        \xymatrix{
            B_1\ot (M\ot B_1\ot B_2\ot N) \ar@<0.5ex>[r]^{r^{12}_M\ot\id_N} \ar@<-0.5ex>[r]_{\id_M\ot l^{12}_N}\ar[d]_{u_M\ot\id} &B_1\ot (M\ot N) \ar[r]^-p \ar[d]_{u_M\ot\id} &B_1\ot (M\otd\limits_{B_1\ot B_2}N) \ar@{-->}[d]^{\exists ! u_{M\otd\limits_{B_1\ot B_2}N}} \\
            M\ot B_1\ot B_2\ot N \ar@<0.5ex>[r]^{r^{12}_M\ot\id_N} \ar@<-0.5ex>[r]_{\id_M\ot l^{12}_N} &M\ot N \ar[r]^-p &M\otd\limits_{B_1\ot B_2}N 
        }
    \end{align*}
    the left subdiagram of above diagram commutes, one case is obvious, another case is demonstrated as follows:
    \begin{align*}
        \resizebox{0.98\linewidth}{!}{$
        \xymatrix@C=5.4em@R=5.6em{
            B_1\ot M\ot B_1\ot B_2
                \ar@<0.55ex>[r]^{\id_{B_1}\ot\beta_{M,B_1}\ot\id_{B_2}}
                \ar@<-0.55ex>[r]_{\beta_{B_1\ot M,B_1}\ot\id_{B_2}}
                \ar[d]_{u_M\ot\id_{B_1\ot B_2}}
            &B_1\ot B_1\ot M\ot B_2
                \ar@(ul,ur)[]^{\beta_{B_1,B_1}\ot\id_{M\ot B_2}}
                \ar[r]^{\id_{B_1}\ot u_M\ot\id_{B_2}}
                \ar@<-0.55ex>[d]_{\id_{B_1}\ot u_M\ot\id_{B_2}}
                \ar@<0.55ex>[d]^{m_1\ot\id_{M\ot B_2}}
            &B_1\ot M\ot B_2
                \ar[r]^{\id_{B_1}\ot d_M}
                \ar[d]_{u_M\ot\id_{B_2}}
            &B_1\ot M
                \ar[d]^{u_M}
            \\
            M\ot B_1\ot B_2
                \ar[r]_{\beta_{M,B_1}\ot\id_{B_2}}
            &B_1\ot M\ot B_2
                \ar[r]_{u_M\ot\id_{B_2}}
            &M\ot B_2
                \ar[r]_{d_M}
            &M .
        }
        $}
    \end{align*}
    So by universal property, we obtain the left $B_1$-action $u_{M\otd_{B_1\ot B_2}N}$.
    and the right $B_2$-action is induced from the right action on the second factor,
    \begin{align*}
        M\ot N\ot B_2\xrightarrow{\id_M\ot d_N}M\ot N\xrightarrow{p}M\otd_{B_1\ot B_2}N.
    \end{align*}
    The verification for the right $B_2$-action is the same, using the $E_2$-commutativity of $B_2$.  The two induced actions commute because the original $B_1$- and $B_2$-actions on $M$ and $N$ commute.

    Associativity of $\otd_{B_1\ot B_2}$ and the unit constraints are inherited from the usual associativity and unit constraints for relative tensor products over the $E_1$-algebra $B_1\ot B_2$.  The unit object is $B_1\ot B_2$, equipped with the evident left $B_1$-action and right $B_2$-action.  Thus, $\hind{B_1}{\CC}{B_2}$ is an $E_1$-monoidal category.
\end{proof}

\begin{crl}
    Let $B$ be an $E_2$-algebra in a finitely cocomplete $E_{m}$-monoidal category $\CC$ ($m\geq 2$) and tensor functor is right exact at each variables.
    Then $(\hind{}{\CC}{B},\otc{B},B)$ is an $E_{1}$-monoidal category.
\end{crl}

\begin{dfn}[Haugseng]
    Let $\CC$ be an $E_{m+2}$-monoidal $1$-category.
    The tricategory $\Mrt_{E_2}(\CC)$ consists of 
    \begin{itemize}
        \item objects are $E_2$-algebras in $\CC$, i.e. objects in $\Alg_{E_2}(\CC)$;
        \item a 1-morphism is a $B_1$-$B_2$-bimodule $E_1$-algebra for two $E_2$-algebras $B_1$ and $B_2$, i.e. an object in $\Alg_{E_1}(\hind{B_1}{\CC}{B_2})$;
        \item  a 2-morphisms between two $E_1$ $B_1$-$B_2$-bimodules algebra $A_1$, $A_2$ is an $A_1$-$A_2$-bimodule internalized to $B_1$-$B_2$-bimodule ($E_0$-algebra), i.e. an object in $\Alg_{E_{0}}(\hind{A_1}{(\hind{B_1}{\CC}{B_2})}{A_2})$ or simply $\hind{A_1}{(\hind{B_1}{\CC}{B_2})}{A_2}$.
        \item a 3-morphism is a morphism in $\hind{A_1}{(\hind{B_1}{\CC}{B_2})}{A_2}$.
\end{itemize}
\end{dfn}

\begin{lem}[\cite{Lur17}]\label{lem:lax_monoidal_functor_preserves_modules}
    Let $F:\CD\to \CC$ be a lax monoidal functor.
    For any $A_1,A_2\in \Alg_{E_1}(\CD)$, there is a canonical functor 
    \begin{align*}
        \hind{A_1}{\CD}{A_2}\to \hind{F(A_1)}{\CC}{F(A_2)}
    \end{align*}
\end{lem}

\begin{prp}\label{prp:forget_is_monoidal}
    The forgetful functor 
    \begin{align*}
        U_{E_1}:\Alg_{E_1}(\CC)\to \CC 
    \end{align*}
    is strong monoidal.
\end{prp}

\begin{lem}\label{lem:alg_bimod_is_bimod_alg}
    Let $A\in \hind{B_1}{\Alg_{E_1}(\CC)}{B_2}$, then we have $A\in \Alg_{E_1}(\hind{B_1}{\CC}{B_2})$.
\end{lem}
\begin{proof}
    Let $((A,m:A\ot A\to A, h:1\to A),u:B_1\ot A\to A,d:A\ot B_2\to A)$ be an object in $\hind{B_1}{\Alg_{E_1}(\CC)}{B_2}$.
    Note that $B_1$ and $B_2$ are $E_1$-algebras in $\Alg_{E_1}(\CC)$, then by Lemma \ref{lem:lax_monoidal_functor_preserves_modules} and Proposition \ref{prp:forget_is_monoidal}, $(A,u,d)$ is an object in $\hind{B_1}{\CC}{B_2}$.
     
    Now we want to induce algebra structure on $(A,u,d)$ by $m$ and $h$.
    \begin{itemize}
        \item We define $h':B_1\ot B_2\to A$ to be 
        \begin{align*}
            h':B_1\ot B_2\xrightarrow{\id\ot h\ot\id}B_1\ot A\ot B_2\xrightarrow{u\ot\id}A\ot B_2\xrightarrow{d} A
        \end{align*}
        By compatibility of bimodule actions, $u$ and $d$ can commute in above definition.
        We also need to check $h'$ is a morphism in $\hind{B_1}{\CC}{B_2}$.
        \begin{enumerate}
            \item Here we show $h'$ is a left $B_1$-module homomorphism
            \begin{align*}
                \xymatrix{
                    B_1 B_1 B_2 \ar[r]^{\id\ot h\ot\id}\ar[d]_{m_1\ot\id} &B_1B_1AB_2\ar[r]^{\id\ot d}\ar[d]_{m_1\ot\id} &B_1B_1 A\ar[d]_{m_1\ot\id}\ar[r]^{u\ot\id} &B_1A\ar[d]^u\\
                    B_1 B_2\ar[r]_{\id\ot h\ot\id} &B_1A B_2\ar[r]_{\id\ot d} &B_1 A\ar[r]_{u} &A
                }
            \end{align*}

            \item we can prove $h'$ is a right $B_2$-module homomorphism similarly.
        \end{enumerate}
        
        \item Consider 
        \begin{align*}
            \xymatrix{
                A\ot (B_1\ot B_2)\ot A\ar@<0.5ex>[r]^{(ud\ot\id)\circ(\beta_{A,B_1}\ot\id)}\ar@<-0.5ex>[r]_{(\id\ot ud)\circ(\id\ot\beta_{B_2,A})} &A\ot A\ar[r]^{p_{A,A}}\ar[dr]_m &A\otd\limits_{B_1\ot B_2}A\ar@{-->}[d]^{m'}\\
                    & &A 
            }
        \end{align*}
        we demonstrate that $m$ coequalize the above two morphisms $(ud\ot\id)\circ(\beta_{A,B_1}\ot\id)$ and $(\id\ot ud)\circ(\id\ot\beta_{B_2,A})$, or equivalently, the following diagram commutes 
        \begin{align*}
            \xymatrix{
                    &B_1AB_2A\ar[r]^{ud\ot\id_{A}}\ar[d] &AA\ar[dr]^m\\
                AB_1B_2A\ar[ur]^{\beta_{A,B_1}\ot\id}\ar[dr]_{\id\ot \beta_{B_2,A}} &B_1AAB_2\ar[r]_{\id\ot m \ot\id} &B_1AB_2\ar[r]_{ud} &A\\
                    &AB_1AB_2\ar[r]_{\id_{A}\ot ud}\ar[u] &AA\ar[ur]_m
            }
        \end{align*}
        where left subdiagram commutes obviously, and we expand right bottom subdiagram as follows
        \begin{align*}
            \xymatrix{
                B_1AAB_2\ar[d]_{h_1\ot\id} &AB_1AB_2\ar[d]_{h_1\ot\id}\ar[l]_{\beta_{}\ot\id}\ar[r]^{\id\ot d} &AB_1A\ar[dr]^{\id\ot u}\ar[d]_{h_1\ot\id }\\
                B_1B_1AAB_2\ar[d]_{m_1\ot\id} & B_1AB_1AB_2\ar[l]^{\id\ot\beta_{}\ot\id}\ar[d]_{m_{B_1\ot A}\ot\id}\ar[r]_{\id\ot d} &B_1AB_1A\ar[r]^{u\ot u}\ar[d]_{m_{B_1\ot A}} &AA\ar[d]^m\\
                B_1AAB_2\ar[r]_{\id\ot m\ot\id} &B_1AB_2\ar[r]_{\id\ot d} &B_1A\ar[r]_u &A
            }
        \end{align*}
        The right top subdiagram commutes similarly.
        So by universal property, there is a unique morphism $m':A\otd\limits_{B_1\ot B_2}A\to A$ in $\CC$ such that $m'\circ p_{A,A}=m$.

        \begin{enumerate}
            \item We need to check $m'$ is a morphism in $\hind{B_1}{\CC}{B_2}$, here we show $m'$ is a left $B_1$-module homomorphism which is the right subdiagram of the follow commutative diagram
            \begin{align*}
                \xymatrix{
                    (B_1 A) (B_1 B_2) A\ar@<0.5ex>[r]\ar@<-0.5ex>[r]\ar[d]_{u\ot\id} &(B_1 A) A\ar[r]^{p_{}}\ar[d]_{u\ot\id}\ar@/^2pc/[rr]^{\id\ot m} &B_1A\otd\limits_{B_1\ot B_2} A\ar@{-->}[d]^{u'}\ar[r]^{\id\ot m'} &B_1A\ar[d]^u\\
                    A (B_1 B_2) A\ar@<0.5ex>[r]\ar@<-0.5ex>[r] &A A\ar[r]_{p_{A,A}}\ar@/_2pc/[rr]_{m} &A\otd\limits_{B_1\ot B_2} A\ar[r]_{m'} &A
                }
            \end{align*}
            It is not hard to obtain $m'\circ u'\circ p=u\circ(\id\ot m')\circ p$ by reducing diagram.
            Since $p$ is an epi, we can also reduce it at both sides, which provide what we want.
            
            \item The proof that $m'$ is a right $B_2$-module homomorphism is similar.
            
            \item It is not hard to check $m'$ is associative, provided by the associativity of $m$.
            
            \item  For unitality, we consider the following diagram
            \begin{align*}
                \xymatrix{
                    (B_1 B_2) (B_1 B_2) A\ar@<0.5ex>[rr]\ar@<-0.5ex>[rr]\ar[dd]_{h'\ot\id} & &(B_1 B_2) A\ar[rr]^{p_{B_1B_2,A}}\ar[dd]_{h'\ot\id}\ar[dr]^{l^{12}_{A}} & &(B_1 B_2)\otd\limits_{B_1\ot B_2} A\ar@{-->}[dd]^{h'\otd\limits_{B_1\ot B_2}\id}\ar@{-->}[dl]^{\lambda}
                    \\
                     & & &A & 
                    \\
                    A (B_1 B_2) A\ar@<0.5ex>[rr]\ar@<-0.5ex>[rr] & &A A\ar[rr]_{p_{A,A}}\ar[ur]_m & &A\otd\limits_{B_1\ot B_2} A\ar@{-->}[ul]_{m'}
                }
            \end{align*}
            Since $h'$ now is $B_1$-$B_2$-bimodule homomorphism, by Theorem \ref{thm:B_1CB_2_to_B_1B_2CB_1B_2_is_fully_faithful}, $h'$ is a $B_1\ot B_2$-$B_1\ot B_2$-bimodule homomorphism.
            By definition of $h'\otd\limits_{B_1\ot B_2} \id_{A}$, we have $(h'\otd\limits_{B_1\ot B_2} \id_{A}) \circ p_{B_1B_2,A}= p_{A,A} \circ (h'\ot\id_{A})$, which means the right square commutes.
            Now bottom triangle commutes by above results, and hence $m'\circ (h'\otd\limits_{B_1\ot B_2} \id_{A}) \circ p_{B_1B_2,A}=m\circ (h'\ot\id)$.
            We claim that the left triangle commutes, as a consequence $m'\circ (h'\otd\limits_{B_1\ot B_2} \id_{A}) \circ p_{B_1B_2,A}=l^{12}_{A}$.
            Top triangle commutes obviously (by universal property), so we have $m'\circ (h'\otd\limits_{B_1\ot B_2} \id_{A}) \circ p_{B_1B_2,A}=\lambda \circ p_{B_1B_2,A}$.
            Since $p_{B_1B_2,A}$ is an epi, we now reach to $m'\circ (h'\otd\limits_{B_1\ot B_2} \id_{A})=\lambda$.

            To show the left triangle commutes, we expand it as follows:
            \begin{align*}
                \xymatrix{
                    B_1 B_2 A\ar[rr]^{\id\ot \beta_{B_2,A}}\ar[dd]_{\id\ot h\ot\id}\ar[ddr]^{\id\ot h\ot\id\ot h_2} & &B_1 A B_2\ar[d]^{\id\ot h\ot \id\ot h_2}\ar[rr]^{u\ot \id } & &A B_2\ar[dd]_{\id}\ar[ddl]_{h\ot \id\ot h_2}\\
                     & &B_1 A A B_2 B_2\ar[dr]^{u\ot \id} && \\
                    B_1 A B_2 A\ar[dd]_{u\ot \id}\ar[r]^{\id\ot h_2} &B_1 A B_2 A B_2\ar[r]^{u\ot\id}\ar[ur]^{\id\ot\beta_{B_2, A}\ot\id} & A B_2 A B_2\ar[r]_{\id\ot\beta_{B_2,A}\ot \id}\ar[dd]_{d\ot d} &A A B_2 B_2\ar[r]_{m\ot m_2} &A B_2\ar[dd]_d\\
                     &&&&\\
                    A B_2 A\ar[rr]_{d\ot\id}\ar[rruu]^{\id\ot h_2} & &A A\ar[rr]_m & &A
                }
            \end{align*}
            where the commutativity of the constituent regions follows from
            naturality of the braiding, the module axioms, and the defining
            universal properties of the relative tensor products.
        \end{enumerate}
    \end{itemize}
\end{proof}

\begin{lem}[\cite{Lur17}]\label{lem:lax_monoidal_functor_preserves_algebras}
    A lax monoidal functor $F:\CD\to \CC$ preserves $E_1$-algebras, i.e. it induces a functor
    \begin{align*}
        \Alg_{E_1}(U):\Alg_{E_1}(\CD)\to \Alg_{E_1}(\CC)
    \end{align*}
\end{lem}

\begin{lem}\label{lem:bimod_alg_is_alg_bimod}
    Let $A\in \Alg_{E_1}(\hind{B_1}{\CC}{B_2})$, then we have $A\in \hind{B_1}{\Alg_{E_1}(\CC)}{B_2}$.
\end{lem}
\begin{proof}
    Let $((A,u:B_1\ot A\to A, d:A\ot B_2\to A), m':A\otd\limits_{B_1\ot B_2} A\to A, h':B_1\ot B_2\to A)$ be an object in $\Alg_{E_1}(\hind{B_1}{\CC}{B_2})$.
    By Lemma \ref{lem:lax_monoidal_structure_on_forget_of_B_1CB_2} and Lemma \ref{lem:lax_monoidal_functor_preserves_algebras}, we have $(A,h:1\xrightarrow{h_1\ot h_2}B_1\ot B_2\xrightarrow{h'}A,m:A\ot A\xrightarrow{p_{A,A}}  A\otd\limits_{B_1\ot B_2} A\xrightarrow{m'}A)$ is an $E_1$-algebra in $\CC$.

    \begin{itemize}

\item We claim $u$ is the left $B_1$-action on $A$ in $\Alg_{E_1}(\CC)$:
        \begin{enumerate}
            \item $u$ is an $E_1$-algebra homomorphism:
            \begin{align*}
                \xymatrix{
                    B_1AB_1A
                        \ar[d]_{\id\ot\beta_{A,B_1}\ot\id}
                        \ar[dr]^{\id\ot u}
                    \\
                    B_1B_1AA
                        \ar[r]_{\id\ot u\ot \id}
                        \ar[d]_{\id\ot p_{A,A}}
                    &BAA
                        \ar[dr]^{u\ot\id}
                        \ar[d]_{\id\ot p_{A,A}}
                    \\
                    B_1B_1A\otd\limits_{B_1\ot B_2}A
                        \ar[d]_{\id\ot m'}
                        \ar[r]_{\id\ot u_{A\otd\limits_{B_1\ot B_2}A}}
                    &BA\otd\limits_{B_1\ot B_2}A
                        \ar[d]_{\id\ot m'}
                        \ar[dr]^{ u_{A\otd\limits_{B_1\ot B_2}A}}
                    &AA
                        \ar[d]^{p_{A,A}}
                    \\
                    B_1B_1A
                        \ar[d]_{m_1\ot\id}
                        \ar[r]_{\id\ot u}
                    &B_1A
                        \ar[dr]_{u}
                    &A\otd\limits_{B_1\ot B_2}A
                        \ar[d]^{m'}
                    \\
                    B_1A
                        \ar[rr]_u
                    &
                    &A
                }
            \end{align*}
            where the lower left square commutes by associativity of $u$; the middle two left and middle two right squares commutes by $m'$ is a left $B_1$-module homomorphism and universal property of relative tensor products;
            the upper left square commutes as we expand it as follows:
            \[\begin{tikzcd}
                &&& {B_1A 1A} &&&&&& \\
                \\
                {AB_1 1A} \\
                &&& {B_1AB_2A} && {B_1AA} \\
                &&& {AB_1A 1} \\
                {AB_1B_2 A} &&&&&&& {AA} && {A\otd\limits_{B_1\ot B_2}A} \\
                \\
                &&& {AB_1AB_2} && {AB_1A}
                \arrow["{\id\ot h_2\ot \id}"', from=1-4, to=4-4]
                \arrow["\cong"{description}, from=1-4, to=4-6]
                \arrow["{\beta_{A,B_1}\ot\id}", from=3-1, to=1-4]
                \arrow["{\id\ot \beta_{1,A}}"{description}, from=3-1, to=5-4]
                \arrow["{\id\ot h_2\ot\id}"', from=3-1, to=6-1]
                \arrow["{\id\ot d\ot \id}"', from=4-4, to=4-6]
                \arrow["{u\ot \id}", from=4-6, to=6-8]
                \arrow["{\id\ot h_2}"', from=5-4, to=8-4]
                \arrow["\cong"{description}, from=5-4, to=8-6]
                \arrow["{\beta_{A,B_1}\ot\id }"{description}, from=6-1, to=4-4]
                \arrow["{\id\ot \beta_{B_2,A}}"', from=6-1, to=8-4]
                \arrow["{p_{A,A}}"{description}, from=6-8, to=6-10]
                \arrow["{\id\ot d}"', from=8-4, to=8-6]
                \arrow["{\id\ot u}"', from=8-6, to=6-8]
            \end{tikzcd}\]
            here the vertical squares are naturality of braidings; vertical triangles are unitality of $d$; the bottom is coequalized by $p_{A,A}$.

            $u$ preserves unit is ensured by the unitality of $u$.
            \item $u$ satisfies the associativity of module action in $\Alg_{E_1}(\CC)$.
            Indeed, now $m_1:B_1\ot B_1\to B_1$ and $u$ are both morphism in $\Alg_{E_1}(\CC)$, so the associative diagram of $u$ holds in $\Alg_{E_1}(\CC)$
            \item Similarly, $u$ satisfies the unitality of module action in $\Alg_{E_1}(\CC)$, since $h_1\in\Alg_{E_1}(\CC)$.
        \end{enumerate}
        \item Similarly, we can show that $d$ is the right $B_2$-action on $A$ in $\Alg_{E_1}(\CC)$.
        And the compatibility of $u$ and $d$ in $\Alg_{E_1}(\CC)$ is automatically lifted from $\CC$.
    \end{itemize}
\end{proof}

\begin{thm}\label{thm:algebras_in_bimodule_categories}
    Let $B_1$ and $B_2$ be two $E_2$-algebras in an $E_{m}$-monoidal category $\CC$ ($m\geq 2$). 
    We have $\hind{B_1}{\Alg_{E_1}(\CC)}{B_2}\simeq \Alg_{E_1}(\hind{B_1}{\CC}{B_2})$ as categories.
\end{thm}
\begin{proof}
    ~
    \begin{itemize}
        \item On object level: $\Rightarrow$:
            By Lemma \ref{lem:alg_bimod_is_bimod_alg}. 
            $\Leftarrow$: By Lemma \ref{lem:bimod_alg_is_alg_bimod}

        \item On morphism level:
        for $A_1$ and $A_2\in \hind{B_1}{\Alg_{E_1}(\CC)}{B_2}$ and let $f:A_1\to A_2$ be a morphism in $\hind{B_1}{\Alg_{E_1}(\CC)}{B_2}$, i.e. $f$ is an $E_1$-algebra homomorphism and a left $B_1$-module homomorphism and a right $B_2$-module homomorphism.
        To show $f\in \Alg_{E_1}(\hind{B_1}{\CC}{B_2})$, we only need to show $f$ is an $E_1$-algebra homomorphism between $(A_1,m_1',h_1')$ and $ (A_2,m_2',h_2')$.
        \begin{align*}
                \xymatrix{
                    A_1 (B_1 B_2) A_1
                        \ar@<0.5ex>[r]
                        \ar@<-0.5ex>[r]
                        \ar[d]_{f\ot\id\ot f} 
                    &A_1 \ot A_1
                        \ar[r]^{p_{A_1,A_1}}
                        \ar[d]_{f\ot f}
                        \ar@/^2pc/[rr]^{m_1} 
                    &A_1 \otd\limits_{B_1\ot B_2}A_1
                        \ar@{-->}[d]^{f\otd\limits_{B_1\ot B_2}f}
                        \ar[r]^{m_1'} 
                    &A_1
                        \ar[d]^f
                    \\
                    A_2 (B_1 B_2) A_2
                        \ar@<0.5ex>[r]
                        \ar@<-0.5ex>[r] 
                    &A_2\ot A_2
                        \ar[r]_{p_{A_2,A_2}}
                        \ar@/_2pc/[rr]_{m_2} 
                    &A_2\otd\limits_{B_1\ot B_2} A_2
                        \ar[r]_{m_2'} 
                    &A_2
                }
            \end{align*}
            By reducing the diagram, we have $f\circ m_1'\circ p_{A_1,A_1}=m_2'\circ f\otd\limits_{B_1\ot B_2}f\circ p_{A_1,A_1}$. Since $p_{A_1,A_1}$ is epic, we can cancel $p_{A_1,A_1}$ at both sides and obtain what we want.

            Conversely, if $f\in\Alg_{E_1}(\hind{B_1}{\CC}{B_2})$, to show $f\in\hind{B_1}{\Alg_{E_1}(\CC)}{B_1}$, we only need to show $f$ is an $E_1$-algebra homomorphism between $(A_1,m_1,h_1)$ and $(A_2,m_2,h_2)$.
            Also consider above diagram, which holds obviously.
    \end{itemize}
    The two constructions above are inverse on objects and morphisms, and therefore define the claimed equivalence of categories.

\end{proof}

By above theorem, an $E_1$ $B_1$-$B_2$-bimodule $A\in \hind{B_1}{\Alg_{E_1}(\CC)}{B_2}$ can be viewed as an $E_1$-algebra in the monoidal category $\hind{B_1}{\CC}{B_2}\simeq \hind{B_1}{\Alg_{E_0}(\CC)}{B_2}$. 
Hence we can consider the $A_1$-$A_2$-bimodules in $\hind{B_1}{\Alg_{E_0}(\CC)}{B_2}$ which leads to the category $\hind{A_1}{\hind{B_1}{\Alg_{E_{0}}(\CC)}{B_2}}{A_2}$.

\begin{lem}\label{lem:bibimod_to_bimod_in_bimod}
    Let $M$ be an $E_0$ $B_1$-$B_2$-bi-$A_1$-$A_2$-bimodule, then we have $M\in \hind{A_1}{(\hind{B_1}{\Alg_{E_{0}}(\CC)}{B_2})}{A_2}$.
\end{lem}
\begin{proof}
    Let $(M,l,r)$ be an $E_0$ $A_1$-$A_2$-bimodule $M$ with respect to $B_1$, $B_2$.

    \begin{itemize}
        \item First we construct an object $(M,u^l_M,d^r_M)$ in $\hind{B_1}{\CC}{B_2}$ where 
        \begin{align*}
            u^l_M:=B_1\ot M\xrightarrow{\id_{B_1}\ot h_1\ot\id_M}B_1\ot A_1\ot M\xrightarrow{u_1\ot\id_M}A_1\ot M\xrightarrow{l}M\\
            d^r_M:=M\ot B_2\xrightarrow{\id_M\ot h_2\otimes\id_{B_2}}M\ot A_2\ot B_2\xrightarrow{\id_M\ot d_2}M\ot A_2\xrightarrow{r} M
        \end{align*}
        By Proposition \ref{prp:u^l_d^r_are_compatible}, $u^l_M$ and $d^r_M$ are compatible $B_1$-$B_2$-bimodule actions.

        \item Now we need a morphism $l':A_1\otd\limits_{B_1\ot B_2}M\to M$ in $\hind{B_1}{\CC}{B_2}$ induced by $l:A_1\ot M\to M$
        \begin{align*}
            \xymatrix{
                A_1\ot (B_1\ot B_2)\ot M
                    \ar@<0.5ex>[r]^{r^{12}_{A_1}\ot\id_M}
                    \ar@<-0.5ex>[r]_{\id\ot l^{12}_M} 
                &A_1\ot M\ar[r]^{p_{A_1,M}}
                    \ar[dr]_l 
                &A_1\otd\limits_{B_1\ot B_2}M
                    \ar@{-->}[d]^{l'}
                \\
                & 
                &M
            }
        \end{align*}
        We prove that $l$ coequalizes the two displayed morphisms.  Let
        $h_1':B_1\ot B_2\to A_1$ be the unit of $A_1$ as an algebra in
        $\hind{B_1}{\CC}{B_2}$.  Restricting the left $A_1$-action along
        $h_1'$ gives a left $B_1\ot B_2$-action
        \begin{align*}
            l\circ(h_1'\ot\id_M):B_1\ot B_2\ot M\longrightarrow M.
        \end{align*}
        Its $B_1$- and $B_2$-restrictions are $u_M^l$ and $d_M^{l,-}$,
        respectively.  By the defining bi-bimodule coherence,
        $d_M^{l,-}=d_M^r$.  Hence Theorem
        \ref{thm:B_1CB_2_to_B_1B_2CB_1B_2_is_fully_faithful}, applied to
        the identity of the underlying object $M$, identifies this restricted
        action with the action induced by $(u_M^l,d_M^r)$:
        \begin{align*}
            l\circ(h_1'\ot\id_M)=l_M^{12}.
        \end{align*}
        Moreover, the right-unit computation for $A_1$ in Lemma
        \ref{lem:alg_bimod_is_bimod_alg} gives
        \begin{align*}
            m_1\circ(\id_{A_1}\ot h_1')=r_{A_1}^{12}.
        \end{align*}
        We now calculate, as morphisms $A_1\ot B_1\ot B_2\ot M\to M$,
        \begin{align*}
            l\circ(r_{A_1}^{12}\ot\id_M)
            &=l\circ(m_1\ot\id_M)
            \circ(\id_{A_1}\ot h_1'\ot\id_M)
            &&\text{by the preceding right-unit identity}\\
            &=l\circ(\id_{A_1}\ot l)
            \circ(\id_{A_1}\ot h_1'\ot\id_M)
            &&\text{by associativity of $l$}\\
            &=l\circ\bigl(\id_{A_1}\ot
            [l\circ(h_1'\ot\id_M)]\bigr)\\
            &=l\circ(\id_{A_1}\ot l_M^{12})
            &&\text{by the restriction identity above.}
        \end{align*}
        Thus $l$ is $B_1\ot B_2$-balanced.  By the universal property of the
        relative tensor product, there exists a unique morphism $l':A_1\otd_{B_1\ot B_2} M\longrightarrow M$
        such that $l'\circ p_{A_1,M}=l$.

        We next check that $l'$ is a morphism in
        $\hind{B_1}{\CC}{B_2}$.  Precomposing the left $B_1$-linearity
        diagram for $l'$ with $\id_{B_1}\ot p_{A_1,M}$ reduces it to the
        corresponding diagram for $l$, which commutes by the definition of
        $u_M^l$ and associativity of the $A_1$-action.  Since tensoring
        preserves the defining coequalizer, $\id_{B_1}\ot p_{A_1,M}$ is
        epic, so $l'$ is left $B_1$-linear.  The same argument, using
        $d_M^{l,-}=d_M^r$, proves right $B_2$-linearity.

        The associativity diagram for $l'$ is verified by precomposing it
        with the canonical quotient map from $A_1\ot A_1\ot M$.  After using
        $m_1'\circ p_{A_1,A_1}=m_1$ and
        $l'\circ p_{A_1,M}=l$, its two sides become
        \begin{align*}
            l\circ(m_1\ot\id_M),
            \qquad
            l\circ(\id_{A_1}\ot l),
        \end{align*}
        which are equal by associativity of $l$.  For the unit axiom,
        precomposition with $p_{B_1\ot B_2,M}:B_1\ot B_2\ot M\to B_1\ot B_2\otd_R M$ gives
        \begin{align*}
            [l'\circ(h_1'\otd_{B_1\ot B_2}\id_M)]\circ p_{B_1\ot B_2,M}
            =l\circ(h_1'\ot\id_M)
            =l_M^{12}
            =\lambda_M^{B_1\ot B_2}\circ p_{B_1\ot B_2,M}.
        \end{align*}
        Since $p_{B_1\ot B_2,M}$ is epic, we obtain
        $l'\circ(h_1'\otd_{B_1\ot B_2}\id_M)=\lambda_M^{B_1\ot B_2}$.  Therefore $l'$ is a
        left $A_1$-module action in $\hind{B_1}{\CC}{B_2}$.

        \item Similarly, the right action $r:M\ot A_2\to M$ satisfies
        \begin{align*}
            r\circ(r_M^{12}\ot\id_{A_2})
            =r\circ(\id_M\ot l_{A_2}^{12})
            :M\ot B_1\ot B_2\ot A_2\longrightarrow M.
        \end{align*}
        Indeed, precompose the associativity identity
        $r\circ(r\ot\id_{A_2})=r\circ(\id_M\ot m_2)$ with
        $\id_M\ot h_2'\ot\id_{A_2}$ and use the two unit-specialization
        identities
        \begin{align*}
            r\circ(\id_M\ot h_2')=r_M^{12},
            \qquad
            m_2\circ(h_2'\ot\id_{A_2})=l_{A_2}^{12}.
        \end{align*}
        Consequently, there is a unique
        \begin{align*}
            r':M\otd_{\sB_1\ot\sB_2} A_2\longrightarrow M
        \end{align*}
        such that $r'\circ p_{M,A_2}=r$.  The same quotient-map arguments
        as above show that $r'$ is a morphism in
        $\hind{B_1}{\CC}{B_2}$ and satisfies the associative and unital right
        $A_2$-module axioms.

        Finally, identify the two bracketings by the canonical relative
        tensor associator.  Precomposing the compatibility diagram for
        $l'$ and $r'$ with the canonical quotient map from
        $A_1\ot M\ot A_2$ reduces its two boundary composites to
        \begin{align*}
            r\circ(l\ot\id_{A_2}),
            \qquad
            l\circ(\id_{A_1}\ot r).
        \end{align*}
        These are equal because $(M,l,r)$ is an ordinary
        $A_1$-$A_2$-bimodule.  The quotient map is epic, so $l'$ and $r'$ are
        compatible.  Hence
        \begin{align*}
            (M,u_M^l,d_M^r,l',r')
            \in\hind{A_1}{\hind{B_1}{\Alg_{E_0}(\CC)}{B_2}}{A_2}.
        \end{align*}
    \end{itemize}
\end{proof}

\begin{lem}\label{lem:bimod_in_bimod_to_bibimod}
    There is a functor
    \begin{align}
        F:\hind{A_1}{(\hind{B_1}{\CC}{B_2})}{A_2}&\to \hind{A_1}{\vind{B_1}{\CC}{B_2}}{A_2}\\
        (M,u_M,d_M,l',r')&\mapsto (M,l,r)\\
        f&\mapsto f
    \end{align}
\end{lem}
\begin{proof}
    Let $(M,u_M,d_M,l',r')\in \hind{A_1}{(\hind{B_1}{\CC}{B_2})}{A_2}$, by Lemma \ref{lem:lax_monoidal_functor_preserves_modules}, we have $(M,l,r)\in \hind{A_1}{\Alg_{E_0}(\CC)}{A_2}$.
    Now we need to show that $(M,l,r)$ satisfies the bi-bimodule coherence conditions.

    We have the following commutative diagrams
    \begin{align*}
        \xymatrix{
            B_1M 
                \ar[r]^{\id\ot h_{B_2}\ot \id}
                \ar[dr]_{i_{11}\ot \id}
            &B_1B_2M
                \ar[r]^{p_{B_1B_2,M}}
                \ar[d]_{h_1'\ot\id}
                \ar@/^4pc/[drr]^{l^{12}_M}
            &B_1B_2\otd\limits_{B_1B_2}M
                \ar[dr]^{\lambda_M^{B_1B_2}}
                \ar[d]_{h_1'\otd\limits_{B_1B_2}\id}
            \\
            &A_1M
                \ar[r]_{p_{A_1,M}}
                \ar@/_2pc/[rr]_l
            &A_1\otd\limits_{B_1B_2}M\
                \ar[r]_{l'}
            &M
        }
    \end{align*}
    where the right triangle is the unitality of $l'$ in $\hind{B_1}{\CC}{B_2}$; the square commutes bt definition of $h_1'\otd\limits_{B_1B_2}\id$; others are by definition and universal properties.
    Now we have $d_M=l^{12}_M\circ (\id\ot h_{B_2}\ot \id)=l\circ (i_{11}\ot\id)=d_M^{l,-}$.
    By replacing the upper left node with $B_2\ot M$, we obtain $u_M^l=u_M$.
    
    Similarly, we can consider the commutative diagram for $h_2'$, and will obtain $u_M=u_M^{r,-}$ and $d_M=d_M^r$.
    As a consequence, $u_M^l=u_M^{r,-}$ and $d_M^{l,-}=d_M^r$,so we are done.
\end{proof}

\begin{thm}\label{thm:bibimod_equivalent_to_bimod_in_bimod}
    There is an equivalence 
    \begin{align*}
        \hind{A_1}{\vind{B_1}{\CC}{B_2}}{A_2}\simeq \hind{A_1}{(\hind{B_1}{\CC}{B_2})}{A_2}
    \end{align*}
    of categories.
\end{thm}
\begin{proof}
    By Lemma \ref{lem:bibimod_to_bimod_in_bimod} and Lemma \ref{lem:bimod_in_bimod_to_bibimod}.
    It is obvious that the two constructions are inverse to each other.
\end{proof}

\begin{dfn}
    Two $E_2$-algebras are {\bf 2-Morita equivalent} if they are equivalent objects in $\Mrt_{E_2}(\CC)$. 
\end{dfn}

\begin{figure}[H]
    \centering

    \tikzset{
        bulkBOne/.style={fill=gray!42},
        bulkBTwo/.style={fill=gray!18},
        wall/.style={draw=black, thick},
        defectline/.style={
            draw=black,
            thick,
            dash pattern=on 4pt off 3pt
        },
        junction/.style={
            fill=white,
            draw=black,
            line width=0.7pt
        }
    }

    $\vfus{A}{\otc{B_2}}{A'}\xrightarrow{\ M\ }B_1,
    \qquad
    B_1\xrightarrow{\ M'\ }\vfus{A}{\otc{B_2}}{A'}$

    \vspace{0.35cm}

    \begin{minipage}[c]{0.41\linewidth}
        \centering
        \begin{tikzpicture}[x=0.78cm,y=0.78cm]
            \path[bulkBOne] (-3,-1.55) rectangle (3,1.55);

            \path[bulkBTwo]
                (-3,0.58) -- (-1.55,0.58)
                .. controls (-1.25,0.58) and (-1.10,0.32) .. (-1.00,0.09)
                -- (-1.00,-0.09)
                .. controls (-1.10,-0.32) and (-1.25,-0.58) .. (-1.55,-0.58)
                -- (-3,-0.58) -- cycle;
            \path[bulkBTwo]
                (3,0.58) -- (1.55,0.58)
                .. controls (1.25,0.58) and (1.10,0.32) .. (1.00,0.09)
                -- (1.00,-0.09)
                .. controls (1.10,-0.32) and (1.25,-0.58) .. (1.55,-0.58)
                -- (3,-0.58) -- cycle;

            \draw[wall] (-3,0.58) -- (-1.55,0.58)
                .. controls (-1.25,0.58) and (-1.10,0.32) .. (-1.00,0.09);
            \draw[wall] (-3,-0.58) -- (-1.55,-0.58)
                .. controls (-1.25,-0.58) and (-1.10,-0.32) .. (-1.00,-0.09);
            \draw[wall] (3,0.58) -- (1.55,0.58)
                .. controls (1.25,0.58) and (1.10,0.32) .. (1.00,0.09);
            \draw[wall] (3,-0.58) -- (1.55,-0.58)
                .. controls (1.25,-0.58) and (1.10,-0.32) .. (1.00,-0.09);

            \draw[defectline] (-0.91,0) -- (0.91,0);
            \fill[black] (-1,0) circle (0.10);
            \fill[black] (1,0) circle (0.10);

            \node at (0,1.12) {$B_1$};
            \node at (0,-1.12) {$B_1$};
            \node at (-2.15,0) {$B_2$};
            \node at (2.15,0) {$B_2$};
            \node at (-1.72,0.78) {$A$};
            \node at (-1.72,-0.78) {$A'$};
            \node at (1.72,0.78) {$A$};
            \node at (1.72,-0.78) {$A'$};
            \node at (-1,-0.48) {$M$};
            \node at (0.9,-0.48) {$M'$};
        \end{tikzpicture}
    \end{minipage}
    \hfill
    \begin{minipage}[c]{0.12\linewidth}
        \centering
        $\Longrightarrow$
    \end{minipage}
    \hfill
    \begin{minipage}[c]{0.41\linewidth}
        \centering
        \begin{tikzpicture}[x=0.78cm,y=0.78cm]
            \path[bulkBOne] (-3,-1.55) rectangle (3,1.55);
            \path[bulkBTwo] (-3,-0.58) rectangle (3,0.58);
            \draw[wall] (-3,0.58) -- (3,0.58);
            \draw[wall] (-3,-0.58) -- (3,-0.58);
            \node at (0,1.08) {$B_1$};
            \node at (0,0) {$B_2$};
            \node at (0,-1.08) {$B_1$};
            \node[above] at (1.35,0.58) {$A$};
            \node[below] at (1.35,-0.58) {$A'$};
        \end{tikzpicture}
    \end{minipage}

    \vspace{0.12cm}

    $M\otc{B_1}M'\cong_3 \vfus{A}{\otc{B_2}}{A'}$

    \vspace{0.35cm}

    \begin{minipage}[c]{0.41\linewidth}
        \centering
        \begin{tikzpicture}[x=0.78cm,y=0.78cm]
            \path[bulkBOne] (-3,-1.55) rectangle (3,1.55);
            \path[bulkBTwo]
                (-1.00,0.09)
                .. controls (-0.75,0.75) and (0.75,0.75) .. (1.00,0.09)
                -- (1.00,-0.09)
                .. controls (0.75,-0.75) and (-0.75,-0.75) .. (-1.00,-0.09)
                -- cycle;

            \draw[wall] (-1.00,0.09)
                .. controls (-0.75,0.75) and (0.75,0.75) .. (1.00,0.09);
            \draw[wall] (-1.00,-0.09)
                .. controls (-0.75,-0.75) and (0.75,-0.75) .. (1.00,-0.09);
            \draw[defectline] (-3,0) -- (-1.09,0);
            \draw[defectline] (1.09,0) -- (3,0);
            \fill[black] (-1,0) circle (0.10);
            \fill[black] (1,0) circle (0.10);

            \node at (-2.20,1.12) {$B_1$};
            \node at (-2.20,-1.12) {$B_1$};
            \node at (0,0) {$B_2$};
            \node at (0.42,0.72) {$A$};
            \node at (0.42,-0.72) {$A'$};
            \node at (-1.18,0.42) {$M'$};
            \node at (1.18,0.42) {$M$};
        \end{tikzpicture}
    \end{minipage}
    \hfill
    \begin{minipage}[c]{0.12\linewidth}
        \centering
        $\Longrightarrow$
    \end{minipage}
    \hfill
    \begin{minipage}[c]{0.41\linewidth}
        \centering
        \begin{tikzpicture}[x=0.78cm,y=0.78cm]
            \path[bulkBOne] (-3,-1.55) rectangle (3,1.55);
            \node at (0,0) {$B_1$};
        \end{tikzpicture}
    \end{minipage}

    \vspace{0.12cm}

    $M'\otd\limits_{\vfus{A}{\otc{B_2}}{A'}}M\cong_3 B_1$

    \caption{This figure shows 2-Morita equivalent data: the inverse 2-morphisms $M$ and $M'$, together with the
    3-isomorphisms identifying their two composites with the identity
    2-morphisms.}
    \label{fig:four-local-configurations}
\end{figure}
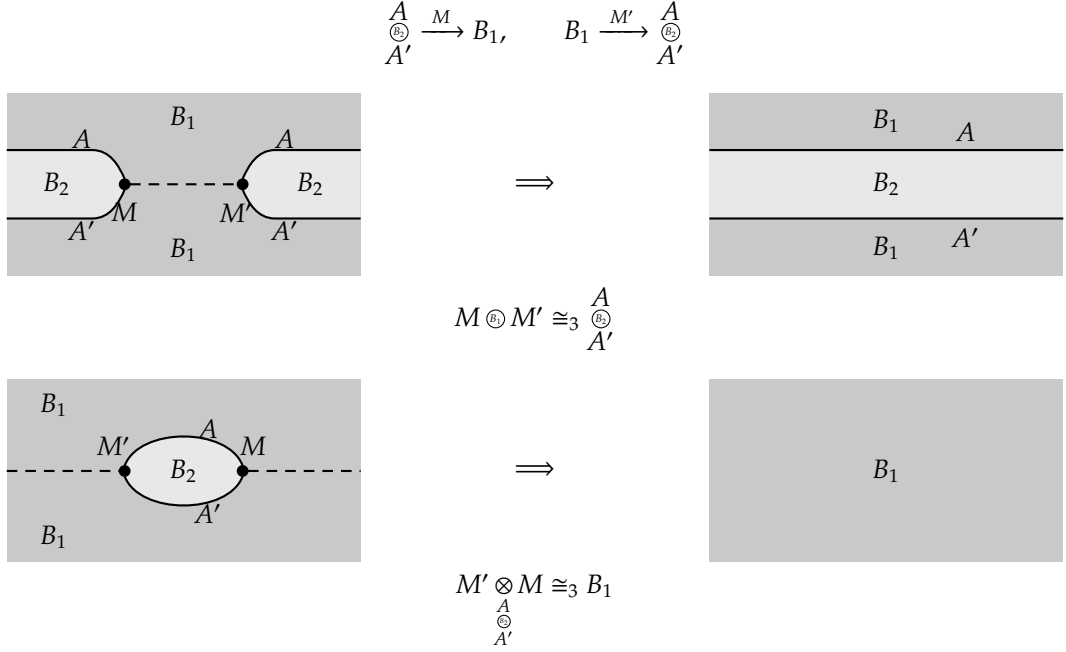

\subsection{The Module Realization Functor}\label{sec:functor_to_module_cat}
We now construct the module realization of the bi-bimodule model and identify its categories of \(2\)-morphisms via an \(E_2\) Eilenberg--Watts theorem.
  \begin{center}
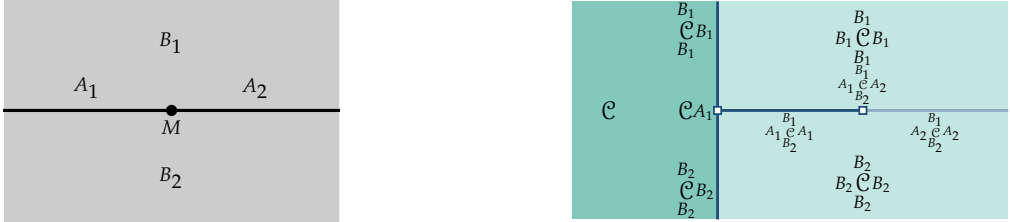
    \begin{minipage}[c]{0.42\textwidth}
      \centering
      \resizebox{0.68\linewidth}{!}{%
      \begin{tikzpicture}
            \filldraw[fill=gray!40,draw=white] (-1.5,1) rectangle(1.5,-1);
            \draw[thick] (-1.5,0) -- (1.5,0);
            \fill[black] (0,0) circle (0.05);

            \node at(0,0.6){\tiny $B_1$};
            \node at(0,-0.6){\tiny $B_2$};

            \node at (0.75,0.2) {\tiny $A_2$};
            \node at (-0.75,0.2) {\tiny $A_1$};

            \node at (0,-0.15) {\tiny $M$};
  \end{tikzpicture}}
    \end{minipage}
    \hfill
    \begin{minipage}[c]{0.54\textwidth}
      \centering
      \resizebox{0.68\linewidth}{!}{%
            \begin{tikzpicture}
        \path[fill=hmDeepCyan] (-2,-1.5) rectangle(0,1.5);
        \path[fill=hmMidCyan] (0,-1.5) rectangle(4,1.5);
        \draw[draw=hmLineGreen,line width=1.1pt] (0,-1.5) -- (0,1.5);
        \draw[draw=hmLineGreen, line width=1.1pt] (0,0) -- (2,0);
         \filldraw[draw= hmLineLightBlue,line width=1.1pt] (2,0) -- (4,0);
        \filldraw[fill=white,draw=hmLineGreen,line width=0.7pt] (-0.05,-0.05) rectangle(0.05,0.05);
        \filldraw[fill=white,draw=hmLineGreen,line width=0.7pt] (1.95,-0.05) rectangle(2.05,0.05);

        \node at(-1.5,0){$\CC$};
        \node at(-0.3,1.1){$\hind{}{\vind{B_1}{\CC}{B_1}}{B_1}$};
        \node at(-0.3,-1.1){$\hind{}{\vind{B_2}{\CC}{B_2}}{B_2}$};
        \node at(-0.3,0){$\hind{}{\CC}{A_1}$};
        \node at(2,1){$\hind{B_1}{\vind{B_1}{\CC}{B_1}}{B_1}$};
        \node at(2,-1){$\hind{B_2}{\vind{B_2}{\CC}{B_2}}{B_2}$};
        \node at(1,-0.3){\tiny$\hind{A_1}{\vind{B_1}{\CC}{B_2}}{A_1}$};
        \node at(2,0.35){\tiny$\hind{A_1}{\vind{B_1}{\CC}{B_2}}{A_2}$};
        \node at(3,-0.3){\tiny$\hind{A_2}{\vind{B_1}{\CC}{B_2}}{A_2}$};
      \end{tikzpicture}}
    \end{minipage}
  \end{center}

We now review a tricategory as follows:
\begin{dfn}[\cite{DR18}]
    Let $\CC$ be an $E_m$-monoidal $n$-category with $m\geq 2$ .
    The tricategory $\Mrt_{E_1}(\LMod(\CC))$ consists of 
    \begin{itemize}
        \item objects are $E_1$-monoidal left $\CC$-module categories;
        \item a 1-morphism between two $E_1$-monoidal $\CC$-modules $\CA_1$ and $\CA_2$ is a $\CA_1$-$\CA_2$-bimodule category in $\LMod(\CC)$; 
        \item 2-morphisms between two $\CA_1$-$\CA_2$-bimodules $\CM$, $\CN$ are $\CA_1$-$\CA_2$-bimodule functors in $\LMod(\CC)$.
        \item a 3-morphism between two $\CA_1$-$\CA_2$-bimodule functors $F$ and $G$, is a $\CA_1$-$\CA_2$-bimodule natural transformation in $\LMod(\CC)$. 
    \end{itemize}
\end{dfn}

We want to construct a 3-functor between two tricategories $\Mrt_{E_2}(\CC)$ and $\Mrt_{E_1}(\LMod(\CC))$.
To construct the module realization on Morita data, we first recall the monoidal structure on the ambient category of \(\CC\)-module categories.

\begin{prp}[\cite{DR18}]
    Let $\CC$ be $E_2$-monoidal, then $\LMod(\CC)$ admits a monoidal structure $(\LMod(\CC),\btc{\CC},\CC)$ where $\btc{\CC}$ is the relative tensor product over $\CC$.
\end{prp}
\begin{proof}
    Let $(\CM,\odotd_{\CM})$, $(\CN,\odotd_{\CN})$ and $(\CP,\odot_{\CP})$ be objects in $\LMod(\CC)$, it is not hard to see that each $\odot$ is cocontinuous on each variables, so we have $\odotd_{\CM}:\CC\btd \CM\to \CM$ respectively. (We will omit the underline notation.)
    
    then the following coequalizer exists:
    \begin{align}
        \xymatrix{
            \CM\btd \CC\btd\CC \btd \CN\ar@<0.8ex>[r]\ar[r]\ar@<-0.8ex>[r] &\CM\btd \CC\btd \CN\ar@<0.5ex>[r]\ar@<-0.5ex>[r] &\CM\btd \CN\ar[r]\ar[dr] &\CM\btc{\CC}\CN\ar@{-->}[d]^{\exists !}\\
             & & &\CP
        }
    \end{align}
    where the left $\CC$-module action is defined by:
    \begin{align*}
        \odot_{\CM\btc{\CC}\CN}:\CC\btd(\CM\btc{\CC}\CN)&\to \CM\btc{\CC}\CN\\
        c\btd m\btc{\CC}n&\mapsto (c\odot_{\CM}m)\btc{\CC}n
    \end{align*}
\end{proof}

\begin{thm}[\cite{Lur17}]
    The functor
    \begin{align*}
        \underline{\mathrm{Mod}_1}:\Alg_{E_1}(\CC)&\to \LMod(\CC)\\
        A&\mapsto \hind{}{\CC}{A}
    \end{align*}
    is monoidal.
\end{thm}

his monoidality is precisely what allows the ordinary module construction to be applied one Morita level higher.
\begin{lem}\label{lem:Mod1_induces_functors}
    The monoidal functor $\underline{\mathrm{Mod}_1}$ induces functor
    \begin{align}
        \Alg_{E_1}(\underline{\mathrm{Mod}_1}):\Alg_{E_2}(\CC)\to \Alg_{E_1}(\LMod(\CC))
    \end{align}
    and functors
    \begin{align}
        \underline{\mathrm{Mod}_1}_{B_1,B_2}: \hind{B_1}{\Alg_{E_1}(\CC)}{B_2}\to \hind{\hind{}{\CC}{B_1}}{\LMod(\CC)}{\hind{}{\CC}{B_2}}
    \end{align}
\end{lem}

\begin{prp}[\cite{DR18}]
    Let $B$ be an $E_2$-algebra in $\CC$, then $\hind{}{\CC}{B}$ is an $E_1$-algebra in $\LMod(\CC)$.
\end{prp}
\begin{proof}
    The tensor functor $\otimes_{B}:\hind{}{\CC}{B}\times\hind{}{\CC}{B}\to \hind{}{\CC}{B}$ is cocontinuous on each variable, so there exists a unique (up to natural isomorphism) functor $\otimes_{B}:\hind{}{\CC}{B}\btd\hind{}{\CC}{B}\to\hind{}{\CC}{B}$ such that the following diagram commutes:
    \begin{align*}
        \xymatrix{
            \hind{}{\CC}{B}\td \hind{}{\CC}{B}\ar[r]^{\btd}\ar[dr]_{\otd_{B}} &\hind{}{\CC}{B}\btd\hind{}{\CC}{B}\ar@{-->}[d]^{\otd_{B}}\\
             &\hind{}{\CC}{B}
        }
    \end{align*}
    
    the tensor functor $\otimes_{B}:\hind{}{\CC}{B}\btd\hind{}{\CC}{B}\to\hind{}{\CC}{B}$ satisfies the following coequalizer diagram, so by universal property, there exists a unique functor $\otimes_{B}:\hind{}{\CC}{B}\btc{\CC}\hind{}{\CC}{B}\to\hind{}{\CC}{B}$.
    \begin{align}
        \xymatrix{
            \hind{}{\CC}{B}\btd\CC\btd\CC\btd\hind{}{\CC}{B}\ar@<0.8ex>[r]\ar[r]\ar@<-0.8ex>[r] &\hind{}{\CC}{B}\btd \CC\btd \hind{}{\CC}{B}\ar@<0.5ex>[r]\ar@<-0.5ex>[r] &\hind{}{\CC}{B}\btd \hind{}{\CC}{B}\ar[r]\ar[dr]_{\otimes_{B}} &\hind{}{\CC}{B}\btc{\CC}\hind{}{\CC}{B}\ar@{-->}[d]^{\otimes_{B}}\\
             & & &\hind{}{\CC}{B}
        }
    \end{align}
\end{proof}

And thus, 
\begin{thm}
    Let $B_1$ and $B_2$ be two $E_2$-algebras in $\CC$ and let $A\in\hind{B_1}{\Alg_{E_1}(\CC)}{B_2}$.
    Then $\hind{}{\CC}{A}$ is a $\hind{}{\CC}{B_1}$-$\hind{}{\CC}{B_2}$-bimodule in $\LMod(\CC)$.
\end{thm}
\begin{proof}
   This is the image of \(A\) under the functor $\underline{\mathrm{Mod}_1}_{B_1,B_2}$ of Lemma \ref{lem:Mod1_induces_functors}.
\end{proof}

Thus the module realization is already defined on objects and \(1\)-morphisms. To understand its higher morphisms, we next compare $B_1$-$B_2$bi-$A$-modules $\hind{}{\vind{B_1}{\CC}{B_2}}{A}\simeq \hind{}{\CC}{A}$ with ordinary \(A\)-modules in \(\CC\).
\begin{thm}\label{thm:modules_in_bimodule_categories}
    Let $B_1$ and $B_2$ be two $E_2$-algebras and let $A\in \hind{B_1}{\Alg_{E_1}(\CC)}{B_2}$, then we have 
    \begin{align*}
        \hind{B_1\ot B_2}{\vind{B_1}{\CC}{B_2}}{A}=:\hind{}{\vind{B_1}{\CC}{B_2}}{A}\simeq \hind{}{\CC}{A}
    \end{align*}
\end{thm}
\begin{proof}
    By definition, an object on the left-hand side is
    an ordinary $B_1\ot B_2$-$A$-bimodule $(M,l,r)$ whose $B_1$- and
    $B_2$-actions induced from $l$ agree with those induced from $r$ in the
    coherence diagram defining a bi-bimodule.  We first construct the
    forgetful functor
    \begin{align*}
        F:\hind{}{\vind{B_1}{\CC}{B_2}}{A}&\to \hind{}{\CC}{A}\\
        (M,l,r)&\mapsto (M,r),
    \end{align*}
    which is the identity on the underlying morphisms.

    Conversely, let $(M,r)\in\hind{}{\CC}{A}$.  
    By Proposition \ref{prp:u^r-_d^r_are_compatible}, we have $(M,u^{r,-}_M,d^r_M)\in \hind{B_1}{\CC}{B_2}$.
    Hence the construction above gives a left $B_1\ot B_2$-action
    \begin{align*}
        l^{12}_M:
        B_1\ot B_2\ot M
        \xrightarrow{\id_{B_1}\ot\beta_{B_2,M}}
        B_1\ot M\ot B_2
        \xrightarrow{u_M\ot\id_{B_2}}
        M\ot B_2
        \xrightarrow{d_M}M.
    \end{align*}
    It remains to show that $l_M^{12}$ is compatible with the right
    $A$-action.  We first record the two compatibility identities
    \begin{align*}
        r\circ(u_M\ot\id_{A})
        &=u_M\circ(\id_{B_1}\ot r)
        :B_1\ot M\ot A\longrightarrow M,
        \qquad\text{(U)}\\
        r\circ(d_M\ot\id_{A})
        &=d_M\circ(r\ot\id_{B_2})
        \circ(\id_M\ot\beta_{B_2,A})
        :M\ot B_2\ot A\longrightarrow M.
        \qquad\text{(D)}
    \end{align*}
    We now explain exactly which identities imply (U) and (D).  Suppress
    associators and unitors.  The sliding identities supplied by the
    multiplicativity of $u$ and $d$ are
    \begin{align*}
        m_{A}\circ(i_1\ot\id_{A})
        &=m_{A}\circ(\id_{A}\ot i_1)
        \circ\beta^{-1}_{A,B_1},
        &&\text{(C$_1$)}\\
        m_{A}\circ(i_2\ot\id_{A})
        &=m_{A}\circ(\id_{A}\ot i_2)
        \circ\beta_{B_2,A}.
        &&\text{(C$_2$)}
    \end{align*}
    Identity (U) is the following equality of morphisms
    $B_1\ot M\ot A\to M$:
    \begin{align*}
        &r\circ(u_M\ot\id_{A})\\
        &\quad=
        r\circ(r\ot\id_{A})
        \circ(\id_M\ot i_1\ot\id_{A})
        \circ(\beta^{-1}_{M,B_1}\ot\id_{A})
        &&\text{by the definition of $u_M$}\\
        &\quad=
        r\circ(\id_M\ot m_{A})
        \circ(\id_M\ot i_1\ot\id_{A})
        \circ(\beta^{-1}_{M,B_1}\ot\id_{A})
        &&\text{by associativity of $r$}\\
        &\quad=
        r\circ(\id_M\ot m_{A})
        \circ(\id_M\ot\id_{A}\ot i_1)
        \circ(\id_M\ot\beta^{-1}_{A,B_1})
        \circ(\beta^{-1}_{M,B_1}\ot\id_{A})
        &&\text{by (C$_1$)}\\
        &\quad=
        r\circ(\id_M\ot m_{A})
        \circ(\id_M\ot\id_{A}\ot i_1)
        \circ\beta^{-1}_{M\ot A,B_1}
        &&\text{by the hexagon axiom}\\
        &\quad=
        r\circ(\id_M\ot i_1)
        \circ(r\ot\id_{B_1})
        \circ\beta^{-1}_{M\ot A,B_1}
        &&\text{by associativity of $r$}\\
        &\quad=
        r\circ(\id_M\ot i_1)
        \circ\beta^{-1}_{M,B_1}
        \circ(\id_{B_1}\ot r)
        &&\text{by naturality of $c^-_{B_1,-}$}\\
        &\quad=u_M\circ(\id_{B_1}\ot r)
        &&\text{by the definition of $u_M$}.
    \end{align*}

    Identity (D) is the following equality of morphisms
    $M\ot B_2\ot A\to M$:
    \begin{align*}
        &r\circ(d_M\ot\id_{A})\\
        &\quad=
        r\circ(r\ot\id_{A})
        \circ(\id_M\ot i_2\ot\id_{A})
        &&\text{by the definition of $d_M$}\\
        &\quad=
        r\circ(\id_M\ot m_{A})
        \circ(\id_M\ot i_2\ot\id_{A})
        &&\text{by associativity of $r$}\\
        &\quad=
        r\circ(\id_M\ot m_{A})
        \circ(\id_M\ot\id_{A}\ot i_2)
        \circ(\id_M\ot\beta_{B_2,A})
        &&\text{by (C$_2$)}\\
        &\quad=
        r\circ(r\ot\id_{A})
        \circ(\id_M\ot\id_{A}\ot i_2)
        \circ(\id_M\ot\beta_{B_2,A})
        &&\text{by associativity of $r$}\\
        &\quad=
        r\circ(\id_M\ot i_2)
        \circ(r\ot\id_{B_2})
        \circ(\id_M\ot\beta_{B_2,A})
        &&\text{by bifunctoriality of $\ot$}\\
        &\quad=
        d_M\circ(r\ot\id_{B_2})
        \circ(\id_M\ot\beta_{B_2,A})
        &&\text{by the definition of $d_M$}.
    \end{align*}

    Thus both (U) and (D) follow from the algebra-homomorphism conditions
    with the indicated, opposite braiding conventions on the two sides.

    We now prove the required bimodule square
    \begin{align*}
        \xymatrix{
            B_1\ot B_2\ot M\ot A
                \ar[r]^{l_M^{12}\ot\id_{A}}
                \ar[d]_{\id_{B_1\ot B_2}\ot r}
            &M\ot A\ar[d]^r\\
            B_1\ot B_2\ot M\ar[r]_{l_M^{12}}
            &M.
        }
    \end{align*}
    The two boundary composites are morphisms
    \begin{align*}
        B_1\ot B_2\ot M\ot A\longrightarrow M,
    \end{align*}
    and they are equal by the following direct calculation:
    \begin{align*}
        &r\circ(l_M^{12}\ot\id_{A})\\
        &=r\circ(d_M\ot\id_{A})
        \circ(u_M\ot\id_{B_2}\ot\id_{A})
        \circ(\id_{B_1}\ot\beta_{B_2,M}\ot\id_{A})\\
        \intertext{by the definition of $l_M^{12}$,}
        &=d_M\circ(r\ot\id_{B_2})
        \circ(\id_M\ot\beta_{B_2,A})
        \notag\\[-0.4em]
        &\qquad\circ(u_M\ot\id_{B_2}\ot\id_{A})
        \circ(\id_{B_1}\ot\beta_{B_2,M}\ot\id_{A})\\
        \intertext{by (D),}
        &=d_M\circ(r\ot\id_{B_2})
        \circ(u_M\ot\id_{A}\ot\id_{B_2})
        \notag\\[-0.4em]
        &\qquad\circ(\id_{B_1}\ot\id_M\ot\beta_{B_2,A})
        \circ(\id_{B_1}\ot\beta_{B_2,M}\ot\id_{A})\\
        \intertext{by naturality of $\beta_{B_2,A}$,}
        &=d_M\circ(u_M\ot\id_{B_2})
        \circ(\id_{B_1}\ot r\ot\id_{B_2})
        \notag\\[-0.4em]
        &\qquad\circ(\id_{B_1}\ot\id_M\ot\beta_{B_2,A})
        \circ(\id_{B_1}\ot\beta_{B_2,M}\ot\id_{A})\\
        \intertext{by (U),}
        &=d_M\circ(u_M\ot\id_{B_2})
        \circ(\id_{B_1}\ot\beta_{B_2,M})
        \circ(\id_{B_1\ot B_2}\ot r)\\
        \intertext{by naturality of $\beta_{B_2,-}$ and the hexagon axiom,}
        &=l_M^{12}\circ(\id_{B_1\ot B_2}\ot r).
    \end{align*}
    The last equality is the definition of $l_M^{12}$.
    Thus $l_M^{12}$ and $r$ are compatible.
    Thus $(M,l_M^{12},r)$ is an ordinary $B_1\ot B_2$-$A$-bimodule.

    We next check the additional bi-bimodule coherence.  Restricting
    $l_M^{12}$ along the units of $B_2$ and $B_1$, respectively, recovers
    $u_M$ and $d_M$.  On the other hand, the two actions induced through the
    right $A$-action are $u_M$ and $d_M$ by their definitions above.
    Therefore the two pairs of paths in the defining coherence diagram agree,
    so $(M,l_M^{12},r)$ is an object of
    $\hind{}{\vind{B_1}{\CC}{B_2}}{A}$.

    If $f:(M,r_M)\to(N,r_N)$ is a right $A$-module homomorphism, then
    naturality of the braiding and $A$-linearity of $f$ imply
    \begin{align*}
        f\circ u_M=u_N\circ(\id_{B_1}\ot f),
        \qquad
        f\circ d_M=d_N\circ(f\ot\id_{B_2}).
    \end{align*}
    Expanding $l_M^{12}$ then gives
    $f\circ l_M^{12}=l_N^{12}\circ(\id_{B_1\ot B_2}\ot f)$.
    Hence we obtain a functor
    \begin{align*}
        G:\hind{}{\CC}{A}&\to \hind{}{\vind{B_1}{\CC}{B_2}}{A}\\
        (M,r)&\mapsto (M,l^{12}_{M},r),\\
        f&\mapsto f.
    \end{align*}

    Clearly $F\circ G=\Id$.  For the other composite, let $(M,l,r)$ be a
    coherent bi-bimodule.  Let $u_M^l,d_M^l$ be the $B_1$- and
    $B_2$-actions induced from $l$, and let $u_M^r,d_M^r$ be those induced
    from $r$.  The defining coherence condition says
    \begin{align*}
        u_M^l=u_M^r,
        \qquad
        d_M^l=d_M^r.
    \end{align*}
    Moreover, any left $B_1\ot B_2
    $-action is recovered from its two
    restrictions by
    \begin{align*}
        l=d_M^l\circ(u_M^l\ot\id_{B_2})
        \circ(\id_{B_1}\ot\beta_{B_2,M}).
    \end{align*}
    Indeed, this is the associativity axiom for $l$ after writing an element
    of $B_1\ot B_2$ as the product of its $B_1$- and $B_2$-components, together
    with the two unit axioms.  Since $G\circ F$ uses $u_M^r,d_M^r$, the last
    two displayed identities imply that its reconstructed action is the
    original $l$.  Thus $G\circ F=\Id$ on objects and morphisms, and $F$ and
    $G$ are inverse equivalences.
\end{proof}

\begin{rmk}
    Physically (or topologically), $\hind{}{\CC}{A}\simeq \hind{1}{\vind{1}{\CC}{1}}{A}$ is the view point at left side, i.e. in (original phase) $\CC$;
    And $\hind{B_1\ot B_2}{\vind{B_1}{\CC}{B_2}}{A}=:\hind{}{\vind{B_1}{\CC}{B_2}}{A}\simeq \hind{}{\CC}{A}$ is the view point at right side, i.e. in (condensed phase) $\hind{A}{\vind{B_1}{\CC}{B_2}}{A}$.
\end{rmk}

\begin{crl}
    $\hind{}{\vind{B}{\CC}{B}}{B}\simeq \hind{}{\CC}{B}$ as categories.
\end{crl}

\begin{crl}\label{crl:B_1CB_2A_equivalent_to_CA}
    $\hind{}{\hind{B_1}{\CC}{B_2}}{A}\simeq \hind{}{\CC}{A}$ as categories.
\end{crl}
\begin{proof}
    There is 
    \begin{align*}
        \hind{}{\hind{B_1}{\CC}{B_2}}{A}\simeq \hind{B_1\ot B_2}{(\hind{B_1}{\CC}{B_2})}{A}\simeq \hind{\vfus{B_1}{\ot}{ B_2}}{\vind{B_1}{\CC}{B_2}}{A}\simeq \hind{}{\CC}{A}
    \end{align*}    
\end{proof}

We now turn from module categories to morphisms between them. The first step is to fold the two endpoint actions into a single module action.
\begin{lem}\label{lem:bimod_functors_in_C Cat_are_iterations}
    Let $\CA_1$ and $\CA_2\in\Alg_{E_1}(\LMod(\CC))$ and let
    $\CM,\CN\in \hind{\CA_1}{\LMod(\CC)}{\CA_2}$.  We have
    \begin{align*}
        \hind{\CC}{\vind{\CA_1}{\Fun}{\CA_2}}{}(\CM,\CN)
        \simeq
        \Fun_{\CA_1\btc{\CC}\CA_2^{rev}}(\CM,\CN).
    \end{align*}
\end{lem}
\begin{proof}
    Set $\CD:=\CA_1\btc{\CC}\CA_2^{rev}$ and equip $\CM$ and $\CN$
    with the left $\CD$-actions from Lemma
    \ref{lem:folded_relative_tensor_action}.

    Let $((F,\xi^F),\delta_F^L,\delta_F^R):\CM\longrightarrow\CN$ be a right-exact $\CC$-linear $\CA_1$-$\CA_2$-bimodule functor.
    Define, before passing to the relative tensor product,
    \begin{align}
        \chi^F_{a,b,m}:\quad
        F\bigl((a\rhd^{\CM}m)\lhd^{\CM}b\bigr)
        \xrightarrow{\ (\delta_F^R)_{a\rhd^{\CM}m,b}\ }
        F(a\rhd^{\CM}m)\lhd^{\CN}b \notag
        \xrightarrow{\ (\delta_F^L)_{a,m}\lhd^{\CN}\id_b\ }
        (a\rhd^{\CN}F(m))\lhd^{\CN}b.
        \label{eq:folded_bimodule_functor_constraint}
    \end{align}
    Denote the balancing isomorphism in
    \eqref{eq:folded_action_balancing} by
    $\vartheta^{\CM}_{c;a,b,m}$, and similarly for $\CN$.  The required
    balancing of \eqref{eq:folded_bimodule_functor_constraint} is the
    commutative square
    \begin{align}
        \xymatrix@C=5.5em{
            F(((c\odot_{\CA_1}a)\rhd^{\CM}m)\lhd^{\CM}b)
                \ar[r]^{\chi^F_{c\odot a,b,m}}
                \ar[d]_{F(\vartheta^{\CM}_{c;a,b,m})}
            &((c\odot_{\CA_1}a)\rhd^{\CN}F(m))\lhd^{\CN}b
                \ar[d]^{\vartheta^{\CN}_{c;a,b,F(m)}}\\
            F((a\rhd^{\CM}m)\lhd^{\CM}(c\odot_{\CA_2}b))
                \ar[r]_{\chi^F_{a,c\odot b,m}}
            &(a\rhd^{\CN}F(m))\lhd^{\CN}(c\odot_{\CA_2}b).
        }
        \label{diag:folded_functor_balancing}
    \end{align}
    After expanding $\chi^F$, this square is the composite of the
    $\CC$-linearity squares for $\delta_F^L$ and $\delta_F^R$ and their
    mixed compatibility square.  Thus it commutes, and the universal
    property of $q:\CA_1\times\CA_2^{rev}\to\CD$ descends $\chi^F$ to a
    natural isomorphism
    \begin{align*}
        \chi^F_{x,m}:F(x\rhd_{\CD}^{\CM}m)
        \xrightarrow{\ \simeq\ }
        x\rhd_{\CD}^{\CN}F(m),
        \qquad x\in\CD.
    \end{align*}
    After precomposition with $q\times q\times\id_{\CM}$, the
    associativity axiom for $\chi^F$ restricts respectively to the left
    module-functor axiom on
    $(\iota_1(a),\iota_1(a'))$, the right module-functor axiom on
    $(\iota_2(b),\iota_2(b'))$, and the mixed axiom on
    $(\iota_1(a),\iota_2(b))$.  The other mixed ordering is obtained from
    the latter by naturality with respect to the canonical isomorphism
    $\iota_1(a)\otimes_{\CD}\iota_2(b)\simeq
    \iota_2(b)\otimes_{\CD}\iota_1(a)$.  Hence the associativity axiom
    holds by the universal property; the unit axiom follows in the same
    way.  We have therefore constructed a right-exact left $\CD$-module
    functor.

    Conversely, let $F:\CM\to\CN$ be a right-exact left $\CD$-module
    functor with constraint $\chi^F$.  Restriction along the strong
    monoidal functors $\iota_1$ and $\iota_2$ of Lemma
    \ref{lem:folded_relative_tensor_monoidal} defines
    \begin{align*}
        (\delta_F^L)_{a,m}:F(a\rhd^{\CM}m)
        &\xrightarrow{\ \simeq\ }a\rhd^{\CN}F(m),\\
        (\delta_F^R)_{m,b}:F(m\lhd^{\CM}b)
        &\xrightarrow{\ \simeq\ }F(m)\lhd^{\CN}b.
    \end{align*}
    The restrictions of the module-functor axioms give the separate left
    and right module-functor axioms.  There is also a canonical monoidal
    natural isomorphism
    \begin{align}
        \kappa_{a,b}:\iota_1(a)\otimes_{\CD}\iota_2(b)
        \xrightarrow{\ \simeq\ }
        \iota_2(b)\otimes_{\CD}\iota_1(a),
        \label{eq:kappa_folded_factors}
    \end{align}
    since both sides are canonically $a\btd_{\CC}b$.  The associativity
    axiom for $\chi^F$, together with its naturality with respect to
    \eqref{eq:kappa_folded_factors}, is exactly the mixed coherence for
    $(\delta_F^L,\delta_F^R)$.

    Restricting along the strong monoidal functor $j$ in
    \eqref{eq:folded_j_functor} defines
    \begin{align}
        \xi^F_{c,m}:F(c\odot_{\CM}m)
        &\simeq F(j(c)\rhd_{\CD}^{\CM}m)
        \xrightarrow{\ \chi^F_{j(c),m}\ }
        j(c)\rhd_{\CD}^{\CN}F(m)
        \simeq c\odot_{\CN}F(m).
        \label{eq:recover_C_linearity_from_j}
    \end{align}
    Its associativity and unit axioms are restrictions of those for
    $\chi^F$.  The $\CC$-linearity of $\delta_F^L$ follows from the
    module-functor associativity for $j(c),\iota_1(a)$ and naturality with
    respect to the canonical centrality isomorphism between their two
    products.  Explicitly, it is the commutativity of
    \begin{align}
        \xymatrix@C=7em{
            F(c\odot_{\CM}(a\rhd^{\CM}m))
                \ar[r]^{(c\odot\delta_F^L)\circ\xi^F}
                \ar[d]_{F(\simeq)}
            &c\odot_{\CN}(a\rhd^{\CN}F(m))
                \ar[d]^{\simeq}\\
            F(a\rhd^{\CM}(c\odot_{\CM}m))
                \ar[r]_{(a\rhd\xi^F)\circ\delta_F^L}
            &a\rhd^{\CN}(c\odot_{\CN}F(m)).
        }
        \label{diag:delta_L_is_C_linear}
    \end{align}
    The same argument with $j(c)$ and $\iota_2(b)$ proves the
    $\CC$-linearity of $\delta_F^R$.  Thus
    $((F,\xi^F),\delta_F^L,\delta_F^R)$ is a right-exact $\CC$-linear
    $\CA_1$-$\CA_2$-bimodule functor.

    It remains to prove that the two constructions are inverse.  Starting
    with a $\CC$-linear bimodule functor, restriction of
    \eqref{eq:folded_bimodule_functor_constraint} along $\iota_1$ and
    $\iota_2$ recovers $\delta_F^L$ and $\delta_F^R$ by their unit axioms.
    The reconstructed $\CC$-module constraint is the composite
    \begin{align*}
        F(c\odot_{\CM}m)
        &\simeq F((H_1(c)\rhd^{\CM}m)\lhd^{\CM}\mathbf{1}_{\CA_2})\\
        &\xrightarrow{\ (\delta_F^R)_{H_1(c)\rhd m,\mathbf{1}}\ }
        F(H_1(c)\rhd^{\CM}m)\lhd^{\CN}\mathbf{1}_{\CA_2}\\
        &\xrightarrow{\ (\delta_F^L)_{H_1(c),m}\lhd\id\ }
        (H_1(c)\rhd^{\CN}F(m))\lhd^{\CN}\mathbf{1}_{\CA_2}
        \simeq c\odot_{\CN}F(m).
    \end{align*}
    The first arrow involving $\delta_F^R$ disappears by the right-module
    unit axiom.  The remaining composite equals the original
    $\xi^F_{c,m}$ by the $\CC$-linearity square for $\delta_F^L$, evaluated
    at $H_1(c)$ and the monoidal unit.  Hence
    \begin{align}
        \widetilde{\xi}^{F}_{c,m}=\xi^F_{c,m}.
        \label{eq:recover_original_C_constraint}
    \end{align}

    In the other direction, start with a $\CD$-module functor and let
    $\widetilde{\chi}^F$ be reconstructed from its restrictions.  The
    module-functor associativity axiom applied to
    $\iota_1(a),\iota_2(b)$, followed by naturality with respect to $\iota_1(a)\otimes_{\CD}\iota_2(b)
        \simeq q(a,b)$, gives
    \begin{align*}
        \widetilde{\chi}^F_{q(a,b),m}=\chi^F_{q(a,b),m}.
    \end{align*}
    This is not an assertion that every object of $\CD$ is a pure tensor.
    Rather, precomposition with the universal balanced functor $q$ is an
    equivalence on right-exact functor categories and on their natural
    transformations.  The displayed equality therefore implies
    $\widetilde{\chi}^F=\chi^F$ on all of $\CD$.

    Finally, a natural transformation is a $\CC$-linear
    $\CA_1$-$\CA_2$-bimodule natural transformation if and only if its
    pullback along $q$ is compatible with the constraints above.  By the
    same universal property, this is equivalent to being a $\CD$-module
    natural transformation.  The two constructions hence define inverse
    equivalences of the stated functor categories.
\end{proof}

\begin{lem}[\cite{XY25}]
    Let $\CB$ be an $E_2$-monoidal category and let $\CA$ be an $E_1$-monoidal right $\CB$-module category. 
    For any $E_2$-algebra $B\in \CB$ we have 
    \begin{align*}
        \CA\btc{\CB}\hind{}{\CB}{B}\simeq \hind{}{\CA}{B}
    \end{align*}
    is an $E_1$-monoidal equivalence.
\end{lem}

\begin{crl}\label{crl:B1CCCB2_is_E1_equivalent_to_B1CB2}
    Let $\CB$ be an $E_2$-monoidal category, and let $B_1,B_2$ be two $E_2$-algebras in $\CB$.
    There is an $E_1$-monoidal equivalence
    \begin{align*}
        \hind{B_1}{\CB}{}\btc{\CB}\hind{}{\CB}{B_2}\simeq \hind{B_1}{\CB}{B_2}
    \end{align*}
\end{crl}

Combining these identifications with the ordinary \(E_1\) Eilenberg--Watts theorem yields the corresponding statement one Morita level higher.
\begin{thm}[$E_2$ Eilenberg-Watts Theorem]\label{thm:E2_EW}
    \begin{align*}
        \hind{\CC}{\vind{\hind{}{\CC}{B_{1}}}{\Fun}{\hind{}{\CC}{B_{2}}}}{}(\hind{}{\CC}{A_1},\hind{}{\CC}{A_2})\simeq \hind{A_1}{\vind{B_1}{\CC}{B_2}}{A_2}
    \end{align*}
\end{thm}
\begin{proof}
    We have 
    \begin{align*}
        &\hind{\CC}{\vind{\hind{}{\CC}{B_{1}}}{\Fun}{\hind{}{\CC}{B_{2}}}}{}(\hind{}{\CC}{A_1},\hind{}{\CC}{A_2})\\
        &\quad\simeq \Fun_{\hind{B_1}{\CC}{}\btc{\CC}\hind{}{\CC}{B_2}}(\hind{}{\CC}{A_1},\hind{}{\CC}{A_2}) &\text{by Lemma \ref{lem:bimod_functors_in_C Cat_are_iterations}}\\
        &\quad\simeq \Fun_{\hind{B_1}{\CC}{B_2}}(\hind{}{\CC}{A_1},\hind{}{\CC}{A_2}) &\text{by Corollary \ref{crl:B1CCCB2_is_E1_equivalent_to_B1CB2}}\\
        &\quad\simeq \Fun_{\hind{B_1}{\CC}{B_2}}(\hind{}{(\hind{B_1}{\CC}{B_2})}{A_1},\hind{}{(\hind{B_1}{\CC}{B_2})}{A_2}) &\text{by Corollary \ref{crl:B_1CB_2A_equivalent_to_CA} }\\
        &\quad\simeq \hind{A_1}{(\hind{B_1}{\CC}{B_2})}{A_2} &\text{by }E_1\text{ Eilenberg-Watts Theorem}\\
        &\quad\simeq \hind{A_1}{\vind{B_1}{\CC}{B_2}}{A_2} &\text{by Theorem \ref{thm:bibimod_equivalent_to_bimod_in_bimod}}
    \end{align*}
\end{proof}

In particular, the endomorphisms of the regular module realization recover the local \(B\)-modules.
\begin{crl}
    \begin{align*}
        \hind{\CC}{\vind{\hind{}{\CC}{B}}{\Fun}{\hind{}{\CC}{B}}}{}(\hind{}{\CC}{B},\hind{}{\CC}{B})\simeq \hind{B}{\vind{B}{\CC}{B}}{B}\simeq \hind{}{\CC}{B}^{loc}
    \end{align*}
\end{crl}

\begin{thm}
    There is a 3-functor from $\Mrt_{E_2}(\CC)$ to $\Mrt_{E_1}(\LMod(\CC))$
    \begin{align*}
        \mathrm{Mod}_2:\Mrt_{E_2}(\CC)\to \Mrt_{E_1}(\LMod(\CC))
    \end{align*}
\end{thm}
\begin{proof}
    By Lemma \ref{lem:Mod1_induces_functors} and Theorem \ref{thm:E2_EW}
\end{proof}

\begin{dfn}
    Let $F:\CC\to \CD$ be an $n$-functor. We say $F$ is {\bf $k$-fully faithful}, if for every parallel $k-1$-morphisms $f,g$, the induced functor on $n-k$-categories of $k$-morphisms, $k+1$-morphisms..., $n$-morphisms is an equivalence.
\end{dfn}

As a consequence of Theorem \ref{thm:E2_EW}, we have 
\begin{thm}
    The 3-functor
        \begin{align*}
        \mathrm{Mod}_2:\Mrt_{E_2}(\CC)\to \Mrt_{E_1}(\LMod(\CC))
    \end{align*}
    is 2-fully faithful 
\end{thm}

The preceding theorem identifies the higher morphisms between representable module categories. It is therefore natural to isolate the representable part of the target tricategory.
\begin{dfn}\label{dfn:MrtE1LmodRep}
The representable subtricategory
$\Mrt_{E_1}\bigl(\LMod^{\mathrm{rep}}(\CC)\bigr)
\subset \Mrt_{E_1}\bigl(\LMod(\CC)\bigr)$ consists of
\begin{itemize}
    \item objects are left $\CC$-module categories $\hind{}{\CC}{B_i}$ for $B_i\in \Alg_{E_2}(\CC)$;
    \item 1-morphisms from $\hind{}{\CC}{B_1}$ to $\hind{}{\CC}{B_2}$ are $\hind{}{\CC}{B_1}$-$\hind{}{\CC}{B_2}$-bimodule categories $\vind{B_1}{\hind{}{\CC}{A_i}}{B_2}$ in $\LMod(\CC)$ for $A_i\in\hind{B_1}{\Alg_{E_1}(\CC)}{B_2}$;
    \item 2-morphisms from $\vind{B_1}{\hind{}{\CC}{A_1}}{B_2}$ to $\vind{B_1}{\hind{}{\CC}{A_2}}{B_2}$ are $\hind{}{\CC}{B_1}$-$\hind{}{\CC}{B_2}$-bimodule functors $-\otimes_{A_1}M\in\hind{\CC}{\vind{\hind{}{\CC}{B_{1}}}{\Fun}{\hind{}{\CC}{B_{2}}}}{}(\hind{}{\CC}{A_1},\hind{}{\CC}{A_2})$ in $\LMod(\CC)$ for $M\in\hind{A_1}{\vind{B_1}{\CC}{B_2}}{A_2}$;
    \item 3-morphisms are $\hind{}{\CC}{B_1}$$\hind{}{\CC}{B_2}$-bimodule natural transformations in $\LMod(\CC)$.
\end{itemize}

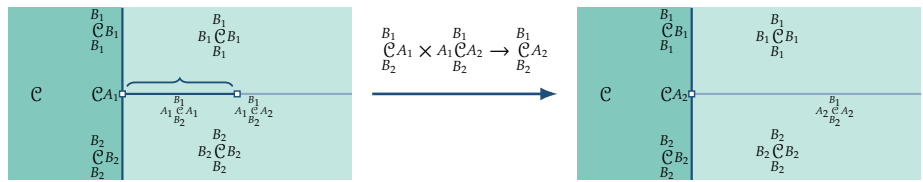
\begin{center}
  \makebox[\textwidth][c]{%
      \resizebox{0.29\textwidth}{!}{%
      \begin{tikzpicture}[baseline=(current bounding box.center)]
        \path[fill=hmDeepCyan] (-2,-1.5) rectangle(0,1.5);
        \path[fill=hmMidCyan] (0,-1.5) rectangle(4,1.5);
        \draw[draw=hmLineGreen,line width=1.1pt] (0,-1.5) -- (0,1.5);
        \draw[draw=hmLineGreen,line width=1.1pt] (0,0) -- (2,0);
        \draw[draw=hmLineGreen,line width=0.9pt,decorate,
          decoration={brace,amplitude=4pt}]
          (0.08,0.13) -- (1.92,0.13);
        \draw[draw=hmLineLightBlue,line width=1.1pt] (2,0) -- (4,0);
        \filldraw[fill=white,draw=hmLineGreen,line width=0.7pt]
          (-0.05,-0.05) rectangle(0.05,0.05);
        \filldraw[fill=white,draw=hmLineGreen,line width=0.7pt]
          (1.95,-0.05) rectangle(2.05,0.05);

        \node at(-1.5,0){$\CC$};
        \node at(-0.3,1.1){$\hind{}{\vind{B_1}{\CC}{B_1}}{B_1}$};
        \node at(-0.3,-1.1){$\hind{}{\vind{B_2}{\CC}{B_2}}{B_2}$};
        \node at(-0.3,0){$\hind{}{\CC}{A_1}$};
        \node at(1.7,1){$\hind{B_1}{\vind{B_1}{\CC}{B_1}}{B_1}$};
        \node at(1.7,-1){$\hind{B_2}{\vind{B_2}{\CC}{B_2}}{B_2}$};
        \node at(1,-0.3){\tiny$\hind{A_1}{\vind{B_1}{\CC}{B_2}}{A_1}$};
        \node at(2.3,-0.3){\tiny$\hind{A_1}{\vind{B_1}{\CC}{B_2}}{A_2}$};
      \end{tikzpicture}%
      }%
      \hspace{0.015\textwidth}%
      \begin{tikzpicture}[baseline=0pt]
        \node[anchor=south] (formula) at (0,0.12cm) {\resizebox{0.14\textwidth}{!}{$
          \hind{}{\vind{B_1}{\CC}{B_2}}{A_1}
          \times
          \vind{B_1}{\hind{A_1}{\CC}{A_2}}{B_2}
          \to
          \hind{}{\vind{B_1}{\CC}{B_2}}{A_2}
        $}};
        \draw[draw=hmLineGreen,-{Latex[length=2.3mm]},line width=1pt]
          ([yshift=-0.12cm]formula.south west) --
          ([yshift=-0.12cm]formula.south east);
      \end{tikzpicture}%
      \hspace{0.015\textwidth}%
      \resizebox{0.29\textwidth}{!}{%
      \begin{tikzpicture}[baseline=(current bounding box.center)]
        \path[fill=hmDeepCyan] (-2,-1.5) rectangle(0,1.5);
        \path[fill=hmMidCyan] (0,-1.5) rectangle(4,1.5);
        \draw[draw=hmLineGreen,line width=1.1pt] (0,-1.5) -- (0,1.5);
        \draw[draw=hmLineLightBlue,line width=1.1pt] (0,0) -- (4,0);
        \filldraw[fill=white,draw=hmLineGreen,line width=0.7pt]
          (-0.05,-0.05) rectangle(0.05,0.05);

        \node at(-1.5,0){$\CC$};
        \node at(-0.3,1.1){$\hind{}{\vind{B_1}{\CC}{B_1}}{B_1}$};
        \node at(-0.3,-1.1){$\hind{}{\vind{B_2}{\CC}{B_2}}{B_2}$};
        \node at(-0.3,0){$\hind{}{\CC}{A_2}$};
        \node at(1.5,1){$\hind{B_1}{\vind{B_1}{\CC}{B_1}}{B_1}$};
        \node at(1.5,-1){$\hind{B_2}{\vind{B_2}{\CC}{B_2}}{B_2}$};
        \node at(2.5,-0.3){\tiny$\hind{A_2}{\vind{B_1}{\CC}{B_2}}{A_2}$};
      \end{tikzpicture}%
      }%
  }
\end{center}
\end{dfn}

Having identified the representable image of \(\mathrm{Mod}_2\), we can now compare \(2\)-Morita equivalence with Morita equivalence after module realization.
\begin{dfn}
    Let $\CC$ be an $E_2$-monoidal $n$-category.
    Two $E_2$-algebras $B_1,B_2$ are {\bf module 2-Morita equivalent} if $\hind{}{\CC}{B_1}$ and $\hind{}{\CC}{B_2}$ are 1-Morita equivalent in $\LMod_{\CC}(n\Cat)$, i.e. if there exists invertible $\hind{}{\CC}{B_1}$-$\hind{}{\CC}{B_2}$-bimodule $\CM$ in $\LMod_{\CC}(n\Cat)$.
\end{dfn}

The module realization therefore provides a second notion of Morita equivalence for \(E_2\)-algebras. The relation between the two notions is one-sided in general, as the following proposition shows.
\begin{prp}
    Let $\CC$ be an $E_2$-monoidal $n$-category.
    If two $E_2$-algebras $B_1,B_2$ are 2-Morita equivalent, then they are module 2-Morita equivalent.
\end{prp}

Under cocompleteness alone, an invertible bimodule category between objects in the image of $\mathrm{Mod}_2$ need not itself be in the essential image of $\mathrm{Mod}_2$.  We work in the cocomplete model $\LMod(\CC)$ and consider the special case of unit endpoints, for which the ambient $\CC$-module and endpoint-bimodule structures can be checked explicitly.
\begin{cexpl}[Nonrepresentability]
    Let $R$ be a connected commutative $G$-ring such that $H^2_{\mathrm{\acute{e}t}}(\operatorname{Spec}R,\mathbb G_m)$ contains a non-torsion class $\alpha$.  
    Set $\CC:=\Mod_R$, the ordinary cocomplete category of $R$-modules, and $B_1=B_2=\bfone_{\CC}=R$.  Then $\hind{}{\CC}{B_1}\simeq \CC\simeq \hind{}{\CC}{B_2}$, and this object is the tensor unit of $\LMod(\CC)$.  
    Let $\mathcal G_\alpha$ be a $\mathbb G_m$-gerbe representing $\alpha$, and let
    \begin{align*}
        \CM_\alpha:=\Mod^{\heartsuit}_{R,\mathcal G_\alpha}
    \end{align*}
    be the cocomplete $R$-linear category of $\mathcal G_\alpha$-twisted $R$-modules.  Its ambient $\CC$-action and its two endpoint actions are
    \begin{align*}
        T\odot X:=T\otimes_R X,
        \qquad
        P\rhd X:=P\otimes_R X,
        \qquad
        X\lhd Q:=Q\otimes_R X.
    \end{align*}
    Associativity and the symmetry of $\otimes_R$ give the strong $\CC$-linearity and mixed-compatibility constraints
    \begin{align*}
        (T\otimes_R P)\rhd X
        \simeq T\odot(P\rhd X),
        \qquad(T\odot X)\lhd Q
        \simeq T\odot(X\lhd Q),
        \qquad(P\rhd X)\lhd Q
        \simeq P\rhd(X\lhd Q).
    \end{align*}
    Hence $\CM_\alpha$ is an internal $\hind{}{\CC}{B_1}$-$\hind{}{\CC}{B_2}$-bimodule object in $\LMod(\CC)$, not merely an ordinary external bimodule category.

    The tensor product of gerbe twists adds their cohomology classes \cite[Theorem~1.0.1]{Ste23}.  Thus $\CM_{-\alpha}$ is an internal inverse and
    \begin{align*}
        &\CM_\alpha\btd_{\hind{}{\CC}{B_2}}\CM_{-\alpha}
        \simeq \CM_\alpha\btd_{\CC}\CM_{-\alpha}
        \simeq \CC, 
        &\CM_{-\alpha}\btd_{\hind{}{\CC}{B_1}}\CM_\alpha
        \simeq \CM_{-\alpha}\btd_{\CC}\CM_\alpha
        \simeq \CC.
    \end{align*}
    These are equivalences of $\CC$-linear bimodule categories.  On the other hand, $\CM_\alpha$ has no nonzero compact-projective object \cite[Example~3.1.25]{Ste23}.  If $\CM_\alpha\simeq\hind{}{\CC}{A}=\Mod_{A}$ for an $E_1$-algebra $A\in\CC$, then the free rank-one module $A$ would be a compact-projective generator of $\CM_\alpha$, a contradiction.  Therefore $\CM_\alpha$ is an invertible target 1-morphism which is not of the form $\hind{}{\CC}{A}$.  In particular, 2-full faithfulness of $\mathrm{Mod}_2$ does not imply essential surjectivity on invertible 1-morphisms. 

    A concrete choice, obtained from the construction in \cite[Section~7.1]{CTS21}, is
    \begin{align*}
        R=\bC[x,y,z]\big/\bigl(z(1+x^3+y^3+z^3)\bigr).
    \end{align*}
    Indeed, $\operatorname{Spec}R$ is the affine chart $X_0\neq 0$ of the union of the smooth cubic surface
    \begin{align*}
        X_0^3+X_1^3+X_2^3+X_3^3=0
    \end{align*}
    and the transverse plane $X_3=0$.  Their intersection is a smooth cubic curve $E$.  Writing $S^\circ$, $H^\circ$, and $E^\circ$ for the intersections with this affine chart, the normalization sequence gives a boundary map
    \begin{align*}
        \operatorname{Pic}(S^\circ)\oplus\operatorname{Pic}(H^\circ)
        \longrightarrow \operatorname{Pic}(E^\circ)
        \longrightarrow H^2_{\mathrm{\acute{e}t}}(\operatorname{Spec}R,\mathbb G_m).
    \end{align*}
    The source of this boundary map is finitely generated.  Moreover, the kernel of $\operatorname{Pic}(E)\to\operatorname{Pic}(E^\circ)$ is generated by the finitely many removed points, whereas $\operatorname{Pic}^0(E)(\bC)$ contains elements of infinite order modulo every finitely generated subgroup.  Hence the cokernel of the first arrow contains a non-torsion class, whose image under the boundary map gives an $\alpha$ as above.
\end{cexpl}

\subsection{Condensation Defects and Their Fusion}\label{subsec:condensation-defects-fusion}

We now specialize the preceding constructions to the finite setting relevant to topological orders.  Let $\CC$ be an indecomposable finite braided multi-tensor category.
For the topological interpretation, we further assume that $\CC$ is a modular fusion category and that the $B_i$ are condensable. In this setting the relative tensor products appearing in the Morita categories exist, and the module realization developed above identifies them with fusion operations of the corresponding topological defects.

A useful feature of the bi-bimodule description is that the two directions of fusion remain visible.  We first record the two elementary fusion laws.  Fix $B_1,B_2$ and let
$A_1,A_2,A_3\in\hind{B_1}{\Alg_{E_1}(\CC)}{B_2}$.

\begin{thm}[Horizontal fusion of bi-bimodules]\label{thm:horizontal_fusion_of_bi-bimod}
There is an equivalence
\begin{align}
    \hind{A_1}{\vind{B_1}{\CC}{B_2}}{A_2}
    \btc{\hind{A_2}{\vind{B_1}{\CC}{B_2}}{A_2}}
    \hind{A_2}{\vind{B_1}{\CC}{B_2}}{A_3}
    \simeq
    \hind{A_1}{\vind{B_1}{\CC}{B_2}}{A_3}.
\end{align}
\end{thm}
\begin{proof}
By Theorem \ref{thm:bibimod_equivalent_to_bimod_in_bimod}, the three bi-bimodule categories are ordinary bimodule categories internal to $\hind{B_1}{\CC}{B_2}$.  Hence
\begin{align*}
    &\hind{A_1}{\vind{B_1}{\CC}{B_2}}{A_2}
    \btc{\hind{A_2}{\vind{B_1}{\CC}{B_2}}{A_2}}
    \hind{A_2}{\vind{B_1}{\CC}{B_2}}{A_3} \\
    &\qquad\simeq
    \hind{A_1}{(\hind{B_1}{\CC}{B_2})}{A_2}
    \btc{\hind{A_2}{(\hind{B_1}{\CC}{B_2})}{A_2}}
    \hind{A_2}{(\hind{B_1}{\CC}{B_2})}{A_3} \\
    &\qquad\simeq
    \hind{A_1}{(\hind{B_1}{\CC}{B_2})}{A_3}
    \simeq
    \hind{A_1}{\vind{B_1}{\CC}{B_2}}{A_3},
\end{align*}
where the middle equivalence is Theorem \ref{thm:fusion_of_1d_phase}.
\end{proof}

There is also a transverse fusion in which the intermediate $E_2$-phase is removed.  Let
\begin{align*}
    A_1,A_2&\in\hind{B_1}{\Alg_{E_1}(\CC)}{B_2},
    &
    A_3,A_4&\in\hind{B_2}{\Alg_{E_1}(\CC)}{B_3}.
\end{align*}

\begin{thm}[\cite{KYZZ26}]\label{thm:functoriality_of _central_functor}
    \begin{align*}
        \hind{}{\vfus{\hind{\CC}{\vind{\CA_1}{\Fun}{\CA_2}}{}(\CM_1,\CM_2)}{\boxtimes}{\hind{\CC}{\vind{\CA_2}{\Fun}{\CA_2}}{}(\CN_1,\CN_2)}}{\hind{\CC}{\vind{\CA_2}{\Fun}{\CA_3}}{}(\CA_2,\CA_2)}
        \simeq \hind{\CC}{\vind{\CA_1}{\Fun}{\CA_3}}{}\left(\vfus{\CM_1}{\btc{\CA_2}}{\CN_1},\vfus{\CM_2}{\btc{\CA_2}}{\CN_2}\right) 
    \end{align*}
\end{thm}
\begin{proof}
    This is proved in the cited work.
\end{proof}

\begin{prp}[\cite{Lur17,Hau17,DSPS19}]\label{prp:module_construction_preserves_relative_tensor_products}
    Let $\CC$ be an indecomposable finite braided multi-tensor category.  Let
    $A_1\in \hind{B_1}{\Alg_{E_1}(\CC)}{B_2}$ and
    $A_2\in \hind{B_2}{\Alg_{E_1}(\CC)}{B_3}$.  Then
    \begin{align*}
        \vfus{\hind{}{\CC}{A_1}}{\btc{\hind{}{\CC}{B_2}}}{\hind{}{\CC}{A_2}}
        \simeq \hind{}{\CC}{\vfus{A_1}{\otc{B_2}}{A_2}}
    \end{align*}
    as $\hind{}{\CC}{B_1}$-$\hind{}{\CC}{B_3}$-bimodule categories in
    $\hind{\CC}{\Cat}{}$.
\end{prp}
\begin{proof}
    Write
    \begin{align*}
        \underline{\mathrm{Mod}}_1:\Alg_{E_1}(\CC)&\longrightarrow \hind{\CC}{\Cat}{},
        &A&\longmapsto \hind{}{\CC}{A}.
    \end{align*}
    By \cite[Theorem~4.8.5.16 and Remark~4.8.5.17]{Lur17}, the module
    construction is strong monoidal and is compatible with the geometric
    realizations defining relative tensor products; see also
    \cite[Remark~5.32]{Hau17} for the induced functor on higher Morita
    categories.  In the finite setting this follows by applying the
    presentable result to $\Ind(\CC)$ and then restricting to compact objects;
    the finite balanced tensor products are also realized explicitly in
    \cite{DSPS19}.

    The relative tensor product of $A_1$ and $A_2$ is the realization of the
    two-sided bar construction
    \begin{align*}
        \vfus{A_1}{\otc{B_2}}{A_2}
        \simeq
        \left|\operatorname{Bar}_{\bullet}(A_1,B_2,A_2)\right|,
        \qquad
        \operatorname{Bar}_{n}(A_1,B_2,A_2)
        =A_1\ot B_2^{\ot n}\ot A_2.
    \end{align*}
    Strong monoidality identifies the image of its degree-$n$ term with
    \begin{align*}
        \underline{\mathrm{Mod}}_1
        \left(A_1\ot B_2^{\ot n}\ot A_2\right)
        \simeq
        \hind{}{\CC}{A_1}
        \btd\limits_{\CC}
        \underbrace{
        \hind{}{\CC}{B_2}\btd\limits_{\CC}\cdots
        \btd\limits_{\CC}\hind{}{\CC}{B_2}
        }_{n\text{ factors}}
        \btd\limits_{\CC}
        \hind{}{\CC}{A_2}.
    \end{align*}
    For $n=0$, the middle string of $\hind{}{\CC}{B_2}$ factors is omitted.
    These equivalences intertwine the face maps induced by the two
    $B_2$-actions and by the multiplication of $B_2$, as well as the degeneracy
    maps induced by its unit.  Hence they identify the two simplicial bar
    constructions.  Taking realizations gives
    \begin{align*}
        \vfus{\hind{}{\CC}{A_1}}{\btc{\hind{}{\CC}{B_2}}}{\hind{}{\CC}{A_2}}
        &\simeq
        \left|
        \operatorname{Bar}_{\bullet}
        \left(\hind{}{\CC}{A_1},\hind{}{\CC}{B_2},\hind{}{\CC}{A_2}\right)
        \right|\\
        &\simeq
        \underline{\mathrm{Mod}}_1
        \left(
        \left|\operatorname{Bar}_{\bullet}(A_1,B_2,A_2)\right|
        \right)\\
        &\simeq
        \hind{}{\CC}{\vfus{A_1}{\otc{B_2}}{A_2}}.
    \end{align*}

    Concretely, the comparison is induced by the
    $\hind{}{\CC}{B_2}$-balanced functor
    \begin{align*}
        \hind{}{\CC}{A_1}\times\hind{}{\CC}{A_2}
        &\longrightarrow
        \hind{}{\CC}{\vfus{A_1}{\otc{B_2}}{A_2}},
        &(M,N)&\longmapsto \vfus{M}{\otc{B_2}}{N}.
    \end{align*}
    Its right $\vfus{A_1}{\otc{B_2}}{A_2}$-action is the composite
    \begin{align*}
        \left(\vfus{M}{\otc{B_2}}{N}\right)
        \ot
        \left(\vfus{A_1}{\otc{B_2}}{A_2}\right)
        &\xrightarrow{\ \chi\ }
        \vfus{M\ot A_1}{\otc{B_2}}{N\ot A_2}
        \xrightarrow{\ r_M\otc{B_2}r_N\ }
        \vfus{M}{\otc{B_2}}{N},
    \end{align*}
    where $\chi$ is the canonical braided interchange map.  The endpoint
    actions make this an equivalence of
    $\hind{}{\CC}{B_1}$-$\hind{}{\CC}{B_3}$-bimodule categories.
\end{proof}

\begin{thm}[Vertical fusion of bi-bimodules]
\label{thm:vertical_fusion_of_bi-bimod}
Let $\CC$ be an indecomposable finite braided multi-tensor category.
Then there is an equivalence
\begin{align*}
    \vfus{
      \hind{A_1}{\vind{B_1}{\CC}{B_2}}{A_2}
    }{
      \btc{\hind{B_2}{\vind{B_2}{\CC}{B_2}}{B_2}}
    }{
      \hind{A_3}{\vind{B_2}{\CC}{B_3}}{A_4}
    }
    \simeq
    \hind{
      \vfus{A_1}{\otc{B_2}}{A_3}
    }{
      \vind{B_1}{\CC}{B_3}
    }{
      \vfus{A_2}{\otc{B_2}}{A_4}
    }.
\end{align*}
\end{thm}

\begin{proof}
By the $E_2$ Eilenberg--Watts equivalence, the two bi-bimodule
categories on the left are identified with the corresponding categories
of bimodule functors in the module realization.  Thus
    \begin{align*}
        \vfus{\hind{A_1}{\vind{B_1}{\CC}{B_2}}{A_2}}{\btc{\hind{B_2}{\vind{B_2}{\CC}{B_2}}{B_2}}}{\hind{A_3}{\vind{B_2}{\CC}{B_3}}{A_4}}
        &\quad  \simeq \hind{}{\vfus{\hind{\CC}{\vind{\hind{}{\CC}{B_{1}}}{\Fun}{\hind{}{\CC}{B_{2}}}}{}(\hind{}{\CC}{A_1},\hind{}{\CC}{A_2})}{\boxtimes}{\hind{\CC}{\vind{\hind{}{\CC}{B_{2}}}{\Fun}{\hind{}{\CC}{B_{3}}}}{}(\hind{}{\CC}{A_3},\hind{}{\CC}{A_4})}}{\hind{\CC}{\vind{\hind{}{\CC}{B_{2}}}{\Fun}{\hind{}{\CC}{B_{2}}}}{}(\hind{}{\CC}{B_2},\hind{}{\CC}{B_2})} &\text{by E2 Eilenberg-Watts theorem}\\
        &\quad \simeq \hind{\CC}{\vind{\hind{}{\CC}{B_{1}}}{\Fun}{\hind{}{\CC}{B_{3}}}}{}\left(\vfus{\hind{}{\CC}{A_1}}{\btc{\hind{}{\CC}{B_2}}}{\hind{}{\CC}{A_3}},\vfus{\hind{}{\CC}{A_2}}{\btc{\hind{}{\CC}{B_2}}}{\hind{}{\CC}{A_4}}\right) &\text{by Theorem \ref{thm:functoriality_of _central_functor}}\\
        &\quad \simeq \hind{\CC}{\vind{\hind{}{\CC}{B_{1}}}{\Fun}{\hind{}{\CC}{B_{3}}}}{}\left(\hind{}{\CC}{\vfus{A_1}{\otc{B_2}}{A_3}},\hind{}{\CC}{\vfus{A_2}{\otc{B_2}}{A_4}}\right) &\text{by Proposition \ref{prp:module_construction_preserves_relative_tensor_products}}\\
        &\quad \simeq \hind{\vfus{A_1}{\otc{B_2}}{A_3}}{\vind{B_1}{\CC}{B_3}}{\vfus{A_2}{\otc{B_2}}{A_4}} &\text{by E2 Eilenberg-Watts theorem}
    \end{align*}
The middle equivalence is precisely the vertical composition of the
corresponding bimodule functors.  Proposition
\ref{prp:module_construction_preserves_relative_tensor_products}
identifies the source and target of this composite with the module
realizations of
$\vfus{A_1}{\otc{B_2}}{A_3}$ and
$\vfus{A_2}{\otc{B_2}}{A_4}$, respectively.
Applying the $E_2$ Eilenberg--Watts equivalence once more gives the
claimed bi-bimodule category.
\end{proof}

\begin{figure}[H]
    \centering
    \begin{minipage}[c]{0.4\textwidth}
        \centering
    \resizebox{0.8\textwidth}{!}{%
    \begin{tikzpicture}[baseline=(current bounding box.center)]
        \path[fill=hmDeepCyan] (-2,-1.5) rectangle(0,1.5);
        \path[fill=hmMidCyan] (0,-1.5) rectangle(4,1.5);
        \draw[draw=hmLineGreen,line width=1.1pt] (0,-1.5) -- (0,1.5);
        \draw[draw=hmLineGreen,line width=1.1pt] (0,0.5) -- (2,0.5);
        \draw[draw=hmLineLightBlue,line width=1.1pt] (2,0.5) -- (4,0.5);
        \draw[draw=hmLineGreen,line width=1.1pt] (0,-0.5) -- (2,-0.5);
        \draw[draw=hmLineLightBlue,line width=1.1pt] (2,-0.5) -- (4,-0.5);
        \filldraw[fill=white,draw=hmLineGreen,line width=0.7pt]
          (-0.05,0.45) rectangle(0.05,0.55);
        \filldraw[fill=white,draw=hmLineGreen,line width=0.7pt]
          (-0.05,-0.45) rectangle(0.05,-0.55);
        \filldraw[fill=white,draw=hmLineGreen,line width=0.7pt]
          (1.95,0.45) rectangle(2.05,0.55);
        \filldraw[fill=white,draw=hmLineGreen,line width=0.7pt]
          (1.95,-0.45) rectangle(2.05,-0.55);

        \node at(-1.5,0){$\CC$};
        \node at(-0.3,1.1){$\hind{}{\CC}{B_1}$};
        \node at(-0.3,0){$\hind{}{\CC}{B_2}$};
        \node at(-0.3,-1.1){$\hind{}{\CC}{B_3}$};
        \node[anchor=east] at(-0.08,0.5){$\hind{}{\CC}{A_1}$};
        \node[anchor=east] at(-0.08,-0.5){$\hind{}{\CC}{A_2}$};
        \node at(3,0){$\hind{B_2}{\vind{B_2}{\CC}{B_2}}{B_2}$};
        \node[above] at(1,0.5){\tiny$\hind{A_1}{\vind{B_1}{\CC}{B_2}}{A_1}$};
        \node[above] at(2.2,0.5){\tiny$\hind{A_1}{\vind{B_1}{\CC}{B_2}}{A_2}$};
        \node[below] at(1,-0.5){\tiny$\hind{A_3}{\vind{B_2}{\CC}{B_3}}{A_3}$};
        \node[below] at(2.2,-0.5){\tiny$\hind{A_3}{\vind{B_2}{\CC}{B_3}}{A_4}$};
    \end{tikzpicture}}%
\end{minipage}
    \begin{minipage}[c]{0.4\textwidth}
        \centering
    \resizebox{0.8\textwidth}{!}{%
    \begin{tikzpicture}[baseline=(current bounding box.center)]
        \path[fill=hmDeepCyan] (-2,-1.5) rectangle(0,1.5);
        \path[fill=hmMidCyan] (0,-1.5) rectangle(4,1.5);
        \draw[draw=hmLineGreen,line width=1.1pt] (0,-1.5) -- (0,1.5);
        \draw[draw=hmLineGreen,line width=1.1pt] (0,0) -- (2,0);
        \draw[draw=hmLineLightBlue,line width=1.1pt] (2,0) -- (4,0);
        \filldraw[fill=white,draw=hmLineGreen,line width=0.7pt]
          (-0.05,-0.05) rectangle(0.05,0.05);
        \filldraw[fill=white,draw=hmLineGreen,line width=0.7pt]
          (1.95,-0.05) rectangle(2.05,0.05);

        \node at(-1.5,0){$\CC$};
        \node at(-0.3,1.1){$\hind{}{\CC}{B_1}$};
        \node at(-0.3,-1.1){$\hind{}{\CC}{B_3}$};
        \node at(-0.3,0){$\hind{}{\CC}{\vfus{A_1}{\otc{B_2}}{A_3}}$};
        \node at(2.8,1){$\hind{B_1}{\vind{B_1}{\CC}{B_1}}{B_1}$};
        \node at(2.8,-1){$\hind{B_3}{\vind{B_3}{\CC}{B_3}}{B_3}$};
        \node[below] at(1,0){\tiny$\hind{\vfus{A_1}{\otc{B_2}}{A_3}}{\vind{B_1}{\CC}{B_3}}{\vfus{A_1}{\otc{B_2}}{A_3}}$};
        \node[above] at(2,0){\tiny$\hind{\vfus{A_1}{\otc{B_2}}{A_3}}{\vind{B_1}{\CC}{B_3}}{\vfus{A_2}{\otc{B_2}}{A_4}}$};
    \end{tikzpicture}}%
\end{minipage}
    
\caption{Vertical fusion of two bi-bimodule defects across the intermediate $B_2$-phase.  The two boundary labels on each side fuse by the corresponding relative tensor products over $B_2$.}
\label{fig:vertical-bibimodule-fusion}
\end{figure}

\subsubsection*{Opening the finite module realization}

The preceding fusion laws take place inside the usual finite Morita target
\begin{align*}
    \Mrt_{E_1}\bigl(\LMod^{\mathrm{fin}}(\CC)\bigr),
    \qquad
    \LMod^{\mathrm{fin}}(\CC):=\LMod_{\CC}(\Cat^{\mathrm{fin}}_{\bk}).
\end{align*}
There is a useful way to read this target geometrically.  For any monoidal category $\CD$, the tensor unit determines the unit corner of its Morita category, and the underlying $1$-morphisms in
$\Hom_{\Mrt_{E_1}(\CD)}(\mathbf 1_{\CD},\mathbf 1_{\CD})$ are precisely the objects of $\CD$ with their canonical unit actions.  Taking $\CD=\LMod^{\mathrm{fin}}(\CC)$, the tensor unit is the regular $\CC$-module category $\CC$.  Thus a representable module category $\hind{}{\CC}{A}$ may first be viewed in this unit corner.  Choosing compatible $B_1$- and $B_2$-actions then refines the same object to the representable $1$-morphism
\begin{align*}
    \vind{B_1}{\hind{}{\CC}{A}}{B_2}
    :\hind{}{\CC}{B_1}\longrightarrow\hind{}{\CC}{B_2}.
\end{align*}
Equivalently, the unit inclusions $1\to B_i$ allow us to forget the transverse $B_i$-actions and return to the ordinary bimodule picture.  The important point is that this ``opening'' does not change the underlying object; it makes additional module structures visible.

\begin{figure}[H]
\centering
\begingroup
\footnotesize
  \noindent
    \begin{minipage}[c]{0.33\textwidth}
      \centering
      \resizebox{0.9\linewidth}{!}{%
\begin{tikzpicture}
        \path[fill=hmDeepCyan] (-2,-1.5) rectangle(4,1.5);
        \draw[draw=hmLineGreen, line width=1.1pt] (0,0) -- (2,0);
         \draw[draw= hmLineLightBlue,line width=1.1pt] (2,0) -- (4,0);
        \draw[draw= hmLineLightBlue, dashed, line width=1.1pt] (-2,0) -- (0,0);
        \filldraw[fill=white,draw=hmLineGreen,line width=0.7pt] (-0.05,-0.05) rectangle(0.05,0.05);
        \filldraw[fill=white,draw=hmLineGreen,line width=0.7pt] (1.95,-0.05) rectangle(2.05,0.05);

        \node at(-1,0){$\CC$};
        \node at(-0.3,0){$\hind{}{\CC}{A_1}$};
        \node at(2,1){$\CC$};
        \node at(2,-1){$\CC$};
        \node at(1,-0.3){\tiny$\hind{A_1}{\CC}{A_1}$};
        \node at(2,0.35){\tiny$\hind{A_1}{\CC}{A_2}$};
        \node at(3,-0.3){\tiny$\hind{A_2}{\CC}{A_2}$};
      \end{tikzpicture}}
    \end{minipage}\hfill%
    \begin{minipage}[c]{0.33\textwidth}
      \centering
      \resizebox{0.9\linewidth}{!}{%
            \begin{tikzpicture}
        \path[fill=hmDeepCyan] (-2,-1.5) rectangle(0,1.5);
        \path[fill=hmMidCyan] (0,-0.8) rectangle(4,0.8);
        \path[fill=hmDeepCyan] (0,-0.8) rectangle(4,-1.5);
        \path[fill=hmDeepCyan] (0,0.8) rectangle(4,1.5);
        \draw[draw=hmLineGreen,line width=1.1pt] (0,-0.8) -- (0,0.8);
        \draw[draw=hmLineGreen, line width=1.1pt] (0,0) -- (2,0);
        \draw[draw=hmLineGreen,line width=1.1pt] (0,0.8) -- (4,0.8);
        \draw[draw=hmLineGreen, line width=1.1pt] (0,-0.8) -- (4,-0.8);
         \filldraw[draw= hmLineLightBlue,line width=1.1pt] (2,0) -- (4,0);
        \filldraw[fill=white,draw=hmLineGreen,line width=0.7pt] (-0.05,-0.05) rectangle(0.05,0.05);
        \filldraw[fill=white,draw=hmLineGreen,line width=0.7pt] (-0.05,-0.85) rectangle(0.05,-0.75);
         \filldraw[fill=white,draw=hmLineGreen,line width=0.7pt] (-0.05,0.75) rectangle(0.05,0.85);
        \filldraw[fill=white,draw=hmLineGreen,line width=0.7pt] (1.95,-0.05) rectangle(2.05,0.05);

        \node at(-1,0){$\CC$};
        \node at(-0.3,0){$\hind{}{\CC}{A_1}$};
        \node at(2,1.15){$\hind{B_1}{\vind{}{\CC}{B_1}}{B_1}$};
        \node at(-0.3,1.15){$\hind{}{\vind{}{\CC}{B_1}}{B_1}$};
        \node at(2,-1.15){$\hind{B_2}{\vind{B_2}{\CC}{}}{B_2}$};
        \node at(-0.3,-1.15){$\hind{}{\vind{B_2}{\CC}{}}{B_2}$};
        \node at(1,-0.3){\tiny$\hind{A_1}{\vind{B_1}{\CC}{B_2}}{A_1}$};
        \node at(2,0.35){\tiny$\hind{A_1}{\vind{B_1}{\CC}{B_2}}{A_2}$};
        \node at(3,-0.3){\tiny$\hind{A_2}{\vind{B_1}{\CC}{B_2}}{A_2}$};
      \end{tikzpicture}}
    \end{minipage}\hfill%
    \begin{minipage}[c]{0.33\textwidth}
      \centering
      \resizebox{0.9\linewidth}{!}{%
                  \begin{tikzpicture}
        \path[fill=hmDeepCyan] (-2,-1.5) rectangle(0,1.5);
        \path[fill=hmMidCyan] (0,-1.5) rectangle(4,1.5);
        \draw[draw=hmLineGreen,line width=1.1pt] (0,-1.5) -- (0,1.5);
        \draw[draw=hmLineGreen, line width=1.1pt] (0,0) -- (2,0);
         \filldraw[draw= hmLineLightBlue,line width=1.1pt] (2,0) -- (4,0);
        \filldraw[fill=white,draw=hmLineGreen,line width=0.7pt] (-0.05,-0.05) rectangle(0.05,0.05);
        \filldraw[fill=white,draw=hmLineGreen,line width=0.7pt] (1.95,-0.05) rectangle(2.05,0.05);

        \node at(-1.5,0){$\CC$};
        \node at(-0.3,1.1){$\hind{}{\vind{B_1}{\CC}{B_1}}{B_1}$};
        \node at(-0.3,-1.1){$\hind{}{\vind{B_2}{\CC}{B_2}}{B_2}$};
        \node at(-0.3,0){$\hind{}{\CC}{A_1}$};
        \node at(2,1){$\hind{B_1}{\vind{B_1}{\CC}{B_1}}{B_1}$};
        \node at(2,-1){$\hind{B_2}{\vind{B_2}{\CC}{B_2}}{B_2}$};
        \node at(1,-0.3){\tiny$\hind{A_1}{\vind{B_1}{\CC}{B_2}}{A_1}$};
        \node at(2,0.35){\tiny$\hind{A_1}{\vind{B_1}{\CC}{B_2}}{A_2}$};
        \node at(3,-0.3){\tiny$\hind{A_2}{\vind{B_1}{\CC}{B_2}}{A_2}$};
      \end{tikzpicture}}
      
    \end{minipage}\par
\endgroup
\caption{Passing from ordinary bimodule data in $\CC$ to bi-bimodule data after making the two adjacent $E_2$-actions explicit.  The left panel lies in the unit corner of the finite Morita target, while the middle and right panels display the same representable cell after choosing the $B_1$- and $B_2$-actions.  Thus the passage should be read as a refinement of module structure rather than as a change of the underlying object.}
\label{fig:finite-opening-bibimodule}
\end{figure}

This is the first sense in which the Morita realization provides a two-step-condensation setup: ordinary $A_1$--$A_2$ bimodule data in $\CC$ can be resolved into bi-bimodule data once the adjacent $E_2$-actions have been chosen.  The next step is to allow those outer $E_2$-actions themselves to be refined.

Suppose that
\begin{align*}
    f_1:B_1\hookrightarrow B_3,
    \qquad
    f_2:B_2\hookrightarrow B_4
\end{align*}
are morphisms of $E_2$-algebras.  Restriction of scalars gives the fully faithful functors discussed in \cref{subsec:two-morita-iterated-bimodules}, in particular
\begin{align*}
    (f_1,f_2)^*:
    \hind{A_1}{\vind{B_3}{\CC}{B_4}}{A_2}
    \longrightarrow
    \hind{A_1^*}{\vind{B_1}{\CC}{B_2}}{A_2^*}.
\end{align*}
Accordingly, if the $B_1$--$B_2$ actions on an algebra $A_i$ extend to $B_3$--$B_4$ actions, we use the same symbol $A_i$ for the chosen extension and its restriction when no confusion can arise.  In this situation
\begin{align*}
    A_i&\in\hind{B_3}{\Alg_{E_1}(\CC)}{B_4}
    \quad\Longrightarrow\quad
    (f_1,f_2)^*A_i\in\hind{B_1}{\Alg_{E_1}(\CC)}{B_2},
\end{align*}
and similarly a $2$-morphism may be viewed at the two resolutions:
\begin{align}\label{eq:restriction-second-stage-2morphism}
    M\in\hind{A_1}{\vind{B_3}{\CC}{B_4}}{A_2}
    \subset
    \hind{A_1}{\vind{B_1}{\CC}{B_2}}{A_2}
    \simeq
    \Alg_{E_0}\!\left(
      \hind{A_1}{(\hind{B_1}{\CC}{B_2})}{A_2}
    \right),
\end{align}
where the inclusion means the fully faithful restriction functor after identifying the restricted $A_i$ with the displayed $A_i$.  This is the precise sense in which the same $A_i$ and $M$ can carry both the $B_1$--$B_2$ and the refined $B_3$--$B_4$ data.

\begin{figure}[ht]
\centering
\begin{subfigure}[c]{0.45\textwidth}
    \centering
    \begin{tikzpicture}[line cap=round,line join=round]
\path[use as bounding box] (-2,-2) rectangle (2,2);

        \filldraw[fill=gray!40,draw=white]
            (-2,1.33) rectangle (2,-1.33);

        \draw[thick] (-2,0) -- (2,0);
        \fill[black] (0,0) circle (0.07);

        \node at (0,0.8) {$B_1$};
        \node at (0,-0.8) {$B_2$};

        \node at (-1,0.25) {$A_1$};
        \node at (1,0.25) {$A_2$};

        \node at (0,-0.2) {$M$};
    \end{tikzpicture}
\end{subfigure}
\begin{subfigure}[c]{0.45\textwidth}
    \centering
    \begin{tikzpicture}[line cap=round,line join=round]
        \path[use as bounding box] (-2,-2) rectangle (2,2);

        \filldraw[fill=gray!40,draw=white]
            (-2,1.33) rectangle (2,-1.33);

        \draw[thick] (-2,0) -- (2,0);
        \draw[thick]
            (-1.55,0) .. controls (-1.12,0.72) and (-0.48,0.82) .. (0,0.82)
            .. controls (0.48,0.82) and (1.12,0.72) .. (1.55,0);
        \draw[thick]
            (-1.55,0) .. controls (-1.12,-0.72) and (-0.48,-0.82) .. (0,-0.82)
            .. controls (0.48,-0.82) and (1.12,-0.72) .. (1.55,0);
        \fill[black] (0,0) circle (0.07);

        \node at (0,1.08) {$B_1$};
        \node at (0,-1.08) {$B_2$};
        \node at (0.25,0.47) {$B_3$};
        \node at (0.35,-0.52) {$B_4$};

        \node at (-1.02,0.27) {$A_1$};
        \node at (1.02,0.27) {$A_2$};
        \node at (0,-0.22) {$M$};
    \end{tikzpicture}
\end{subfigure}
\caption{Refining the outer $E_2$-actions.  Suppose the $B_1$--$B_2$ actions on $A_1,A_2$ and $M$ extend along $B_1\hookrightarrow B_3$ and $B_2\hookrightarrow B_4$.  The right panel displays the same underlying $A_1,A_2,M$ with the larger $B_3$--$B_4$ actions; restriction of scalars recovers the left panel.  This is a refinement of the same cells in $\Mrt_{E_2}(\CC)$, and provides the algebraic input for the second condensation step after applying $\mathrm{Mod}_2$.}
\label{fig:restriction-outer-actions}
\end{figure}
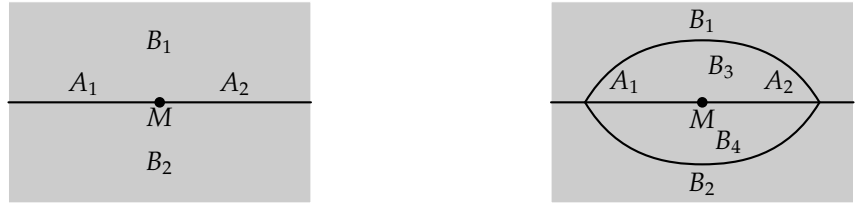

The figure should therefore not be read as saying that an arbitrary $B_1$--$B_2$ algebra automatically acquires $B_3$--$B_4$ actions.  Rather, it records a chosen extension of the outer actions.  This is analogous to the elementary fact that every $E_1$-algebra is canonically a $1$--$1$ bimodule algebra, while a non-trivial choice of $B_1$--$B_2$ actions is extra structure.

\subsubsection*{A topological enlargement for the second condensation step}

Definition \ref{dfn:MrtE1LmodRep} is globular: a $2$-morphism compares representable $1$-morphisms with the same two endpoint objects.  To describe a second condensation step, the pictures naturally contain non-trivial side walls
$B_1\to B_3$ and $B_2\to B_4$.  We therefore use the following enlargement of the displayed finite Morita cells.

\begin{dfn}[Topological enlargement of the finite Morita target]\label{dfn:MrtE1LmodTop}
We write
\begin{align*}
    \Mrt_{E_1}^{\mathrm{top}}\bigl(\LMod^{\mathrm{fin}}(\CC)\bigr)
\end{align*}
for the following topological cell structure.
\begin{itemize}
    \item 0-cells are the representable monoidal left $\CC$-module categories
    $\hind{}{\CC}{B}$, with $B\in\Alg_{E_2}(\CC)$.

    \item Horizontal $1$-cells from $\hind{}{\CC}{B_1}$ to
    $\hind{}{\CC}{B_2}$ are the representable bimodule categories
    \begin{align*}
        \vind{B_1}{\hind{}{\CC}{A}}{B_2},
        \qquad
        A\in\hind{B_1}{\Alg_{E_1}(\CC)}{B_2}.
    \end{align*}

    \item A vertical boundary from $\hind{}{\CC}{B}$ to
    $\hind{}{\CC}{B'}$ is specified by an $E_2$-algebra map
    $f:B\to B'$.  Under module realization it is represented by the corresponding extension-of-scalars condensation morphism; in the pictures below we keep the algebra map itself as the boundary label.

    \item Consider horizontal $1$-cells
    \begin{align*}
        \vind{B_1}{\hind{}{\CC}{A_1}}{B_2},
        \qquad
        \vind{B_3}{\hind{}{\CC}{A_2}}{B_4},
    \end{align*}
    together with vertical boundaries $f_1:B_1\to B_3$ and
    $f_2:B_2\to B_4$.  Assume that the $B_1$--$B_2$ structure on $A_1$ has been extended along $(f_1,f_2)$ to a $B_3$--$B_4$ structure; we suppress the extension from the notation.  A topological $2$-cell filling this square is represented by
    \begin{align*}
        M\in\hind{A_1}{\vind{B_3}{\CC}{B_4}}{A_2}.
    \end{align*}

\end{itemize}
\end{dfn}

\begin{rmk}
When the two vertical boundaries are identities, this reduces to the usual globular cells of
$\Mrt_{E_1}(\LMod^{\mathrm{rep}}(\CC))$, and in the finite setting to the corresponding representable part of
$\Mrt_{E_1}(\LMod^{\mathrm{fin}}(\CC))$.  Thus there is a canonical inclusion of the ordinary finite displayed cells into the topological enlargement by taking identity side walls.

For non-identity $f_1,f_2$, a $2$-cell has different endpoint objects on its upper and lower horizontal boundaries.  In this subsection we only use its displayed cells and their relative tensor product fusion, so no additional global double-categorical coherence is required.
\end{rmk}

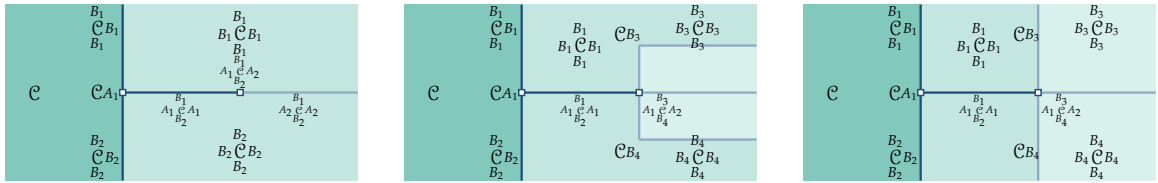
\begin{figure}[H]
\centering
    \noindent
    \begin{minipage}[c]{0.33\textwidth}
      \centering
      \resizebox{0.9\linewidth}{!}{%
            \begin{tikzpicture}
        \path[fill=hmDeepCyan] (-2,-1.5) rectangle(0,1.5);
        \path[fill=hmMidCyan] (0,-1.5) rectangle(4,1.5);
        \draw[draw=hmLineGreen,line width=1.1pt] (0,-1.5) -- (0,1.5);
        \draw[draw=hmLineGreen, line width=1.1pt] (0,0) -- (2,0);
         \filldraw[draw= hmLineLightBlue,line width=1.1pt] (2,0) -- (4,0);
        \filldraw[fill=white,draw=hmLineGreen,line width=0.7pt] (-0.05,-0.05) rectangle(0.05,0.05);
        \filldraw[fill=white,draw=hmLineGreen,line width=0.7pt] (1.95,-0.05) rectangle(2.05,0.05);

        \node at(-1.5,0){$\CC$};
        \node at(-0.3,1.1){$\hind{}{\vind{B_1}{\CC}{B_1}}{B_1}$};
        \node at(-0.3,-1.1){$\hind{}{\vind{B_2}{\CC}{B_2}}{B_2}$};
        \node at(-0.3,0){$\hind{}{\CC}{A_1}$};
        \node at(2,1){$\hind{B_1}{\vind{B_1}{\CC}{B_1}}{B_1}$};
        \node at(2,-1){$\hind{B_2}{\vind{B_2}{\CC}{B_2}}{B_2}$};
        \node at(1,-0.3){\tiny$\hind{A_1}{\vind{B_1}{\CC}{B_2}}{A_1}$};
        \node at(2,0.35){\tiny$\hind{A_1}{\vind{B_1}{\CC}{B_2}}{A_2}$};
        \node at(3,-0.3){\tiny$\hind{A_2}{\vind{B_1}{\CC}{B_2}}{A_2}$};
      \end{tikzpicture}}
    \end{minipage}\hfill%
    \begin{minipage}[c]{0.33\textwidth}
      \centering
      \resizebox{0.9\linewidth}{!}{%
 \begin{tikzpicture}
        \path[fill=hmDeepCyan] (-2,-1.5) rectangle(0,1.5);
        \path[fill=hmMidCyan] (0,-1.5) rectangle(2,1.5);
        \path[fill=hmMidCyan] (2,0.8) rectangle(4,1.5);
        \path[fill=hmMidCyan] (2,-0.8) rectangle(4,-1.5);
         \path[fill=hmLightCyan] (2,-0.8) rectangle(4,0.8);
        \draw[draw=hmLineGreen,line width=1.1pt] (0,-1.5) -- (0,1.5);
        \draw[draw=hmLineLightBlue,line width=1.1pt] (2,-0.8) -- (2,0.8);
        \draw[draw=hmLineLightBlue,line width=1.1pt] (2,-0.8) -- (4,-0.8);
        \draw[draw=hmLineLightBlue,line width=1.1pt] (2,0.8) -- (4,0.8);
        \draw[draw=hmLineGreen, line width=1.1pt] (0,0) -- (2,0);
         \filldraw[draw= hmLineLightBlue,line width=1.1pt] (2,0) -- (4,0);
        \filldraw[fill=white,draw=hmLineGreen,line width=0.7pt] (-0.05,-0.05) rectangle(0.05,0.05);
        \filldraw[fill=white,draw=hmLineGreen,line width=0.7pt] (1.95,-0.05) rectangle(2.05,0.05);

        \node at(-1.5,0){$\CC$};
        \node at(-0.3,1.1){$\hind{}{\vind{B_1}{\CC}{B_1}}{B_1}$};
        \node at(-0.3,-1.1){$\hind{}{\vind{B_2}{\CC}{B_2}}{B_2}$};
        \node at(-0.3,0){$\hind{}{\CC}{A_1}$};
        \node at(1,0.8){$\hind{B_1}{\vind{B_1}{\CC}{B_1}}{B_1}$};
        \node at(1.8,1){$\hind{}{\CC}{B_3}$};
        \node at(1.8,-1){$\hind{}{\CC}{B_4}$};
        \node at(1,-0.3){\tiny$\hind{A_1}{\vind{B_1}{\CC}{B_2}}{A_1}$};
        \node at(2.4,-0.3){\tiny$\hind{A_1}{\vind{B_3}{\CC}{B_4}}{A_2}$};
          \node at(3,1.1){$\hind{B_3}{\vind{B_3}{\CC}{B_3}}{B_3}$};
        \node at(3,-1.1){$\hind{B_4}{\vind{B_4}{\CC}{B_4}}{B_4}$};
      \end{tikzpicture}}
    \end{minipage}\hfill%
    \begin{minipage}[c]{0.33\textwidth}
      \centering
      \resizebox{0.9\linewidth}{!}{%
 \begin{tikzpicture}
        \path[fill=hmDeepCyan] (-2,-1.5) rectangle(0,1.5);
        \path[fill=hmMidCyan] (0,-1.5) rectangle(2,1.5);
         \path[fill=hmLightCyan] (2,-1.5) rectangle(4,1.5);
        \draw[draw=hmLineGreen,line width=1.1pt] (0,-1.5) -- (0,1.5);
        \draw[draw=hmLineLightBlue,line width=1.1pt] (2,-1.5) -- (2,1.5);
        \draw[draw=hmLineGreen, line width=1.1pt] (0,0) -- (2,0);
         \filldraw[draw= hmLineLightBlue,line width=1.1pt] (2,0) -- (4,0);
        \filldraw[fill=white,draw=hmLineGreen,line width=0.7pt] (-0.05,-0.05) rectangle(0.05,0.05);
        \filldraw[fill=white,draw=hmLineGreen,line width=0.7pt] (1.95,-0.05) rectangle(2.05,0.05);

        \node at(-1.5,0){$\CC$};
        \node at(-0.3,1.1){$\hind{}{\vind{B_1}{\CC}{B_1}}{B_1}$};
        \node at(-0.3,-1.1){$\hind{}{\vind{B_2}{\CC}{B_2}}{B_2}$};
        \node at(-0.3,0){$\hind{}{\CC}{A_1}$};
        \node at(1,0.8){$\hind{B_1}{\vind{B_1}{\CC}{B_1}}{B_1}$};
        \node at(1.8,1){$\hind{}{\CC}{B_3}$};
        \node at(1.8,-1){$\hind{}{\CC}{B_4}$};
        \node at(1,-0.3){\tiny$\hind{A_1}{\vind{B_1}{\CC}{B_2}}{A_1}$};
        \node at(2.4,-0.3){\tiny$\hind{A_1}{\vind{B_3}{\CC}{B_4}}{A_2}$};
          \node at(3,1.1){$\hind{B_3}{\vind{B_3}{\CC}{B_3}}{B_3}$};
        \node at(3,-1.1){$\hind{B_4}{\vind{B_4}{\CC}{B_4}}{B_4}$};
      \end{tikzpicture}}
    \end{minipage}\par
\caption{The same iterated condensation viewed at three resolutions.  The left panel is an ordinary globular cell in $\Mrt_{E_1}(\LMod^{\mathrm{fin}}(\CC))$.  The middle panel resolves non-trivial side condensation morphisms $B_1\hookrightarrow B_3$ and $B_2\hookrightarrow B_4$, thereby viewing the cell in $\Mrt_{E_1}^{\mathrm{top}}(\LMod^{\mathrm{fin}}(\CC))$; the right panel records the corresponding $B_3$--$B_4$ bi-bimodule data.}
\label{fig:topological-three-resolutions}
\end{figure}

The second and third panels of \cref{fig:topological-three-resolutions} are the additional ``peeling'' that is unavailable in the ordinary globular target: after resolving the side condensation morphisms, the interior defect is naturally described by a bi-bimodule with the larger $B_3$--$B_4$ actions.  This is precisely the setup in which a second condensation can be fused with the first one without leaving the new language.

To formulate the functorial version of this statement, let
\begin{align*}
    \Fun^{\mathrm{top},\CC}_{(f_1,f_2)}
    \left(
      \vind{B_1}{\hind{}{\CC}{A_1}}{B_2},
      \vind{B_3}{\hind{}{\CC}{A_2}}{B_4}
    \right)
\end{align*}
denote the category of $\CC$-linear functors equipped with the two boundary change constraints induced by $f_1$ and $f_2$, together with the usual mixed coherence.  Equivalently, after choosing the $B_3$--$B_4$ extension of the source action, this is the boundary changing version of the bimodule functor category
\begin{align*}
    \hind{\CC}{\vind{\hind{}{\CC}{B_3}}{\Fun}{\hind{}{\CC}{B_4}}}{}
    \left(
      \vind{B_1}{\hind{}{\CC}{A_1}}{B_2},
      \vind{B_3}{\hind{}{\CC}{A_2}}{B_4}
    \right).
\end{align*}
The fixed-endpoint $E_2$ Eilenberg--Watts theorem proved above does not by itself identify these boundary changing cells.  We therefore isolate the required extension as a conjecture.

\begin{cnj}[$E_2$ generalized Eilenberg--Watts]\label{cnj:E2-generalized-EW}
With the data above, tensoring with a bi-bimodule induces an equivalence
\begin{align*}
    \hind{A_1}{\vind{B_3}{\CC}{B_4}}{A_2}
    \simeq
    \Fun^{\mathrm{top},\CC}_{(f_1,f_2)}
    \left(
      \vind{B_1}{\hind{}{\CC}{A_1}}{B_2},
      \vind{B_3}{\hind{}{\CC}{A_2}}{B_4}
    \right),
    \qquad
    M\longmapsto -\otc{A_1}M.
\end{align*}
In the shorthand notation above, the right-hand side is written as
\begin{align*}
    \hind{\CC}{\vind{\hind{}{\CC}{B_3}}{\Fun}{\hind{}{\CC}{B_4}}}{}
    \left(
      \vind{B_1}{\hind{}{\CC}{A_1}}{B_2},
      \vind{B_3}{\hind{}{\CC}{A_2}}{B_4}
    \right),
\end{align*}
with the boundary-change data understood.
\end{cnj}

Assuming Conjecture \ref{cnj:E2-generalized-EW}, the square represented by $M$ is literally a $2$-morphism in the topological enlargement.  Its composition with the first-stage condensation defect is the relative tensor product of the corresponding representable module categories.  In the configuration displayed below, this takes the schematic form
\begin{align}\label{eq:second-stage-topological-fusion}
    \vind{B_1}{\hind{}{\CC}{A_1}}{B_2}
    \btc{\vind{B_3}{\hind{A_1}{\CC}{A_1}}{B_2}}
    \vind{B_3}{\hind{A_1}{\CC}{A_2}}{B_4}
    \simeq
    \vind{B_3}{\hind{}{\CC}{A_2}}{B_4}.
\end{align}
This is not a new fusion operation: it is the same Morita multiplication, now applied to a non-globular $2$-cell with non-trivial side condensation boundaries.

\begin{figure}[H]
\centering
\begin{center}
  \makebox[\textwidth][c]{%
      \resizebox{0.29\textwidth}{!}{%
      \begin{tikzpicture}[baseline=(current bounding box.center)]
        \path[fill=hmDeepCyan] (-2,-1.5) rectangle(0,1.5);
        \path[fill=hmMidCyan] (0,-1.5) rectangle(2,1.5);
        \path[fill=hmLightCyan] (2,-1.5) rectangle(4,1.5);
        \draw[draw=hmLineGreen,line width=1.1pt] (0,-1.5) -- (0,1.5);
        \draw[draw=hmLineLightBlue,line width=1.1pt] (2,-1.5) -- (2,1.5);
        \draw[draw=hmLineGreen,line width=1.1pt] (0,0) -- (2,0);
        \draw[draw=hmLineGreen,line width=0.9pt,decorate,
          decoration={brace,amplitude=4pt}]
          (0.08,0.13) -- (1.92,0.13);
        \draw[draw=hmLineLightBlue,line width=1.1pt] (2,0) -- (4,0);
        \filldraw[fill=white,draw=hmLineGreen,line width=0.7pt]
          (-0.05,-0.05) rectangle(0.05,0.05);
        \filldraw[fill=white,draw=hmLineGreen,line width=0.7pt]
          (1.95,-0.05) rectangle(2.05,0.05);

        \node at(-1.5,0){$\CC$};
        \node at(-0.3,1.1){$\hind{}{\vind{B_1}{\CC}{B_1}}{B_1}$};
        \node at(-0.3,-1.1){$\hind{}{\vind{B_2}{\CC}{B_2}}{B_2}$};
        \node at(-0.3,0){$\hind{}{\CC}{A_1}$};
        \node at(1.7,1){$\hind{B_1}{\vind{B_3}{\CC}{B_3}}{B_3}$};
         \node at(1.7,-1){$\hind{B_2}{\vind{B_4}{\CC}{B_4}}{B_4}$};
        \node at(1,-0.3){\tiny$\hind{A_1}{\vind{B_1}{\CC}{B_2}}{A_1}$};
        \node at(2.3,-0.3){\tiny$\hind{A_1}{\vind{B_3}{\CC}{B_4}}{A_2}$};
        \node at(3.2,1){$\hind{B_3}{\vind{B_3}{\CC}{B_3}}{B_3}$};
        \node at(3.2,-1){$\hind{B_4}{\vind{B_4}{\CC}{B_4}}{B_4}$};
      \end{tikzpicture}%
      }%
      \hspace{0.015\textwidth}%
      \begin{tikzpicture}[baseline=0pt]
        \node[anchor=south] (formula) at (0,0.12cm) {\resizebox{0.14\textwidth}{!}{$
          \vind{B_1}{\hind{}{\CC}{A_1}}{B_2}
          \btc{\vind{B_3}{\hind{A_1}{\CC}{A_1}}{B_2}}
          \vind{B_3}{\hind{A_1}{\CC}{A_2}}{B_4}
          \simeq
          \vind{B_1}{\hind{}{\CC}{A_2}}{B_4}
        $}};
        \draw[draw=hmLineGreen,-{Latex[length=2.3mm]},line width=1pt]
          ([yshift=-0.12cm]formula.south west) --
          ([yshift=-0.12cm]formula.south east);
      \end{tikzpicture}%
      \hspace{0.015\textwidth}%
      \resizebox{0.29\textwidth}{!}{%
  \begin{tikzpicture}[baseline=(current bounding box.center)]
        \path[fill=hmDeepCyan] (-2,-1.5) rectangle(0,1.5);
        \path[fill=hmLightCyan] (0,-1.5) rectangle(4,1.5);
        \draw[draw=hmLineLightBlue,line width=1.1pt] (0,-1.5) -- (0,1.5);
        \draw[draw=hmLineLightBlue,line width=1.1pt] (0,0) -- (4,0);
        \filldraw[fill=white,draw=hmLineGreen,line width=0.7pt]
          (-0.05,-0.05) rectangle(0.05,0.05);

        \node at(-1.5,0){$\CC$};
        \node at(-0.3,1.1){$\hind{}{\vind{B_3}{\CC}{B_3}}{B_3}$};
        \node at(-0.3,-1.1){$\hind{}{\vind{B_4}{\CC}{B_4}}{B_4}$};
        \node at(-0.3,0){$\hind{}{\CC}{A_2}$};
        \node at(1.5,1){$\hind{B_3}{\vind{B_3}{\CC}{B_3}}{B_3}$};
        \node at(1.5,-1){$\hind{B_4}{\vind{B_4}{\CC}{B_4}}{B_4}$};
        \node at(2.5,-0.3){\tiny$\hind{A_2}{\vind{B_3}{\CC}{B_4}}{A_2}$};
      \end{tikzpicture}%
      }%
  }
\end{center}
\caption{A second-stage condensation defect as a topological $2$-cell and its fusion with the first-stage defect.  The side walls change the outer labels from $(B_1,B_2)$ to $(B_3,B_4)$, and the middle relative tensor product is the composition law in $\Mrt_{E_1}^{\mathrm{top}}(\LMod^{\mathrm{fin}}(\CC))$.}
\label{fig:topological-second-stage-fusion}
\end{figure}
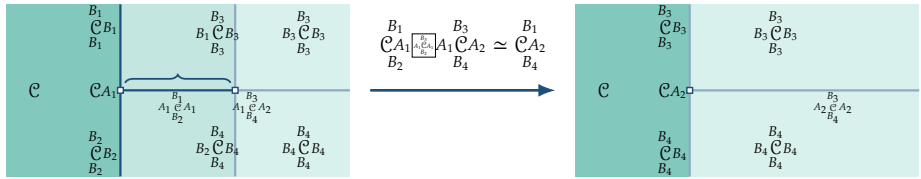

The relation with two-step condensation is now transparent.  The first passage
\begin{align*}
    \LMod^{\mathrm{fin}}(\CC)
    \rightsquigarrow
    \Mrt_{E_1}\bigl(\LMod^{\mathrm{fin}}(\CC)\bigr)
\end{align*}
opens an ordinary module object into a representable Morita cell and exposes the first pair of $E_2$-actions.  The inclusion
\begin{align*}
    \Mrt_{E_1}\bigl(\LMod^{\mathrm{fin}}(\CC)\bigr)
    \hookrightarrow
    \Mrt_{E_1}^{\mathrm{top}}\bigl(\LMod^{\mathrm{fin}}(\CC)\bigr)
\end{align*}
then allows non-trivial side condensation morphisms and exposes a second pair of outer actions.  Repeating the same relative-tensor-product calculus fuses the intermediate strips.  
\subsubsection*{Automorphism-induced invertible walls}

The topological enlargement also clarifies the role of algebra automorphisms in two-step condensation.  Let $B$ be a condensable $E_2$-algebra and let
$\varphi\in\Aut_{\Alg_{E_2}(\CC)}(B)$.  Precomposing one of the regular $B$-actions with $\varphi$ gives a twisted regular bi-bimodule, which we denote by $B^{\mathrm{tw}}$.  Since $\varphi$ is invertible, $B^{\mathrm{tw}}$ is invertible with inverse obtained from $\varphi^{-1}$.  Under $\hind{B}{\vind{B}{\CC}{B}}{B}$
this determines an invertible one-dimensional domain wall in the $B$-condensed phase; denote it by $\Phi_{\varphi}$.  The wall may be non-trivial before the intermediate $B$-strip is fused away.

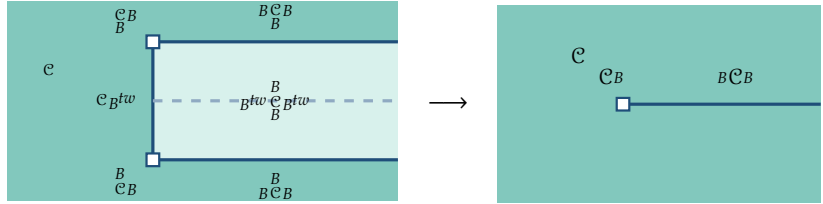
\begin{figure}[H]
\centering
\noindent
\begin{minipage}[c]{0.329\textwidth}
  \centering
  \resizebox{\linewidth}{!}{%
  \begin{tikzpicture}[every node/.style={font=\tiny}]
    \path[fill=hmDeepCyan] (0,-1.1) rectangle (4.3,1.1);
    \path[fill=hmLightCyan] (1.6,-0.65) rectangle (4.3,0.65);
    \draw[draw=hmLineGreen,line width=1pt]
      (1.6,-0.65) -- (1.6,0.65);
    \draw[draw=hmLineGreen,line width=1pt]
      (1.6,0.65) -- (4.3,0.65);
    \draw[draw=hmLineGreen,line width=1pt]
      (1.6,-0.65) -- (4.3,-0.65);
    \filldraw[fill=white,draw=hmLineGreen,line width=0.7pt]
      (1.53,0.58) rectangle (1.67,0.72);
    \filldraw[fill=white,draw=hmLineGreen,line width=0.7pt]
      (1.53,-0.72) rectangle (1.67,-0.58);
    \draw[draw=hmLineLightBlue,dashed,line width=1pt]
      (1.6,0) -- (4.3,0);

    \node at (0.45,0.35) {$\CC$};
    \node[anchor=east] at (1.55,0) {$\hind{}{\CC}{B^{tw}}$};
    \node[anchor=south east] at (1.55,0.6)
      {$\hind{}{\vind{}{\CC}{B}}{B}$};
    \node[anchor=north east] at (1.55,-0.6)
      {$\hind{}{\vind{B}{\CC}{}}{B}$};
    \node[above] at (2.95,0.65)
      {$\hind{B}{\vind{}{\CC}{B}}{B}$};
    \node[below] at (2.95,-0.65)
      {$\hind{B}{\vind{B}{\CC}{}}{B}$};
    \node at (2.95,0)
      {$\hind{B^{tw}}{\vind{B}{\CC}{B}}{B^{tw}}$};
  \end{tikzpicture}%
  }
\end{minipage}\hspace{0.02\textwidth}%
\begin{minipage}[c]{0.042\textwidth}
  \centering
  $\longrightarrow$
\end{minipage}\hspace{0.02\textwidth}%
\begin{minipage}[c]{0.273\textwidth}
  \centering
  \resizebox{\linewidth}{!}{%
  \begin{tikzpicture}[every node/.style={font=\scriptsize}]
    \path[fill=hmDeepCyan] (0,-1.1) rectangle (3.6,1.1);
    \draw[draw=hmLineGreen,line width=1pt]
      (1.4,0) -- (3.6,0);
    \filldraw[fill=white,draw=hmLineGreen,line width=0.7pt]
      (1.33,-0.07) rectangle (1.47,0.07);

    \node at (0.9,0.55) {$\CC$};
    \node[anchor=south east] at (1.53,0.1) {$\hind{}{\CC}{B}$};
    \node[above] at (2.65,0.1) {$\hind{B}{\CC}{B}$};
  \end{tikzpicture}%
  }
\end{minipage}\par
\caption{An algebra automorphism of $B$ may induce a non-trivial invertible one-dimensional wall $\Phi_{\varphi}$ inside $\hind{B}{\vind{B}{\CC}{B}}{B}$. After the two $B$-condensation side walls are fused away, only the underlying $1$--$1$ bimodule remains, so the automorphism twist has no effect on the resulting fused defect in $\CC$.}
\label{fig:automorphism-closing}
\end{figure}

\begin{prp}[Automorphism twists are invisible after closing]\label{prp:automorphism_wall_closing}
Let $\Phi_{\varphi}$ be the invertible wall induced by an algebra automorphism $\varphi$ of $B$.  Fusing it with the two $B$-condensation side walls gives the same fused defect as in the untwisted case.  At the level of bi-bimodule labels,
\begin{align*}
      \vfus{\vfus{B}{\otc{B}}{B^{\mathrm{tw}}}}{\otc{B}}{B}
    \simeq
    \vind{1}{B^{\mathrm{tw}}}{1}
    \cong
    B.
\end{align*}
\end{prp}
\begin{proof}
The first equivalence is the unit property of the relative tensor product over $B$.  The difference between $B^{\mathrm{tw}}$ and the regular bimodule $B$ lies entirely in the $B$-action precomposed with $\varphi$.  Once the two exterior $B$-actions are closed by restriction along the unit $1\to B$, this twist is forgotten, and $B^{\mathrm{tw}}$ has the same underlying $1$--$1$ bimodule as $B$.  The bi-bimodule multiplication therefore identifies the closed composite with the untwisted one.
\end{proof}

At the categorical level, Proposition
\ref{prp:automorphism_wall_closing} is precisely what makes the
automorphism twist disappear under the vertical fusion shown in
Figure~\ref{fig:automorphism-closing}.  Applying
Theorem~\ref{thm:vertical_fusion_of_bi-bimod} twice gives
\begin{align*}
\vfus{
    \vfus{
        \hind{B}{\vind{}{\CC}{B}}{B}
    }{
        \btc{\hind{B}{\vind{B}{\CC}{B}}{B}}
    }{
        \hind{B^{\mathrm{tw}}}{\vind{B}{\CC}{B}}{B^{\mathrm{tw}}}
    }
}{
    \btc{\hind{B}{\vind{B}{\CC}{B}}{B}}
}{
    \hind{B}{\vind{B}{\CC}{}}{B}
}
\simeq
\hind{
    \vfus{\vfus{B}{\otc{B}}{B^{\mathrm{tw}}}}{\otc{B}}{B}
}{
    \CC
}{
    \vfus{\vfus{B}{\otc{B}}{B^{\mathrm{tw}}}}{\otc{B}}{B}
}
\simeq
\hind{B}{\CC}{B},
\end{align*}
where the last equivalence is exactly
Proposition~\ref{prp:automorphism_wall_closing}.
Thus the non-trivial wall
$\hind{B^{\mathrm{tw}}}{\vind{B}{\CC}{B}}{B^{\mathrm{tw}}}$
in the intermediate $B$-condensed phase leaves no twist in the vertically
fused wall.

The same cancellation occurs at the endpoint point defect.  The point
$\hind{}{\CC}{B^{\mathrm{tw}}}$ in the left-hand side of
Figure~\ref{fig:automorphism-closing} is carried by the same composite
algebra label, and hence after closing the two transverse condensation
walls one obtains
\begin{align*}
    \hind{}{\CC}{
        \vfus{\vfus{B}{\otc{B}}{B^{\mathrm{tw}}}}{\otc{B}}{B}
    }
    \simeq
    \hind{}{\CC}{B}.
\end{align*}
Consequently, both the one-dimensional wall and its endpoint reduce to
their untwisted counterparts after closing.  The reason is not that
$\Phi_\varphi$ is trivial inside the $B$-condensed phase, but that the
two $B$-actions carrying the twist have been fused away; the remaining
$1$--$1$ actions on $B^{\mathrm{tw}}$ are trivial.

Thus an algebra automorphism may produce a genuine invertible wall inside the intermediate condensed phase, while having no effect on the final defect after the intermediate strip is fused away.  Equivalently, two second-stage condensation configurations that differ only by insertion of an automorphism-induced wall $\Phi_{\varphi}$ represent the same fused defect in $\CC$.  This is the bi-bimodule reformulation of the automorphism ambiguity in two-step condensation discussed in \cite{XY25}; in the present language the mechanism is simply relative-tensor-product multiplication together with restriction of the outer $B$-actions.

\subsection{Local 2-Morita Equivalence and Local Modules}\label{subsec:local-two-morita}
The twice-looped endomorphism category of an \(E_2\)-algebra \(B\) in $\Mrt_{E_2}(\CC)$ 
\begin{align*}
    \Omega_{B}^2(\Mrt_{E_2}(\CC))=\hind{B}{\vind{B}{\CC}{B}}{B}\simeq \hind{}{\CC}{B}.
\end{align*}
recovers its category of $E_2$-local $E_0$-modules. This motivates a weaker notion of \(2\)-Morita equivalence detected only by this local invariant.

\begin{dfn}
    Two $E_2$-algebras $B_1$ and $B_2$ are {\bf locally 2-Morita equivalent} if their $E_2$-locl $E_0$-module categories $\hind{B_1}{\vind{B_1}{\CC}{B_1}}{B_1}$ and $\hind{B_2}{\vind{B_2}{\CC}{B_2}}{B_2}$ are $E_2$-monoidal equivalent.
\end{dfn}

Since iterated endomorphism categories are invariant under equivalence of objects in the Morita tricategory, \(2\)-Morita equivalence always implies this local version.
\begin{prp}\label{prp:2-Morita_equiv_implies_local_2-Morita_equiv}
    If two $E_2$-algebras are 2-Morita equivalent, then they are locally 2-Morita equivalent.
\end{prp}
\begin{proof}
    For equivalent object $B_1$ and $B_2$ in $\Mrt_{E_2}(\CC)$, we have $\Omega^2_{B_1}(\Mrt_{E_2}(\CC))\simeq \Omega^2_{B_2}(\Mrt_{E_2}(\CC))$ as $E_2$-monoidal categories.
\end{proof}

In general, the local invariant need not a priori recover the full \(2\)-Morita class. For condensable algebras in a modular fusion category, however, the converse holds.

\begin{thm}[\cite{XY25}, Proposition~C.11]
    For a modular fusion category $\CC$, two condensable $E_2$-algebra $B_1$ and $B_2$ are $E_2$-Morita equivalent if and only if the are locally 2-Morita equivalent.
\end{thm}
\begin{proof}
    $\Rightarrow$: Proposition \ref{prp:2-Morita_equiv_implies_local_2-Morita_equiv}.
    $\Leftarrow$: It was proved in \cite{XY25} using double centralizer theorem \cite{Mueg03b}.
\end{proof}

\section{n-Morita Equivalences}\label{sec:n-morita}
This section formulates the general recursive pattern of $n$-Morita equivalence suggested by the preceding $n=1,2$ constructions and compares the resulting descriptions at each morphism level.

\begin{dfn}[Haugseng's recursive model, \cite{Hau17}]
    Let $\CC$ be an $E_{m+2}$-monoidal $1$-category.
    The $(n+1)$-category $\Mrt_{E_n}(\CC)$ consists of 
    \begin{itemize}
        \item objects are $E_n$-algebras in $\CC$, i.e. objects in $\Alg_{E_n}(\CC)$;
        \item for two $E_n$-algebras $\mathsf{X}_1$ and $\mathsf{X}_2$, the mapping $n$-category is $\Mrt_{E_{n-1}}(\hind{\mathsf{X}_1}{\CC}{\mathsf{X}_2})$.
    \end{itemize}
\end{dfn}
In details, we have 
\begin{dfn}[Haugseng's iterated-bimodule model, \cite{Hau17}]
    Let $\CC$ be an $E_{m+2}$-monoidal $1$-category.
    The $(n+1)$-category $\Mrt^{Haug}_{E_n}(\CC)$ consists of 
    \begin{itemize}
        \item objects are $E_n$-algebras in $\CC$, i.e. objects in $\Alg_{E_n}(\CC)$;
        \item for two $E_n$-algebras $\mathsf{X}_1$ and $\mathsf{X}_2$, 1-morphisms are objects of $\Alg_{E_{n-1}}(\hind{\sX_1}{\CC}{\sX_2})$;
        \item 2-morphisms are objects of $\Alg_{E_{n-2}}(\hind{\sW_1}{\hind{\sX_1}{\CC}{\sX_2}}{\sW_2})$;
        \item $\cdots$.
    \end{itemize}
\end{dfn}

\begin{dfn}[GS's iterated-bimodule model]
    Let $\CC$ be an $E_{m+2}$-monoidal $1$-category.
    The $(n+1)$-category $\Mrt^{GS}_{E_n}(\CC)$ consists of 
    \begin{itemize}
        \item objects are $E_n$-algebras in $\CC$, i.e. objects in $\Alg_{E_n}(\CC)$;
        \item for two $E_n$-algebras $\mathsf{X}_1$ and $\mathsf{X}_2$, 1-morphisms are objects of $\hind{\sX_1}{\Alg_{E_{n-1}}(\CC)}{\sX_2}$;
        \item 2-morphisms are objects of $\hind{\sW_1}{\hind{\sX_1}{\Alg_{E_{n-2}}(\CC)}{\sX_2}}{\sW_2}$;
        \item $\cdots$.
    \end{itemize}
\end{dfn}

\begin{dfn}[bi-bimodule model]
    Let $\CC$ be an $E_{m+2}$-monoidal $1$-category.
    The $(n+1)$-category $\Mrt_{E_n}(\CC)$ consists of 
    \begin{itemize}
        \item objects are $E_n$-algebras in $\CC$, i.e. objects in $\Alg_{E_n}(\CC)$;
        \item for two $E_n$-algebras $\mathsf{X}_1$ and $\mathsf{X}_2$, 1-morphisms are objects of $\hind{\sX_1}{\Alg_{E_{n-1}}(\CC)}{\sX_2}$;
        \item 2-morphisms are objects of $\hind{\sW_1}{\vind{\sX_1}{\Alg_{E_{n-2}}(\CC)}{\sX_2}}{\sW_2}$;
        \item $\cdots$.
    \end{itemize}
\end{dfn}
where our higher morphism is defined by bi-bi-$\cdots$-bi-bimodules, i.e. a bimodule category that satisfies all higher coherence conditions.

\begin{crl}
    Let $\mathsf{X}_1$ and $\mathsf{X}_2$ be two $E_n$-algebras in an $E_m$-monoidal category $\CC$ ($m\geq n$), then $\hind{\mathsf{X}_1}{\Alg_{E_{n-1}}(\CC)}{\mathsf{X}_2}\simeq \Alg_{E_{n-1}}(\hind{\mathsf{X}_1}{\CC}{\mathsf{X}_2})$.
\end{crl}
\begin{proof}
    We prove by induction.
    We have proved for $n=0$ and $n=1$ case, suppose this holds for $E_{n-1}$-algebras $\mathsf{X}_1$, $\mathsf{X}_2$ in some category $\CE$, i.e. $\hind{\mathsf{X}_1}{\Alg_{E_{n-2}}(\CE)}{\mathsf{X}_2}\simeq \Alg_{E_{n-2}}(\hind{\mathsf{X}_1}{\CE}{\mathsf{X}_2})$, now we prove for $E_n$ case.
     
    Note that $\Alg_{E_i}(\CD)=\Alg_{E_{i-1}}(\Alg_{E_1}(\CD))$ for any $E_{m}$-monoidal category $\CD$ ($m\geq i$) and $i$.
    Let $\CD$ be $\hind{\mathsf{X}_1}{\CC}{\mathsf{X}_2}$, we have $\Alg_{E_{n-1}}(\hind{\mathsf{X}_1}{\CC}{\mathsf{X}_2})\simeq \Alg_{E_{n-2}}(\Alg_{E_1}(\hind{\mathsf{X}_1}{\CC}{\mathsf{X}_2}))$.
    By Theorem \ref{thm:algebras_in_bimodule_categories}, we have 
    \begin{align*}
        \Alg_{E_{n-2}}(\Alg_{E_1}(\hind{\mathsf{X}_1}{\CC}{\mathsf{X}_2}))\simeq \Alg_{E_{n-2}}(\hind{\mathsf{X}_1}{\Alg_{E_{1}}(\CC)}{\mathsf{X}_2}).
    \end{align*}
    Since $\mathsf{X}_1$, $\mathsf{X}_2$ are now $E_{n-1}$-algebras in $\Alg_{E_{1}}(\CC)$, let $\CE=\Alg_{E_{1}}(\CC)$ in $E_{n-1}$ case, we have 
    \begin{align*}
        \hind{\mathsf{X}_1}{\Alg_{E_{n-2}}(\Alg_{E_{1}}(\CC))}{\mathsf{X}_2}\simeq \Alg_{E_{n-2}}(\hind{\mathsf{X}_1}{\Alg_{E_{1}}(\CC)}{\mathsf{X}_2})
    \end{align*}
    Therefore we have $\hind{\mathsf{X}_1}{\Alg_{E_{n-1}}(\CC)}{\mathsf{X}_2}\simeq \Alg_{E_{n-1}}(\hind{\mathsf{X}_1}{\CC}{\mathsf{X}_2})$.
\end{proof}

Using above corollary, we can iteratively prove the following statement and even for more layers
\begin{crl}
    Let $\mathsf{X}_1$ and $\mathsf{X}_2$ be two $E_n$-algebras in an $E_m$-monoidal category $\CC$ ($m\geq n$) and let $W_1$ and $W_2$ be two $E_{n-1}$-algebra $\mathsf{X}_1$-$\mathsf{X}_2$-bimodules, i.e. objects in $\hind{\mathsf{X}_1}{\Alg_{E_{n-1}}(\CC)}{\mathsf{X}_2}\simeq \Alg_{E_{n-1}}(\hind{\mathsf{X}_1}{\CC}{\mathsf{X}_2})$.
    We have 
    \begin{align*}
        \hind{W_1}{\hind{\mathsf{X}_1}{\Alg_{E_{n-2}}(\CC)}{\mathsf{X}_2}}{W_2}\simeq \Alg_{E_{n-2}}(\hind{W_1}{\hind{\mathsf{X}_1}{\CC}{\mathsf{X}_2}}{W_2})
    \end{align*}
\end{crl}
 
Using above sequence of corollaries we can show the equivalence of each level of morphisms between Haugseng's model and GS's model 
\begin{thm}
    For each $k\leq n$, we have 
    \begin{align*}
        \hom^{k}_{\Mrt_{E_n}^{Haug}(\CC)}\simeq \hom^{k}_{\Mrt_{E_n}^{GS}(\CC)}
    \end{align*}
\end{thm}

Using Theorem \ref{thm:bibimod_equivalent_to_bimod_in_bimod}, we have 
\begin{prp}
    there is an equivalence 
    \begin{align*}
        \hind{\sW_1}{\hind{\sX_1}{\Alg_{E_{n-2}}(\CC)}{\sX_2}}{\sW_2}\simeq \hind{\sW_1}{\vind{\sX_1}{\Alg_{E_{n-2}}(\CC)}{\sX_2}}{\sW_2}
    \end{align*}
    for any $E_n$-algebras $X_1, X_2$ and their $E_{n-1}$-bimodules $W_1$ and $W_2$.
\end{prp}
So at least for $n=2$ we can prove that 
\begin{thm}
    For each $k\leq n$, we have 
    \begin{align*}
        \hom^{k}_{\Mrt_{E_n}^{Haug}(\CC)}\simeq \hom^{k}_{\Mrt_{E_n}^{GS}(\CC)}\simeq \hom^{k}_{\Mrt_{E_n}(\CC)}
    \end{align*}
\end{thm}

We propose that the category of bi-bi-bimodules is equivalent to the category of bimodules internalized to bi-bi-modules and this equivalence also holds for $n=4,5,\dots$ cases.
Thus for any $n$, there should be 
\begin{cnj}
    For each $k\leq n$, we have 
    \begin{align*}
        \hom^{k}_{\Mrt_{E_n}^{Haug}(\CC)}\simeq \hom^{k}_{\Mrt_{E_n}^{GS}(\CC)}\simeq \hom^{k}_{\Mrt_{E_n}(\CC)}
    \end{align*}
\end{cnj}

Fully faithfulness keeps the $E_1$-algebraic defect data visible after passing to module categories.
    We now propose the higher analogue:   
    \[
    \mathsf{Mod_n}\colon
    \Mrt_{E_n}(\CC)
    \longrightarrow
    \Mrt_{E_{n-1}}\bigl(\LMod(\CC)\bigr).
  \]
        \vspace{-1em}
  \[
    \begin{array}{@{}l@{\quad}c@{\quad}l@{\quad}c@{\quad}l@{}}
      E_n \text{-algs}
      & \longrightarrow &
      E_{n-1}\text{-alg bimodules}
      & \longrightarrow &
      E_{n-2}\text{-alg bimodules between}\, E_{n-1}\text{-alg bimodules}
      \cdots.\\[0.4em]
      E_{n-1}\text{-module cats}
      & \longrightarrow &
      E_{n-2}\text{-bimodule cats}
      & \longrightarrow &
      E_{n-3}\text{-bimodule cats between} E_{n-2}\text{-bimodule cats}
      \cdots.
    \end{array}
  \]

\appendix
\appendixpage

\section{Some Results on category of $B_1$-$B_2$-bimodules}
We collect several general facts about monoidal structures and
free--forgetful adjunctions that will be used below.
We first recall two elementary ways of inheriting and transporting
monoidal structures. We then apply these observations to the category
$\hind{B_1}{\CC}{B_2}$ of $B_1$-$B_2$-bimodules.
The main purpose is to identify the free functor
$B_1\ot -\ot B_2$ as a strong monoidal left adjoint and, consequently,
to equip the corresponding forgetful functor with a canonical lax
monoidal structure.

\begin{lem}\label{lem:inheritance_monoidal_of_full_sub_cat}
    Let $(\CC,\otimes,\bfone,\alpha,\lambda,\rho)$ be a monoidal category.
    Let $\CE\subset \CC$ be a full subcategory of $\CC$ such that
    \begin{enumerate}
        \item $\bfone\in \CE$
        \item and for any $U,V\in \CE$, $U\ot V\in \CE$.
    \end{enumerate}
    Then $\CE$ inherits the monoidal sturcture from $\CC$ and the inclusion functor $I:\CE\to \CC$ is strong monoidal.
\end{lem}
\begin{proof}
    We construct monoidal structure on $\CE$ by directly inheritance.
    \begin{itemize}
        \item First we restrict the tensor product on $\CE$ to obtain $\otimes_{\CE}:\CE\times \CE\to \CE$.
        The closedness of $\otimes_{\CE}$ on object is ensured by assumption, and on morphism level, since $\CE$ is a full subcategory of $\CC$, so it is also quite obvious.
        Thus $\otimes_{\CE}$ is indeed a bifunctor.

        \item  Tensor unit of $\CE$ is again $\bfone$.
        
        \item For any $U,V,W\in \CE$, $U\ot_{\CE}(V\ot_{\CE} W)$ and $(U\ot_{\CE}V\ot_{\CE}) W$ are both object in $\CE$, so the associator $\alpha_{U,V,W}$ is also an isomorphism in $\CE$ by fully faithfulness.
        We can directly use it as the associator in $\CE$.
        
        \item The same argument applies to the left and right unitors.
    \end{itemize}
    All coherence axioms are inherited from $\CC$.
\end{proof}

The preceding observation also allows a monoidal structure to be
transported along a fully faithful functor whose essential image is closed
under tensor products.
\begin{lem}
    Let $(\CC,\otimes,\bfone,\alpha,\lambda,\rho)$ be a monoidal category.
    Let $F:\CD\to \CC$ be a fully faithful functor such that 
    \begin{enumerate}
        \item there exists an object $I\in \CD$ such that $F(I)\cong \bfone$.
        \item for any $X,Y\in \CD$, there is an object $T\in \CD$ such that there is a natural isomorphism 
        \begin{align*}
            F(T)\cong F(X)\otimes F(Y) 
        \end{align*}
    \end{enumerate}
    Then $\CD$ can be equipped with a monoidal structure such that $F$ is a strong monoidal functor.
\end{lem}
\begin{proof}
    Let $\CE=\mathrm{EssIm}(F)$, the full subcategory of $\CC$ consists of objects that are isomorphic to $F(X)$ for some $X\in \CD$.  
    Then by Lemma \ref{lem:inheritance_monoidal_of_full_sub_cat}, we have $\CE$ is monoidal.

    Now since $F:\CD\to \CE$ is fully faithful and essentially surjective, then $F$ is an equivalence between categories, meaning there is a quasi-inverse $G:\CE\to \CD$ and natural isomorphism $\epsilon:FG\cong \Id_{\CE}$.

    Define $X\otimes_{\CD} Y:=G(F(X)\otimes F(Y))=GF(T)\cong T$ and $\bfone_{\CD}:=G(\bfone)=GF(I)\cong I$.
    Since $F, G,\otimes$ are both functors, so $\otimes_{\CD}$ is a bifunctor.

    Now define $\alpha^{\CD}_{X,Y,Z}$ be the image of $F(X\ot_{\CD}(Y\ot_{\CD}Z))\xrightarrow{\epsilon_{}}F(X)\otimes F(Y\otimes_{\CD}Z)\xrightarrow{\epsilon}F(X)\ot(F(Y)\ot F(Z))\xrightarrow{\alpha_{X,Y,Z}}(F(X)\ot F(Y))\ot F(Z)\xrightarrow{\epsilon^{-1}}F(X\otimes_{\CD}Y)\otimes F(Z)\xrightarrow{\epsilon^{-1}}F((X\ot_{\CD}Y)\ot_{\CD}Z)$ under $G$.
    By fullness, this image is unique.

    The pentagon of $\alpha^{\CD}$ in $\CD$ now can be proved by transferring it into $\CC$.
\end{proof}

We now specialize these observations to bimodule categories.
The next two lemmas show that the usual free--forgetful adjunction between
$\CC$ and $\hind{B_1}{\CC}{B_2}$ is compatible with the relevant monoidal
structures.

\begin{lem}\label{lem:free_functor_is_monoidal}
    The free functor
    \begin{align*}
        B_1\ot-\ot B_2:\CC&\to \hind{B_1}{\CC}{B_2}\\
        X&\mapsto (B_1\ot X\ot B_2,m_1\ot\id,\id\ot m_2)
    \end{align*}
    is strong monoidal.
\end{lem}
\begin{proof}
    We factor the free functor as
    \begin{align*}
        \CC
        \xrightarrow{F_2=-\ot B_2}
        \hind{}{\CC}{B_2}
        \xrightarrow{F_1=B_1\ot-}
        \hind{B_1}{\CC}{B_2},
    \end{align*}
    where
    \begin{align*}
        F_2(X)&=(X\ot B_2,\id_X\ot m_2),\\
        F_1(M,d_M)&=(B_1\ot M,m_1\ot\id_M,
        \id_{B_1}\ot d_M).
    \end{align*}
    Thus $F=F_1F_2$ is canonically the functor in the statement.  We
    prove that both factors are strong monoidal.  We omit canonical
    associators and unitors in the formulas below.

    For right $B_2$-modules $M,N$, denote the quotient map by
    \begin{align*}
        p^2_{M,N}:M\ot N\longrightarrow M\otd_{B_2}N,
    \end{align*}
    where the left $B_2$-action on $N$ is
    $d_N\circ\beta_{B_2,N}$.  For $X,Y\in\CC$, consider
    \begin{align*}
        \overline\alpha_{X,Y}:
        X\ot B_2\ot Y\ot B_2
        \xrightarrow{\id_X\ot\beta_{B_2,Y}\ot\id_{B_2}}
        X\ot Y\ot B_2\ot B_2\xrightarrow{\id_{X\ot Y}\ot m_2}
        X\ot Y\ot B_2.
    \end{align*}
    Associativity and braided commutativity of $m_2$, together with the
    hexagon axiom, give the balancing identity
    \begin{align*}
        \overline\alpha_{X,Y}
        \circ\bigl((\id_X\ot m_2)\ot\id_{Y\ot B_2}\bigr)
        =
        \overline\alpha_{X,Y}
        \circ\bigl(\id_{X\ot B_2}\ot
        ((\id_Y\ot m_2)\circ\beta_{B_2,Y\ot B_2})\bigr).
    \end{align*}
    Hence $\overline\alpha_{X,Y}$ induces a natural morphism
    \begin{align*}
        \alpha_{X,Y}:
        F_2(X)\otd_{B_2}F_2(Y)
        \longrightarrow F_2(X\ot Y).
    \end{align*}
    Its inverse is
    \begin{align*}
        F_2(X\ot Y)=X\ot Y\ot B_2
        \xrightarrow{\id_X\ot h_2\ot\id_{Y\ot B_2}}
        X\ot B_2\ot Y\ot B_2\xrightarrow{p^2_{F_2(X),F_2(Y)}}
        F_2(X)\otc{B_2}F_2(Y).
    \end{align*}
    Indeed, one composite is the identity by the unit axiom of $B_2$;
    after precomposition with $p^2_{F_2(X),F_2(Y)}$, the other is the
    identity by the defining $B_2$-balancing relation.  Since a
    coequalizer map is an epimorphism, $\alpha_{X,Y}$ is an isomorphism.

    We next prove the relative one-sided statement for $F_1$.  For
    $M,N\in\hind{}{\CC}{B_2}$, define
    \begin{align*}
        \overline\gamma_{M,N}:
        B_1\ot M\ot B_1\ot N
        \xrightarrow{\id_{B_1}\ot\beta_{M,B_1}\ot\id_N}
        B_1\ot B_1\ot M\ot N
        \xrightarrow{m_1\ot\id_{M\ot N}}
        B_1\ot M\ot N
        \xrightarrow{\id_{B_1}\ot p^2_{M,N}}
        B_1\ot(M\otc{B_2}N).
    \end{align*}
    Let $l^{12}$ and $r^{12}$ be the induced $K$-actions used in the
    definition of $\otd_K$.  Expanding those actions, associativity and
    braided commutativity of $m_1$ identify their $B_1$-parts, while
    the defining relation
    \begin{align*}
        p^2_{M,N}\circ(d_M\ot\id_N)
        =p^2_{M,N}\circ
        \bigl(\id_M\ot(d_N\circ\beta_{B_2,N})\bigr)
    \end{align*}
    identifies their $B_2$-parts.  Consequently,
    \begin{align*}
        \overline\gamma_{M,N}\circ
        (r^{12}_{F_1(M)}\ot\id_{F_1(N)})
        =
        \overline\gamma_{M,N}\circ
        (\id_{F_1(M)}\ot l^{12}_{F_1(N)}).
    \end{align*}
    Thus $\overline\gamma_{M,N}$ descends uniquely to a morphism in
    $\hind{B_1}{\CC}{B_2}$,
    \begin{align*}
        \gamma_{M,N}:
        F_1(M)\otd_KF_1(N)
        \longrightarrow F_1(M\otd_{B_2}N).
    \end{align*}
    Its inverse is the unique morphism $\delta_{M,N}$ satisfying
    \begin{align*}
        \delta_{M,N}\circ(\id_{B_1}\ot p^2_{M,N})
        =q^K_{F_1(M),F_1(N)}\circ
        (\id_{B_1\ot M}\ot h_1\ot\id_N),
    \end{align*}
    where $q^K$ is the quotient map defining $\otd_K$.  The right-hand
    side is $B_2$-balanced because $q^K$ balances the algebra map
    $h_1\ot\id_{B_2}:B_2\to K$.  The unit axiom of $B_1$ and the
    $K$-balancing relation give
    \begin{align*}
        \gamma_{M,N}\delta_{M,N}
        \circ(\id_{B_1}\ot p^2_{M,N})
        &=\id_{B_1}\ot p^2_{M,N},\\
        \delta_{M,N}\gamma_{M,N}\circ q^K_{F_1(M),F_1(N)}
        &=q^K_{F_1(M),F_1(N)}.
    \end{align*}
    Both quotient maps are epimorphisms, so $\delta_{M,N}$ and
    $\gamma_{M,N}$ are mutually inverse.

    The tensor constraint of $F=F_1F_2$ is the composite
    \begin{align*}
        F(X)\otd_{B_1\ot B_2}F(Y)
        \xrightarrow{\gamma_{F_2(X),F_2(Y)}}
        B_1\ot\bigl(F_2(X)\otd_{B_2}F_2(Y)\bigr)
        \xrightarrow{\id_{B_1}\ot\alpha_{X,Y}}
        B_1\ot X\ot Y\ot B_2
        =F(X\ot Y),
    \end{align*}
    and the unit constraint is the canonical isomorphism
    \begin{align*}
        B_1\ot B_2
        \xrightarrow{\sim}
        B_1\ot\bfone\ot B_2=F(\bfone).
    \end{align*}
    It remains to record coherence.  After precomposition with the
    relevant quotient maps, the two sides of the associativity diagram
    for $F_2$ are the same morphism obtained by moving three
    $B_2$-factors together and multiplying them.  Similarly, the two
    sides for $F_1$ move three $B_1$-factors together, multiply them,
    and then pass from $M\ot N\ot P$ to the iterated relative tensor
    product over $B_2$.  The two equalities follow from the $E_2$
    algebra coherences and the relative-tensor associator, respectively.
    The quotient maps are epimorphisms, so the associativity diagrams
    commute.  The unit diagrams reduce in the same way to the unit
    axioms of $B_1$ and $B_2$.  Therefore $F$ is strong monoidal.
\end{proof}

Having established the monoidal behavior of the free functor, we now
identify its right adjoint. This will allow the lax monoidal structure on
the forgetful functor to be obtained formally from the adjunction.
\begin{lem}\label{lem:free_forget_adjunction}
    The forgetful functor 
    \begin{align*}
        U:\hind{B_1}{\CC}{B_2}&\to \CC\\
        (M,l,r)&\mapsto M
    \end{align*}
    is the right adjoint of the free functor $B_1\ot -\ot  B_2$.
\end{lem}
\begin{proof}
    Let
    \begin{align*}
        U_1:\hind{B_1}{\CC}{B_2}\longrightarrow
        \hind{}{\CC}{B_2}
        \qquad\text{and}\qquad
        U_2:\hind{}{\CC}{B_2}\longrightarrow\CC
    \end{align*}
    forget the left $B_1$-action and the right $B_2$-action,
    respectively.  Then $U=U_2U_1$.  With the notation of the previous
    proof, the usual free--forgetful adjunctions give
    \begin{align*}
        F_1=B_1\ot-\dashv U_1,
        \qquad
        F_2=-\ot B_2\dashv U_2.
    \end{align*}
    Hence their composite satisfies $F_1F_2\dashv U_2U_1$.  We spell
    out the resulting adjunction in order to fix all structure maps.

    Let $X\in\CC$ and let
    $(M,u_M,d_M)\in\hind{B_1}{\CC}{B_2}$.  Define
    \begin{align*}
        \Phi_{X,M}:\Hom_{\hind{B_1}{\CC}{B_2}}
        (B_1\ot X\ot B_2,M)
        &\longrightarrow \Hom_{\CC}(X,M),\\
        g&\longmapsto
        g\circ(h_1\ot\id_X\ot h_2).
    \end{align*}
    Conversely, for $f:X\to M$, set
    \begin{align*}
        \Psi_{X,M}(f):
        B_1\ot X\ot B_2
        \xrightarrow{\id_{B_1}\ot f\ot\id_{B_2}}
        B_1\ot M\ot B_2
        \xrightarrow{u_M\ot\id_{B_2}}
        M\ot B_2
        \xrightarrow{d_M}M.
    \end{align*}
    Associativity of $u_M$ and $d_M$ shows that $\Psi_{X,M}(f)$ is
    left $B_1$-linear and right $B_2$-linear; the mixed
    compatibility of the two actions identifies the two possible orders
    in which those linearity conditions are applied.  Thus
    $\Psi_{X,M}(f)$ is a morphism in
    $\hind{B_1}{\CC}{B_2}$.

    The unit axioms of the two actions immediately give
    \begin{align*}
        \Phi_{X,M}\Psi_{X,M}(f)=f.
    \end{align*}
    If $g:B_1\ot X\ot B_2\to M$ is a bimodule morphism, then its
    left and right linearity say
    \begin{align*}
        u_M\circ(\id_{B_1}\ot g)
        &=g\circ(m_1\ot\id_{X\ot B_2}),\\
        d_M\circ(g\ot\id_{B_2})
        &=g\circ(\id_{B_1\ot X}\ot m_2).
    \end{align*}
    Substituting these identities into $\Psi_{X,M}\Phi_{X,M}(g)$ and
    using the unit axioms for $m_1$ and $m_2$ yields
    \begin{align*}
        \Psi_{X,M}\Phi_{X,M}(g)=g.
    \end{align*}
    Therefore $\Phi_{X,M}$ and $\Psi_{X,M}$ are mutually inverse and
    natural in $X$ and $M$.

    The unit and counit of the adjunction are, explicitly,
    \begin{align*}
        \eta_X&=h_1\ot\id_X\ot h_2:
        X\longrightarrow B_1\ot X\ot B_2,\\
        \epsilon_M&=d_M\circ(u_M\ot\id_{B_2}):
        B_1\ot M\ot B_2\longrightarrow M.
    \end{align*}
    Their triangle identities are the unit identities for the regular
    actions on the free bimodule and for the actions $u_M,d_M$.
    Consequently, $B_1\ot-\ot B_2\dashv U$.
\end{proof}

\begin{thm}[Doctrinal adjunction, \cite{Kel74}]\label{thm:doctrinal_adjuncion}
    Let $F \dashv U$ be an adjunction. 
    If $F$ is strongly monoidal, then $U$ is lax monoidal.
\end{thm}

Applying this result to the free--forgetful adjunction gives the monoidal
structure on the forgetful functor that will be used below.
\begin{lem}\label{lem:lax_monoidal_structure_on_forget_of_B_1CB_2}
    Let $B_1$ and $B_2\in \Alg_{E_2}(\CC)$.
    Then the forgetful functor
    \begin{align*}
        U:\hind{B_1}{\CC}{B_2}&\to \CC\\
        (M,l,r)&\mapsto M
    \end{align*}
    admits a canonical lax monoidal structure
    \begin{align*}
        h_1\ot h_2: 1\to B_1\ot B_2\\
        \nabla_{M,N}:=M\ot N\xrightarrow{p_{M,N}} M\otd\limits_{B_1\ot B_2} N
    \end{align*}
\end{lem}
\begin{proof}
    By Lemma \ref{lem:free_forget_adjunction}, we have $B_1\otimes -\otimes B_2$ is left adjoint of $U$.
    Then by Lemma \ref{lem:free_functor_is_monoidal} and Theorem \ref{thm:doctrinal_adjuncion}, we are done.
\end{proof}

\section{Module Categories}\label{sec:appendix-module-categories}

This appendix collects the module-category constructions used in the main text.
We first recall the basic notions of module categories, module functors, and
module natural transformations, which together form the 2-category
$\LMod(\CC)$. We then describe algebra and bimodule objects in
$\LMod(\CC)$. Finally, we introduce a folded description of bimodule
categories in terms of the relative tensor product
$\CA_1\btc{\CC}\CA_2^{\mathrm{rev}}$.

\begingroup

\let\EM\CM
\let\EN\CN
\let\one\bone
\let\ot\otimes

\newcommand{\Obj}{\operatorname{Obj}}

\subsection{Module Categories}

Let $(\CC,\otimes,\one,\alpha)$ be an $E_1$-algebra in the 2-category
$\Cat$, i.e. a monoidal category.

\begin{dfn}
    A {\bf left $\CC$-module category} consists of a tuple
    $(\EM,\odot,\delta,\lambda^\EM)$, where
    \begin{itemize}
        \item $\EM$ is a category;

        \item $\odot:\CC\times\EM\to\EM$ is a bifunctor;

        \item
        \[
            \delta:
            \odot\circ(\ot\times\Id_\EM)
            \Rightarrow
            \odot\circ(\Id_\CC\times\odot)
        \]
        is an invertible natural transformation, called the
        {\bf module associator},
        \begin{align*}
            \xymatrix{
                \CC\times\CC\times\EM
                \ar[r]^{\Id_{\CC}\times\odot}
                \ar[d]_{\ot\times\Id_{\EM}}
                \drtwocell\omit{\delta^{-1}}
                &
                \CC\times\EM
                \ar[d]^{\odot}
                \\
                \CC\times\EM
                \ar[r]_{\odot}
                &
                \EM;
            }
        \end{align*}

        \item
        \[
            \lambda^\EM:
            \odot\circ(I\times\Id_\EM)
            \Rightarrow
            \Id_\EM
        \]
        is an invertible natural transformation expressing the unit
        constraint:
        \begin{align*}
            \xymatrix{
                \ast\times\EM
                \ar[r]^{I\times\Id_{\EM}}
                \ar[dr]
                \drtwocell\omit{\lambda^{\EM}}
                &
                \CC\times\EM
                \ar[d]_{\odot}
                \\
                &
                \EM.
            }
        \end{align*}
    \end{itemize}

    These data satisfy the usual associativity and unit coherence conditions.

    \begin{enumerate}
        \item The associativity coherence is expressed by the commutativity
        of the following diagram:
        
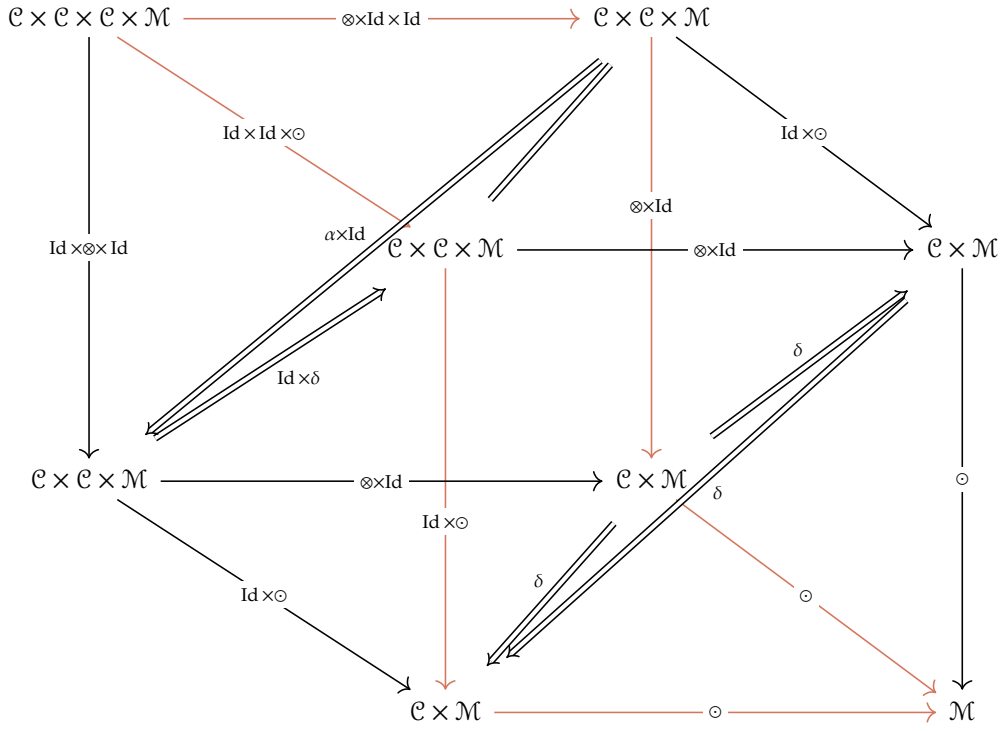
\[
        \begin{tikzcd}[row sep=3.6em]
        {\CC\times \CC\times \CC\times \EM} &&&& {\CC\times \CC\times \EM} \\
        \\
        &&& {\CC\times \CC\times \EM} &&&& {\CC\times \EM} \\
        \\
        {\CC\times \CC\times \EM} &&&& {\CC\times \EM} \\
        \\
        &&& {\CC\times \EM} &&&& \EM \\
        &&&&&&&& {}
        \arrow["{\otimes\times \Id \times \Id}"{description}, draw={rgb,255:red,213;green,117;blue,93}, from=1-1, to=1-5]
        \arrow["{\Id\times \Id\times\odot}"{description}, draw={rgb,255:red,213;green,117;blue,93}, from=1-1, to=3-4]
        \arrow["{\Id\times\otimes\times \Id}"{description}, from=1-1, to=5-1]
        \arrow["{\Id\times\odot}"{description}, from=1-5, to=3-8]
        \arrow["{\otimes\times \Id}"{description, pos=0.4}, draw={rgb,255:red,213;green,117;blue,93}, from=1-5, to=5-5]
        \arrow["{\otimes\times \Id}"{description}, from=3-4, to=3-8]
        \arrow["{\Id\times\odot}"{description, pos=0.6}, draw={rgb,255:red,213;green,117;blue,93}, from=3-4, to=7-4]
        \arrow["\odot"{description}, from=3-8, to=7-8]
        \arrow["{\otimes\times \Id}"{description}, from=5-1, to=5-5]
        \arrow["{\Id\times\odot}"{description}, from=5-1, to=7-4]
        \arrow["\odot"{description}, draw={rgb,255:red,213;green,117;blue,93}, from=5-5, to=7-8]
        \arrow["\odot"{description}, draw={rgb,255:red,213;green,117;blue,93}, from=7-4, to=7-8]
        \arrow[equals,from=1-5, to=3-4,shorten <=14pt, shorten >=14pt]
        \arrow["\delta",Rightarrow, from=5-5, to=3-8,shorten <=14pt, shorten >=14pt]
        \arrow["\delta", Rightarrow, swap,from=5-5, to=7-4,shorten <=12pt, shorten >=12pt]
        \arrow["\delta", Rightarrow, from=3-8, to=7-4,shorten <=18pt, shorten >=18pt]
        \arrow["{\Id\times\delta}",swap, Rightarrow, from=5-1, to=3-4,shorten <=14pt, shorten >=14pt]
        \arrow["{\alpha\times \Id}",swap, Rightarrow, from=1-5, to=5-1,shorten <=14pt, shorten >=14pt]
        \end{tikzcd}
        \]

        \item The unit coherence is expressed by the commutativity of
        
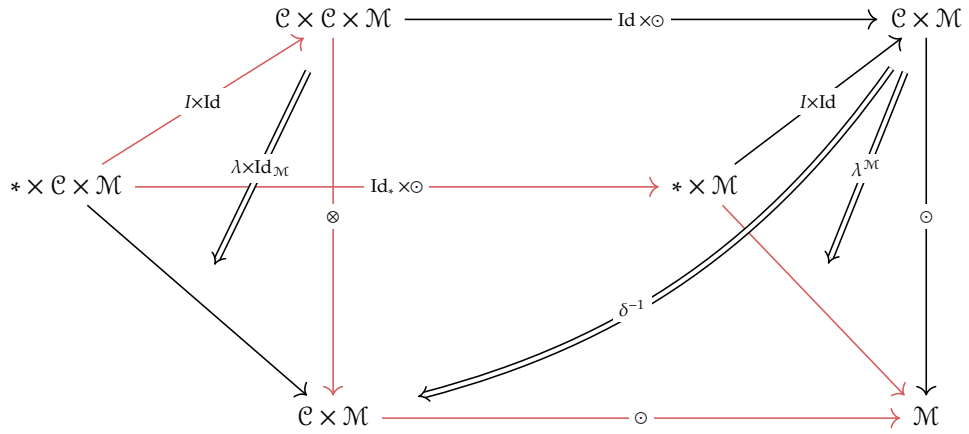
\[
        \begin{tikzcd}[row sep=2.4em]
            && {\CC\times \CC \times \EM} &&&&&& {\CC\times \EM} \\
            \\
            {\ast\times \CC\times \EM} &&&&&& {\ast\times \EM} \\
            \\
            \\
            && {\CC\times\EM} &&&&&& \EM
            \arrow["{\Id\times \odot}"{description}, from=1-3, to=1-9]
            \arrow["{\otimes}"{description},draw={rgb,255:red,214;green,92;blue,92}, from=1-3, to=6-3]
            \arrow["{\odot}"{description}, from=1-9, to=6-9]
            \arrow["{I\times \Id}"{description}, draw={rgb,255:red,214;green,92;blue,92}, from=3-1, to=1-3]
            \arrow["{\Id_{\ast}\times \odot}"{description}, draw={rgb,255:red,214;green,92;blue,92}, from=3-1, to=3-7]
            \arrow[""{name=0, anchor=center, inner sep=0}, from=3-1, to=6-3]
            \arrow["{I\times \Id}"{description}, from=3-7, to=1-9]
            \arrow[""{name=1, anchor=center, inner sep=0},draw={rgb,255:red,214;green,92;blue,92},from=3-7, to=6-9]
            \arrow["\odot"{description}, draw={rgb,255:red,214;green,92;blue,92}, from=6-3, to=6-9]
            \arrow["{\lambda\times\Id_{\EM}}"{description}, Rightarrow, from=1-3, to=0,shorten <=14pt, shorten >=14pt]
            \arrow["{\lambda^\EM}"{description}, Rightarrow, from=1-9, to=1,shorten <=14pt, shorten >=14pt]
            \arrow["\delta^{-1}"{description, pos=0.6}, bend left=20, Rightarrow, from=1-9, to=6-3,shorten <=14pt, shorten >=14pt]
        \end{tikzcd}
        \]
    \end{enumerate}
\end{dfn}

For later reference, we record the same coherence conditions in components.

\begin{prp}
    For $x,y\in\Obj(\CC)$ and $m\in\Obj(\EM)$, the module associator
    has components
    \[
        \delta_{x,y,m}:(x\otimes y)\odot m
        \longrightarrow
        x\odot(y\odot m),
    \]
    and the unit constraint has components
    \[
        \lambda^\EM_m:\one\odot m\longrightarrow m.
    \]

    For any $x,y,z\in\CC$ and $m\in\EM$, the following pentagon commutes:
    \begin{align*}
        \xymatrix{
            & ((x \otimes y) \otimes z) \odot m
            \ar[dl]_{\alpha_{x,y,z} \odot \id_m}
            \ar[dr]^{\delta_{x \otimes y,z,m}}
            \\
            (x \otimes (y \otimes z)) \odot m
            \ar[d]^{\delta_{x,y \otimes z,m}}
            &&
            (x \otimes y) \odot (z \odot m)
            \ar[d]^{\delta_{x,y,z \odot m}}
            \\
            x \odot ((y \otimes z) \odot m)
            \ar[rr]^{\id_x \odot \delta_{y,z,m}}
            &&
            x \odot (y \odot (z \odot m)).
        }
    \end{align*}

    For any $x\in\CC$ and $m\in\EM$, the following triangle commutes:
    \begin{align*}
        \xymatrix{
            (\one\otimes x)\odot m
            \ar[r]^{\delta_{\one,x,m}}
            \ar[dr]_{\lambda_x\odot \id_m}
            &
            \one\odot(x\odot m)
            \ar[d]^{\lambda^\EM_{x\odot m}}
            \\
            &
            x\odot m.
        }
    \end{align*}
\end{prp}

\subsection{Module Functors and Module Natural Transformations}

We next recall morphisms between module categories.

\begin{dfn}
    Let $(\EM,\odot_\EM,\delta^\EM,\lambda^\EM)$ and
    $(\EN,\odot_\EN,\delta^\EN,\lambda^\EN)$ be two left
    $\CC$-module categories.

    A {\bf left $\CC$-module functor}
    $(F,\zeta):\EM\to\EN$ consists of
    \begin{itemize}
        \item a functor $F:\EM\to\EN$;

        \item an invertible natural transformation
        \[
            \zeta:
            \odot_\EN\circ(\Id_\CC\times F)
            \Rightarrow
            F\circ\odot_\EM,
        \]
        represented by
        \begin{align*}
            \xymatrix{
                \CC\times\EM
                \ar[r]^{\Id_\CC\times F}
                \ar[d]_{\odot_\EM}
                \drtwocell\omit{\zeta}
                &
                \CC\times\EN
                \ar[d]^{\odot_\EN}
                \\
                \EM
                \ar[r]_{F}
                &
                \EN.
            }
        \end{align*}
    \end{itemize}

    These data satisfy the following associativity and unit coherence
    conditions.

    \begin{enumerate}
        \item
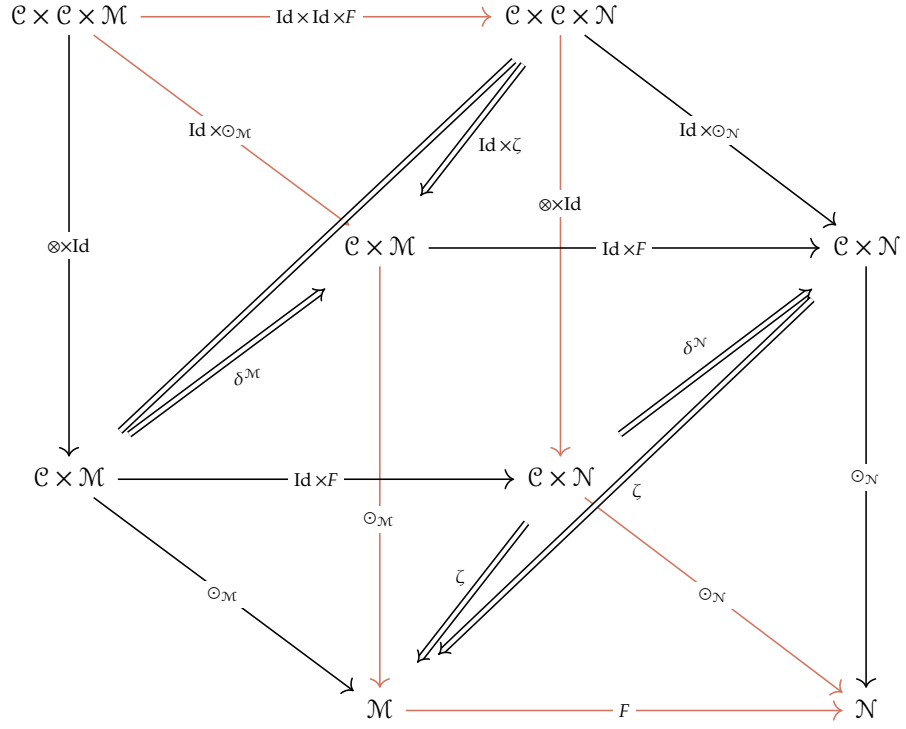
        \[
        \begin{tikzcd}[row sep=3.6em]
        {\CC\times \CC\times \EM} &&&& {\CC\times \CC\times \EN} \\
        \\
        &&& {\CC\times \EM} &&&& {\CC\times \EN} \\
        \\
        {\CC\times \EM} &&&& {\CC\times \EN} \\
        \\
        &&& {\EM} &&&& \EN \\
        &&&&&&&& {}
        \arrow["{\Id \times \Id \times F}"{description}, draw={rgb,255:red,213;green,117;blue,93}, from=1-1, to=1-5]
        \arrow["{\Id \times \odot_{\EM}}"{description}, draw={rgb,255:red,213;green,117;blue,93}, from=1-1, to=3-4]
        \arrow["{\otimes \times \Id}"{description}, from=1-1, to=5-1]
        \arrow["{\Id \times \odot_{\EN}}"{description}, from=1-5, to=3-8]
        \arrow["{\otimes \times \Id}"{description, pos=0.4}, draw={rgb,255:red,213;green,117;blue,93}, from=1-5, to=5-5]
        \arrow["{\Id\times F}"{description}, from=3-4, to=3-8]
        \arrow["{\odot_{\EM}}"{description, pos=0.6}, draw={rgb,255:red,213;green,117;blue,93}, from=3-4, to=7-4]
        \arrow["\odot_{\EN}"{description}, from=3-8, to=7-8]
        \arrow["{\Id\times F}"{description}, from=5-1, to=5-5]
        \arrow["{\odot_{\EM}}"{description}, from=5-1, to=7-4]
        \arrow["\odot_{\EN}"{description}, draw={rgb,255:red,213;green,117;blue,93}, from=5-5, to=7-8]
        \arrow["F"{description}, draw={rgb,255:red,213;green,117;blue,93}, from=7-4, to=7-8]
        \arrow["\Id\times\zeta",Rightarrow,from=1-5, to=3-4,shorten <=14pt, shorten >=14pt]
        \arrow["\delta^{\EN}",Rightarrow, from=5-5, to=3-8,shorten <=14pt, shorten >=14pt]
        \arrow["\zeta", Rightarrow, swap,from=5-5, to=7-4,shorten <=12pt, shorten >=12pt]
        \arrow["\zeta", Rightarrow, from=3-8, to=7-4,shorten <=18pt, shorten >=18pt]
        \arrow["{\delta^{\EM}}",swap, Rightarrow, from=5-1, to=3-4,shorten <=14pt, shorten >=14pt]
        \arrow[equals, from=1-5, to=5-1,shorten <=14pt, shorten >=14pt]
        \end{tikzcd}
        \]

        \item
        
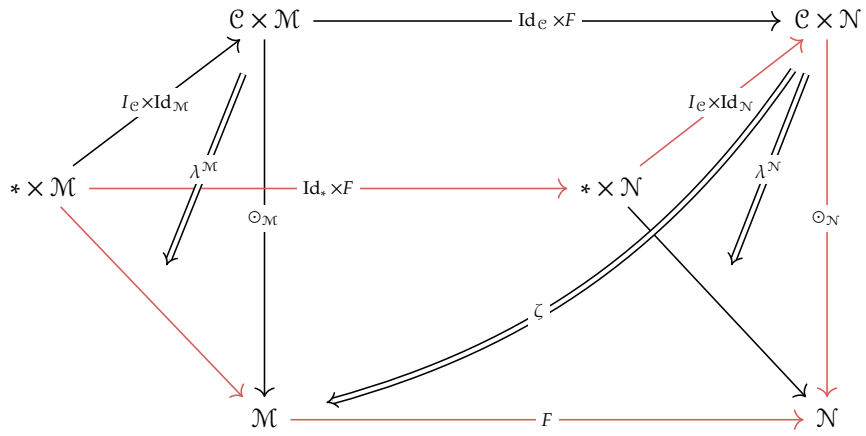
\[
        \begin{tikzcd}[row sep=2.4em]
            && {\CC\times \EM} &&&&&& {\CC\times \EN} \\
            \\
            {\ast\times \EM} &&&&&& {\ast\times \EN} \\
            \\
            \\
            && \EM &&&&&& \EN
            \arrow["{\Id_\CC\times F}"{description}, from=1-3, to=1-9]
            \arrow["{\odot_{\EM}}"{description}, from=1-3, to=6-3]
            \arrow["{\odot_{\EN}}"{description}, draw={rgb,255:red,214;green,92;blue,92}, from=1-9, to=6-9]
            \arrow["{I_\CC\times \Id_\EM}"{description}, from=3-1, to=1-3]
            \arrow["{\Id_{\ast}\times F}"{description}, draw={rgb,255:red,214;green,92;blue,92}, from=3-1, to=3-7]
            \arrow["{I_\CC\times \Id_\EN}"{description}, draw={rgb,255:red,214;green,92;blue,92}, from=3-7, to=1-9]
            \arrow["F"{description}, draw={rgb,255:red,214;green,92;blue,92}, from=6-3, to=6-9]
            \arrow[""{name=0, anchor=center, inner sep=0}, draw={rgb,255:red,214;green,92;blue,92}, from=3-1, to=6-3]
            \arrow[""{name=1, anchor=center, inner sep=0}, from=3-7, to=6-9]
            \arrow["{\lambda^\EM}"{description}, Rightarrow, from=1-3, to=0,shorten <=14pt, shorten >=14pt]
            \arrow["{\lambda^\EN}"{description}, Rightarrow, from=1-9, to=1,shorten <=14pt, shorten >=14pt]
            \arrow["\zeta"{description, pos=0.6}, bend left=20, Rightarrow, from=1-9, to=6-3,shorten <=14pt, shorten >=14pt]
        \end{tikzcd}
        \]
    \end{enumerate}
\end{dfn}

Equivalently, the structure map has components
\[
    \zeta_{x,m}:
    x\odot_\EN F(m)
    \longrightarrow
    F(x\odot_\EM m).
\]
The coherence conditions can be written componentwise as follows.

\begin{prp}
    For $x,y\in\CC$ and $m\in\EM$, the following diagram commutes:
    \begin{align*}
        \xymatrix{
            (x\otimes y)\odot_\EN F(m)
            \ar[r]_{\delta^\EN_{x,y,F(m)}}
            \ar[d]_{\zeta_{x\otimes y,m}}
            &
            x\odot_\EN(y\odot_\EN F(m))
            \ar[r]^{\id\odot_\EN\zeta_{y,m}}
            &
            x\odot_\EN F(y\odot_\EM m)
            \ar[dl]^{\zeta_{x,y\odot m}}
            \\
            F((x\otimes y)\odot_\EM m)
            \ar[r]^{F(\delta^\EM_{x,y,m})}
            &
            F(x\odot_\EM(y\odot_\EM m)).
        }
    \end{align*}

    The unit compatibility is expressed by
    \begin{align*}
        \xymatrix{
            \one\odot_\EN F(m)
            \ar[dr]_{\lambda^\EN_{F(m)}}
            \ar[rr]^{\zeta_{\one,m}}
            &&
            F(\one\odot_\EM m)
            \ar[dl]^{F(\lambda^\EM_m)}
            \\
            &
            F(m).
        }
    \end{align*}
\end{prp}

\begin{dfn}
    Let $(F,\zeta^F)$ and $(G,\zeta^G)$ be two left
    $\CC$-module functors from $\EM$ to $\EN$.
    A {\bf left $\CC$-module natural transformation}
    \[
        \gamma:(F,\zeta^F)\Rightarrow(G,\zeta^G)
    \]
    is a natural transformation $\gamma:F\Rightarrow G$ compatible
    with the module structures, as expressed by the following diagram:
    \[
    \begin{tikzcd}[row sep=2.4em]
    {\CC\times \EM} &&& {\CC\times \EN} \\
    \\
    \\
    \EM &&& \EN
    \arrow[""{name=0, anchor=center, inner sep=0}, "{\Id_\CC\times F}"{description}, draw={rgb,255:red,214;green,92;blue,92}, bend left=20, from=1-1, to=1-4]
    \arrow[""{name=1, anchor=center, inner sep=0}, "{\Id_\CC\times G}"{description}, bend right=20, from=1-1, to=1-4]
    \arrow["{\odot_\EM}"{description}, draw={rgb,255:red,214;green,92;blue,92}, from=1-1, to=4-1]
    \arrow["{\odot_{\EN}}"{description}, draw={rgb,255:red,214;green,92;blue,92}, from=1-4, to=4-4]
    \arrow[""{name=2, anchor=center, inner sep=0}, "F"{description}, bend left=20, from=4-1, to=4-4]
    \arrow[""{name=3, anchor=center, inner sep=0}, "G"{description}, draw={rgb,255:red,214;green,92;blue,92}, bend right=20, from=4-1, to=4-4]
    \arrow["{1\times\gamma}", Rightarrow, from=0, to=1,shorten <=3pt, shorten >=3pt]
    \arrow["\gamma", Rightarrow, from=2, to=3,shorten <=3pt, shorten >=3pt]
    \arrow["{\zeta^G}"{description}, bend left=15, Rightarrow, from=1-4, to=4-1,shorten <=10pt, shorten >=10pt]
    \arrow["{\zeta^F}"{description}, bend right=15, Rightarrow, from=1-4, to=4-1,shorten <=10pt, shorten >=10pt]
    \end{tikzcd}
    \]
\end{dfn}

In components, this condition takes the familiar form:

\begin{prp}
    A left $\CC$-module natural transformation satisfies
    \begin{align*}
        \xymatrix{
            x\odot_\EN F(m)
            \ar[r]^{\zeta^F_{x,m}}
            \ar[d]_{\id_x\odot\gamma_m}
            &
            F(x\odot_\EM m)
            \ar[d]^{\gamma_{x\odot m}}
            \\
            x\odot_\EN G(m)
            \ar[r]_{\zeta^G_{x,m}}
            &
            G(x\odot_\EM m).
        }
    \end{align*}
\end{prp}

\begin{dfn}
    The category $\Fun_\CC(\EM,\EN)$ consists of
    \begin{itemize}
        \item objects given by left $\CC$-module functors
        $\EM\to\EN$;
        \item morphisms given by left $\CC$-module natural transformations.
    \end{itemize}
\end{dfn}

\begin{dfn}
    The 2-category $\LMod(\CC)$ consists of
    \begin{itemize}
        \item left $\CC$-module categories as objects;
        \item the categories $\Fun_\CC(\EM,\EN)$ as hom categories.
    \end{itemize}
\end{dfn}

\endgroup

\subsection{Algebra and Bimodule Objects in $\LMod(\CC)$}

We now unpack the algebraic structures in $\LMod(\CC)$ that appear in
the module realization of 2-Morita theory.

\begin{dfn}
    An $E_1$-algebra in $\LMod(\CC)$ consists of
    \[
        \bigl(
        (\CA,L_\CA,\alpha^{L_\CA},\lambda^{L_\CA}),
        (H,\xi^H),
        (\otimes,\xi^\otimes),
        \alpha,\lambda,\rho
        \bigr),
    \]
    where
    \begin{itemize}
        \item
        $(\CA,L_\CA,\alpha^{L_\CA},\lambda^{L_\CA})$
        is an object of $\LMod(\CC)$;

        \item $(H,\xi^H)$ is a 1-morphism in $\LMod(\CC)$.
        In particular, $H:\CC\to\CA$ is a functor and $\xi^H$
        is the corresponding module natural isomorphism. Diagrammatically,
        \begin{align*}
            \xymatrix{
                &\ast
                \ar[dl]_{H_{\CC}}
                \ar[dr]^{H}
                \\
                \CC
                \ar[rr]
                &&
                \CA;
            }
        \end{align*}

        \item $(\otimes,\xi^\otimes)$ is a 1-morphism in
        $\LMod(\CC)$, with underlying functor
        \[
            \otimes:
            \CA\btc{\CC}\CA
            \longrightarrow
            \CA;
        \]

        \item $\alpha$, $\lambda$, and $\rho$ are invertible
        2-morphisms in $\LMod(\CC)$ giving the associator and unit
        constraints.
    \end{itemize}
    These data satisfy the usual monoidal coherence conditions.
\end{dfn}

Thus an $E_1$-algebra in $\LMod(\CC)$ may be regarded as a monoidal
$\CC$-module category whose multiplication and unit are themselves
$\CC$-module functors.

Let now
\[
    \CA_1,\CA_2\in\Alg_{E_1}(\LMod(\CC)).
\]

\begin{dfn}
    An object
    \[
        \CM\in
        \hind{\CA_1}{\LMod(\CC)}{\CA_2}
    \]
    consists of
    \[
        \bigl(
        (\CM,L,\alpha^\CM,\lambda^\CM),
        (\rhd,\xi^{\CA_1}),
        \alpha^{\CA_1},\lambda^{\CA_1},
        (\lhd,\xi^{\CA_2}),
        \alpha^{\CA_2},\rho^{\CA_2},
        \delta
        \bigr),
    \]
    where
    \begin{itemize}
        \item $(\CM,L,\alpha^\CM,\lambda^\CM)$ is an object of
        $\LMod(\CC)$;

        \item $(\rhd,\xi^{\CA_1})$ is a 1-morphism in
        $\LMod(\CC)$ whose underlying functor is
        \[
            \rhd:
            \CA_1\btc{\CC}\CM
            \longrightarrow
            \CM;
        \]

        \item $\alpha^{\CA_1}$ and $\lambda^{\CA_1}$ give the
        associativity and unit constraints for the left action;

        \item $(\lhd,\xi^{\CA_2})$ is a 1-morphism in
        $\LMod(\CC)$ whose underlying functor is
        \[
            \lhd:
            \CM\btc{\CC}\CA_2
            \longrightarrow
            \CM;
        \]

        \item $\alpha^{\CA_2}$ and $\rho^{\CA_2}$ give the
        associativity and unit constraints for the right action;

        \item
        \[
            \delta:
            \rhd\circ(\Id\btc{\CC}\lhd)
            \Rightarrow
            \lhd\circ(\rhd\btc{\CC}\Id)
        \]
        is an invertible 2-morphism in $\LMod(\CC)$ expressing the
        compatibility of the two actions.
    \end{itemize}
\end{dfn}

\begin{rmk}
    In the notation used in the main text, this is precisely the structure
    encoded by
    \[
        \hind{\CA_1}{\vind{\CC}{\Cat}{1}}{\CA_2}.
    \]
\end{rmk}

Morphisms between such bimodule objects must preserve both actions.

\begin{dfn}
    Let
    \[
        \CM,\CN\in
        \hind{\CA_1}{\LMod(\CC)}{\CA_2}.
    \]
    A 1-morphism from $\CM$ to $\CN$ consists of
    \[
        ((F,\xi^F),\delta^L,\delta^R),
    \]
    where
    \begin{itemize}
        \item $(F,\xi^F)$ is a 1-morphism in $\LMod(\CC)$, with
        underlying functor $F:\CM\to\CN$;

        \item $\delta^L$ is an invertible 2-morphism in
        $\LMod(\CC)$ comparing $F$ with the left $\CA_1$-actions:
        \begin{align*}
            \xymatrix{
                \CA_1\btc{\CC}\CM
                \ar[r]^{\Id\btc{\CC}F}
                \ar[d]_{\rhd^\CM}
                &
                \CA_1\btc{\CC}\CN
                \ar[d]^{\rhd^\CN}
                \\
                \CM
                \ar[r]_{F}
                &
                \CN;
            }
        \end{align*}

        \item $\delta^R$ is an invertible 2-morphism in
        $\LMod(\CC)$ comparing $F$ with the right $\CA_2$-actions:
        \begin{align*}
            \xymatrix{
                \CM\btc{\CC}\CA_2
                \ar[r]^{F\btc{\CC}\Id}
                \ar[d]_{\lhd^\CM}
                &
                \CN\btc{\CC}\CA_2
                \ar[d]^{\lhd^\CN}
                \\
                \CM
                \ar[r]_{F}
                &
                \CN.
            }
        \end{align*}
    \end{itemize}
    These data satisfy the usual compatibility conditions with the
    associativity, unit, and mixed bimodule constraints.
\end{dfn}

\subsection{The Folded Relative Tensor Product}

The preceding description involves two separate actions, by $\CA_1$
from the left and by $\CA_2$ from the right. We now combine them into
a single action. This is the categorical analogue of replacing an
$A_1$-$A_2$-bimodule by a left module over
$A_1\otimes A_2^{\mathrm{op}}$.

\begin{lem}
\label{lem:folded_relative_tensor_monoidal}
    Let
    $\CA_1,\CA_2\in\Alg_{E_1}(\LMod(\CC))$ and set
    \[
        \CD:=\CA_1\btc{\CC}\CA_2^{rev}.
    \]
    Then $\CD$ has a canonical monoidal structure characterized by
    \begin{align}
        (a\btd_{\CC}b)\otimes_{\CD}(a'\btd_{\CC}b')
        &:=
        (a\otimes_1a')\btd_{\CC}(b'\otimes_2b),
        \label{eq:folded_relative_tensor_product}\\
        \mathbf{1}_{\CD}
        &:=
        \mathbf{1}_{\CA_1}\btd_{\CC}\mathbf{1}_{\CA_2}.
        \notag
    \end{align}

    Moreover, the functors
    \begin{align*}
        \iota_1:\CA_1&\longrightarrow\CD,
        &a&\longmapsto a\btd_{\CC}\mathbf{1}_{\CA_2},\\
        \iota_2:\CA_2^{rev}&\longrightarrow\CD,
        &b&\longmapsto \mathbf{1}_{\CA_1}\btd_{\CC}b
    \end{align*}
    are strong monoidal.

    If $H_i:\CC\to\CA_i$ is the unit 1-morphism of $\CA_i$, then
    there is a strong monoidal functor
    \begin{align}
        j:=\iota_1H_1
        \simeq
        \iota_2H_2:
        \CC\longrightarrow\CD.
        \label{eq:folded_j_functor}
    \end{align}
\end{lem}

\begin{proof}
    We use $\odot$ for the left $\CC$-actions and also for the
    associated right actions obtained from the braiding of $\CC$.
    The algebra-object structure of $\CA_i$ supplies coherent
    centrality isomorphisms
    \begin{align}
        H_i(c)\otimes_i x
        \simeq
        c\odot_{\CA_i}x
        \simeq
        x\otimes_iH_i(c).
        \label{eq:centrality_of_algebra_object_in_CCat}
    \end{align}
    These are induced by the $\CC$-linearity and balancing of the
    multiplication
    $\CA_i\btd_{\CC}\CA_i\to\CA_i$, together with its two unit
    constraints.

    Let
    \[
        q:\CA_1\times\CA_2^{rev}\longrightarrow\CD,
        \qquad
        (a,b)\longmapsto a\btd_{\CC}b
    \]
    be the universal $\CC$-balanced right-exact bifunctor.
    Consider the four-variable right-exact functor
    \[
        \overline{\mu}(a,b,a',b')
        :=
        q(a\otimes_1a',b'\otimes_2b).
    \]

    Its balancing in the first pair is given by
    \begin{align*}
        \overline{\mu}(c\odot_{\CA_1}a,b,a',b')
        &\simeq
        q\bigl(
        c\odot_{\CA_1}(a\otimes_1a'),
        b'\otimes_2b
        \bigr)\\
        &\simeq
        q\bigl(
        a\otimes_1a',
        c\odot_{\CA_2}(b'\otimes_2b)
        \bigr)\\
        &\simeq
        q\bigl(
        a\otimes_1a',
        b'\otimes_2(c\odot_{\CA_2}b)
        \bigr)\\
        &=
        \overline{\mu}(a,c\odot_{\CA_2}b,a',b').
    \end{align*}

    Similarly, balancing in the second pair is
    \begin{align*}
        \overline{\mu}(a,b,c\odot_{\CA_1}a',b')
        &\simeq
        q\bigl(
        c\odot_{\CA_1}(a\otimes_1a'),
        b'\otimes_2b
        \bigr)\\
        &\simeq
        q\bigl(
        a\otimes_1a',
        c\odot_{\CA_2}(b'\otimes_2b)
        \bigr)\\
        &\simeq
        q\bigl(
        a\otimes_1a',
        (c\odot_{\CA_2}b')\otimes_2b
        \bigr)\\
        &=
        \overline{\mu}(a,b,a',c\odot_{\CA_2}b').
    \end{align*}

    Hence the universal property of the relative Kelly tensor product
    descends $\overline{\mu}$ uniquely to a right-exact functor
    \[
        \mu_{\CD}:\CD\btd\CD\to\CD,
    \]
    giving \eqref{eq:folded_relative_tensor_product}.

    After precomposition with $q^{\times3}$, the two bracketings are
    \[
        ((a\otimes_1a')\otimes_1a'')
        \btd_{\CC}
        \bigl(b''\otimes_2(b'\otimes_2b)\bigr)
    \]
    and
    \[
        (a\otimes_1(a'\otimes_1a''))
        \btd_{\CC}
        \bigl((b''\otimes_2b')\otimes_2b\bigr).
    \]
    The associators of $\CA_1$ and $\CA_2$ therefore induce the
    associator of $\CD$. The unitors are induced in the same way.
    The pentagon and triangle identities follow from those of
    $\CA_1$ and $\CA_2$ by the universal property.

    Formula \eqref{eq:folded_relative_tensor_product} immediately
    gives the strong monoidal structures on $\iota_1$ and $\iota_2$.
    Moreover,
    \[
        H_i(c)\otimes_iH_i(c')
        \simeq
        H_i(c\otimes_{\CC}c'),
        \qquad
        H_i(\mathbf{1}_{\CC})
        \simeq
        \mathbf{1}_{\CA_i},
    \]
    so $\iota_1H_1$ is strong monoidal.

    For the second presentation, regard
    $H_2:\CC\to\CA_2^{rev}$ as a strong monoidal functor via
    \begin{align*}
        H_2(c)\otimes_{\CA_2^{rev}}H_2(c')
        &=
        H_2(c')\otimes_2H_2(c)
        \simeq
        H_2(c'\otimes_{\CC}c)\\
        &\xrightarrow{H_2(\beta^\CC_{c',c})}
        H_2(c\otimes_{\CC}c').
    \end{align*}
    Hence $\iota_2H_2$ is strong monoidal as well.

    Finally, the balancing isomorphism of $q$ gives a monoidal natural
    isomorphism
    \[
        H_1(c)\btd_{\CC}\mathbf{1}_{\CA_2}
        \simeq
        \mathbf{1}_{\CA_1}\btd_{\CC}H_2(c),
    \]
    proving \eqref{eq:folded_j_functor}.
\end{proof}

\begin{rmk}
    In the finite linear setting, the corresponding monoidal structure
    on the relative tensor product of a monoidal right module and a
    monoidal left module over a braided multi-tensor category is recorded
    in \cite[Section~2.2]{KYZ21}; see also \cite{DSPS19} for the balanced
    right-exact universal property of the finite relative tensor product.

    Lemma \ref{lem:folded_relative_tensor_monoidal} is the finitely
    cocomplete version: it uses neither finiteness nor rigidity, but only
    the existence and universal property of the relative Kelly tensor
    product in $\Cat^{\mathrm{rex}}$.
\end{rmk}

The previous lemma folds the two algebra actions into a single monoidal
category. The next lemma shows that the same construction folds an
$\CA_1$-$\CA_2$-bimodule category into a left $\CD$-module category.

\begin{lem}
\label{lem:folded_relative_tensor_action}
    Let
    $\CA_1,\CA_2\in\Alg_{E_1}(\LMod(\CC))$ and set
    \[
        \CD:=\CA_1\btc{\CC}\CA_2^{rev},
    \]
    equipped with the monoidal structure of
    Lemma \ref{lem:folded_relative_tensor_monoidal}.
    For every
    \[
        \CM\in
        \hind{\CA_1}{\LMod(\CC)}{\CA_2},
    \]
    the formula
    \begin{align}
        (a\btc{\CC}b)\rhd_{\CD}^{\CM}m
        &:=
        (a\rhd^\CM m)\lhd^\CM b
        \label{eq:folded_relative_tensor_action}
    \end{align}
    defines a left $\CD$-module structure on $\CM$.
    Its restriction along the strong monoidal functor
    $j:\CC\to\CD$ of \eqref{eq:folded_j_functor} is canonically
    the original $\CC$-action on $\CM$.
\end{lem}

\begin{proof}
    Let
    \[
        q:\CA_1\times\CA_2^{rev}
        \longrightarrow
        \CD,
        \qquad
        (a,b)\longmapsto a\btc{\CC}b
    \]
    be the universal $\CC$-balanced right-exact bifunctor.
    Before passing to the relative tensor product, define
    \[
        \overline{\mu}_{\CM}(a,b,m)
        :=
        (a\rhd^\CM m)\lhd^\CM b.
    \]

    The balancing needed to descend this functor is the canonical
    composite
    \begin{align}
        ((c\odot_{\CA_1}a)\rhd^\CM m)\lhd^\CM b
        &\simeq
        (a\rhd^\CM(c\odot_\CM m))\lhd^\CM b\\
        &\simeq
        (a\rhd^\CM m)\lhd^\CM(c\odot_{\CA_2}b).
        \label{eq:folded_action_balancing}
    \end{align}
    Here we use that the two actions and their mixed constraint are
    1- and 2-morphisms in $\LMod(\CC)$.
    Thus \eqref{eq:folded_relative_tensor_action} is well-defined.

    Its module associator is induced from
    \begin{align*}
        a\rhd^\CM
        \bigl((a'\rhd^\CM m)\lhd^\CM b'\bigr)
        \lhd^\CM b
        &\simeq
        \bigl(a\rhd^\CM(a'\rhd^\CM m)\bigr)
        \lhd^\CM b'\lhd^\CM b\\
        &\simeq
        ((a\otimes_1a')\rhd^\CM m)
        \lhd^\CM(b'\otimes_2b).
    \end{align*}
    Its coherence follows from the left, right, and mixed bimodule
    coherences after precomposition with
    $q^{\times2}\times\id_\CM$.

    Finally, the two presentations of $j(c)$ give
    \begin{align*}
        j(c)\rhd_{\CD}^{\CM}m
        &\simeq
        (H_1(c)\rhd^\CM m)
        \lhd^\CM\mathbf{1}_{\CA_2}
        \simeq
        c\odot_\CM m,\\
        j(c)\rhd_{\CD}^{\CM}m
        &\simeq
        (\mathbf{1}_{\CA_1}\rhd^\CM m)
        \lhd^\CM H_2(c)
        \simeq
        c\odot_\CM m.
    \end{align*}
    These composites agree by
    \eqref{eq:folded_action_balancing} and the mixed bimodule unit axiom.
    Hence the restricted action is precisely the original $\CC$-action.
\end{proof}

The two lemmas above therefore provide the folded realization
\[
    \hind{\CA_1}{\LMod(\CC)}{\CA_2}
    \longrightarrow
    \LMod\!\left(
        \CA_1\btc{\CC}\CA_2^{\mathrm{rev}}
    \right),
\]
which sends an $\CA_1$-$\CA_2$-bimodule category to the same underlying
category equipped with the folded action
\[
    (a\btc{\CC}b)\rhd m
    :=
    (a\rhd m)\lhd b.
\]
This is the form of the construction used in the main text.

\newpage

\bibliographystyle{alpha_zzh}
\bibliography{Top}

\end{document}